\documentclass[reqno,centertags,11pt,a4paper]{amsart}
\usepackage{amsmath,amsthm,amscd,amssymb}
\usepackage{latexsym}
\usepackage{mathabx}  
\usepackage{bbm}
\usepackage{mathtools} 
\usepackage[colorlinks=true]{hyperref}

\usepackage{enumitem}   

\usepackage {verbatim}
\usepackage{verbatim}
\newcommand{\Hmm}[1]{\leavevmode{\marginpar{\tiny%
$\hbox to 0mm{\hspace*{-0.5mm}$\leftarrow$\hss}%
\vcenter{\vrule depth 0.1mm height 0.1mm width \the\marginparwidth}%
\hbox to 0mm{\hss$\rightarrow$\hspace*{-0.5mm}}$\\\relax\raggedright #1}}}

\newcommand{\N}{{\mathbb{N}}}

\newcommand{\R}{{\mathbb{R}}}

\newcommand{\C}{{\mathbb{C}}}

\newcommand{\Z}{{\mathbb{Z}}}

\newcommand{\ol}{\overline}

\newcommand{\wti}[1]{\widetilde{#1}}
\newcommand{\Oh}{O}

\newcommand{\ess}{\text{\rm{ess}}}

\newcommand{\hatt}{\widehat}
\newcommand{\beq}{\begin{equation}}
\newcommand{\eeq}{\end{equation}}
\newcommand{\bdm}{\begin{displaymath}}
\newcommand{\edm}{\end{displaymath}}
\newcommand{\ba}{\begin{align}}
\newcommand{\ea}{\end{align}}
\newcommand{\bpf}{\begin{proof}}
\newcommand{\epf}{\end{proof}}

\newcommand{\la}{\langle}
\newcommand{\ra}{\rangle}

\newcommand{\dist}{\operatorname{dist}}               

\newcommand{\e}{\mathrm{e}}
\renewcommand{\d}{\mathrm{d}}

\newcommand{\veps}{\varepsilon}

\newcommand{\re}{\mathrm{Re}}

\newcommand{\gpm}{g_{\scriptscriptstyle P,M}}

\newcommand{\id}{\mathbf{1}}                

\DeclareMathOperator{\diam}{diam}

\allowdisplaybreaks

\newcommand{\calC}{\mathcal{C}}

\newcommand{\calF}{\mathcal{F}}

\newcommand{\calO}{\mathcal{O}}
\newcommand{\calQ}{\mathcal{Q}}

\newcommand{\calU}{\mathcal{U}}

\newtheorem{theorem}{Theorem}

\newtheorem{lemma}[theorem]{Lemma}
\newtheorem{corollary}[theorem]{Corollary}

\theoremstyle{definition}

\newtheorem{definition}[theorem]{Definition}

\newtheorem{remark}[theorem]{Remark}
\newtheorem{remarks}[theorem]{Remarks}

\newcounter{theoremi}[theorem]

\newcommand{\itemthm}{\refstepcounter{theoremi} {\rm(\roman{theoremi})}{~}}

\numberwithin{theorem}{section}
\numberwithin{equation}{section}

\newenvironment{thmlist}{\begin{enumerate}[label=\roman*{\rm)},ref=\roman*, leftmargin=1.4em]}{\end{enumerate}}
\newcounter{smalllist}

\newcounter{smallenum}
\newenvironment{SE}{\begin{list}{{\rm\arabic{smallenum})}}{%
\setlength{\topsep}{0mm}\setlength{\parsep}{0mm}\setlength{\itemsep}{0mm}%
\setlength{\labelwidth}{2em}\setlength{\leftmargin}{2em}\usecounter{smallenum}%
}}{\end{list}}
\newcounter{smallalpha} 
\newenvironment{SLalpha}{\begin{list}{{\rm\alph{smallalpha})}}{%
\setlength{\topsep}{0mm}\setlength{\parsep}{0mm}\setlength{\itemsep}{0mm}%
\setlength{\labelwidth}{2em}\setlength{\leftmargin}{1.4em}\usecounter{smallalpha}%
}}{\end{list}} 
\newcounter{assumptions}

\begin{document}

\title[]{Quantitative bounds versus existence of weakly coupled bound states for Schr\"odinger type operators}
\author[V.~Hoang, D.~Hundertmark, J.~Richter, S.~Vugalter]{Vu Hoang, Dirk Hundertmark, Johanna Richter, Semjon Vugalter}
\address{Department of Mathematics, Rice University,  MS 136
6100 Main Street, Houston, TX 77005 (USA)}%
\email{vu.hoang@rice.edu}
\address{Institute for Analysis, Karlsruhe Institute of Technology, 76131 Karlsruhe, Germany.}%
\email{dirk.hundertmark@kit.edu}%
\address{Institute for Analysis, Dynamics, and Modelling, University Stuttgart,
Pfaffenwaldring 57, 70569 Stuttgart, Germany.}%
\email{johanna.richter@mathematik.uni-stuttgart.de}%
\address{Department of Mathematics, Institute for Analysis, Karlsruhe Institute of Technology, 76131 Karlsruhe, Germany.}%
\email{semjon.wugalter@kit.edu}%

\thanks{\copyright 2017 by the authors. Faithful reproduction of this article, in its entirety, by any means is permitted for non-commercial purposes}

\date{\today, version bound-states-6-4.tex}


\begin{abstract}
It is well-known that for usual Schr\"odinger operators weakly coupled bound states exist in dimensions one and two, whereas in higher dimensions the famous Cwikel-Lieb-Rozenblum bound holds. 
We show for a large class of Schr\"odinger-type operators with general kinetic energies that these two phenomena are complementary. In particular, we explicitly get a natural semi-classical type bound on the number of bound states precisely in the situation when weakly coupled bound states exist not. 
%
\end{abstract}

\maketitle
{\hypersetup{linkcolor=black}
\tableofcontents}

\section{Introduction: quantitative bounds versus weakly coupled bound states}\label{introduction}
In this paper, we study operators of the form
\begin{align*}
T(p) + V
\end{align*}
where $p=-i\nabla$ is the quantum-mechanical momentum operator and the kinetic energy symbol $T:\R^d\to [0,\infty)$ is a measurable function. See Section \ref{sec:main results} for the precise definition of $T(p)+V$ and the condition we need on the kinetic energy symbol $T$ and the interaction potential $V$. 
 There has been an enormous amount of interest in the study of bound states for such operators. 
 Usually, in standard quantum mechanics the symbol is given by $T(\eta)=\eta^2$, that is, the kinetic energy is given by the the Laplacian $p^2=-\Delta$. 
 In this case, it is well-known that quantum mechanics in three and more dimensions is quite different from one and two dimensions.  In three and more dimensions the perturbed operator $p^2+V$ can be \emph{unitarily equivalent} to the free operator $p^2$ for potentials $V$ which are small in some sense, \cite[Theorem XIII.27]{ReedSimon}, and for a suitable class of potentials the famous Cwikel--Lieb--Rozenblum bound (CLR) holds, which gives a bound on the number of bound states in terms of a semi--classical phase-space volume, see \cite{Cwikel,Lieb,Rozenblum}. On the other hand, in one and two dimensions, arbitrarily small attractive potentials produce a bound state, see Problem 2 on page 156 in \cite{Landau-Lifshitz}. More precisely, it was rigorously shown\footnote{For a textbook treatment of weakly coupled bound states for the Schr\"odinger operator in one and two dimensions, see page 679--682 in \cite{Simon-course2}. } by Barry Simon many years ago, see \cite{Simon}, that it suffices that $V$ is not identically zero and 
 $V \in L^1$ with $\int V\, \d x\le 0$, together with some mild moment condition on $V$, so that $-\Delta +\lambda V$ has a strictly negative bound state in one and two dimensions, for any $\lambda>0$. This had been generalized in \cite{NetrusovWeidl-1996} to higher order Schr\"odinger--type operators.

Recently renewed interest in weakly coupled  bound states arose due to the observation that such states can be found in many other physically interesting cases and they are responsible for different physical behavior of these systems compared to what one is used to from usual quantum mechanics in high dimensions.   
These systems include quantum wave guides, systems with homogenous or increasing magnetic fields, the Bardeen--Cooper--Schrieffer (BCS) model for superconductivity. These examples are not necessarily one or two-dimensional, but they are described by Schr\"odinger type operators with potentially \emph{strongly degenerate} kinetic energies   
 \cite{HainzlSeiringer-degenerate,FrankHainzlNabokoSeiringer,LaptevSafranovWeidl, Pankrashkin}, that is, the kinetic energy $T$ can be degenerate, $T(\eta)=0$,  not only at a single point but on a ``large" set in momentum space. For example, the kinetic energy could have a zero set which is an embedded hypersurface in  
 $\R^d$. At this point it is important to emphasize that the results of \cite{FrankHainzlNabokoSeiringer} concern the special BCS Hamiltonian in three dimensions where, in particular, the zero set of the kinetic energy is a two-dimensional sphere in $\R^3$. The works \cite{HainzlSeiringer-degenerate,LaptevSafranovWeidl, Pankrashkin} require that the kinetic energy $T$ is locally bounded, satisfying some growth conditions at infinity,  the zero set of the kinetic energy $T$ is a smooth co--dimension one submanifold of $\R^d$, in \cite{Pankrashkin} only locally, and that $\int V\, dx <0$ or, even stronger, $V\le 0$ and $V\neq0$. These last two conditions are stronger than the conditions on the potential in the original work of Simon. In particular, they leave open the question what happens if $V\neq0$ and $\int V\, dx =0$ or if the co-dimension of the zero set of the kinetic energy $T$ is larger than one. 
 
 Motivated by the above questions, we consider a \emph{very general class} of kinetic energies and potentials. More importantly, we have an \emph{essentially sharp} existence result for weakly coupled bound states: We give conditions under which 
 weakly coupled bound states exist for any non-trivial attractive potential, $\int V\, dx\le 0$, and if our conditions are not met, then for any strictly negative but sufficiently small potential weakly coupled bound states exist. 
 Moreover, in the second case, we prove a quantitative bound on the number of negative bound states.    
 
 We are able to do this by  \emph{identifying the mechanism} which is responsible for the creation of weakly coupled bound states:  Roughly speaking $\eta\mapsto T(\eta)^{-1}$ being integrable or not in a small neighborhood of the zero set of $T$ distinguishes between having a quantitative bound on the number of negative bound states in the first case or having weakly coupled bound states for arbitrarily small attractive potentials in the second, see Theorems \ref{thm:superduper1} and \ref{thm:superduper2}, for the precise conditions\footnote{and Theorem \ref{thm:both}, which shows that under some slight additional assumptions these are complementary.}.

Our first theorem concerns the existence of bound states. We write $V\neq 0$ if $V$ is not the zero function. 

\begin{theorem}[Weakly coupled bound states]\label{thm:superduper1}
	Let $T:\R^d\to [0,\infty)$ be measurable. 
	Assume that there exists a compact set $M\subset \R^d$ such that 
	\begin{equation}\label{eq:superduper1}
		\int_{M_\delta} T(\eta)^{-1}\, \d\eta = \infty \text{ for all } \delta>0\, ,
	\end{equation}
	where $M_\delta\coloneq\{\eta\in \R^d:\, \dist(\eta,M)\le \delta\}$. 
	Then $\inf\sigma(T(p))=0$ and if the potential $V\not=0$, is relatively form compact with respect to $T(p)$ and $V\le 0$, then $T(p)+V$ has at least one strictly negative bound state. 
	
	Moreover, if $T$ is locally bounded this also holds for \emph{sign indefinite} potentials in the sense that if $V\in L^1(\R^d)$, $V\neq 0$, is relatively form compact with respect to $T(p)$ and $\int V\,\d x\le 0$, then the 
	operator $T(p)+V$  has again at least one strictly negative eigenvalue. 
\end{theorem}

\begin{remarks}\label{rmk:superduper1}
\begin{thmlist}
	\item We would like to stress the fact that our theorem poses rather weak conditions on the zero set of the kinetic energy symbol $T$ and it does not have to satisfy any growth conditions at infinity.  Of course, if $T$ does not satisfy a growth condition at infinity, then, even if $V\in L^1$, it does not have to be relatively form compact with respect to $T(p)$. Without this compactness, one can formulate a version of Theorem \ref{thm:superduper1} which still guarantees the existence of  some negative spectrum, but not necessarily discrete.   
	\item 	We would also like to stress that for the existence of a strictly negative eigenvalue we can allow sign changing potentials with $\int V\, dx =0$ as long as $V$ is not identically zero and $T$ is locally bounded, which holds in all physically interesting cases.  Our method in the proof of Theorem \ref{thm:superduper1} relies on a very simple and natural variational calculation, which also works in the \emph{critical case} where $\int V\, \d x=0$.  It is a \emph{purely local} method and its main advantage is its excruciating simplicity and its wide range of applications.
	 If the potential $V$ does not change sign, there can even be infinitely many bound states below the essential spectrum, see Theorem \ref{thm:many} and Corollaries \ref{cor:isolated-points} and \ref{cor:submanifold} below.
	\item Previously, the weakest condition on the potential which guaranteed existence of at least one strictly negative eigenvalue is due to Pankrashkin \cite{Pankrashkin} who showed  that  a strictly negative eigenvalue exists if $V\in  L^1 $ and $\int V\, dx <0$. This condition is weaker than the condition of \cite{FrankHainzlNabokoSeiringer, HainzlSeiringer-degenerate} where the authors have to assume that the Fourier transform 
 	   $\hatt{V}$ is non-vanishing within a large enough ball centered at the origin\footnote{More precisely, in \cite{HainzlSeiringer-degenerate} they require that $\hatt{V}$ is non-vanishing on the set $S-S=\{\eta_1-\eta_2: \eta_j\in S\}$, where $S$ is the zero set of $T$. }. 
	    In \cite{FrankHainzlNabokoSeiringer,HainzlSeiringer-degenerate,LaptevSafranovWeidl}, they adapt the method of Simon, using the Birman-Schwinger principle, to identify a singular piece of the Birman-Schwinger operator. This approach needs \emph{global assumptions} on the zero set of the kinetic energy. The work \cite{Pankrashkin} uses a construction of appropriate trial functions, as such the assumptions on the zero set of the kinetic energy in \cite{Pankrashkin} are local. More precisely, there is an open set such that locally within this set the zero set of $T$ is a smooth submanifold of $\R^d$. 
		The work  \cite{LaptevSafranovWeidl} establishes the precise asymptotic rate of the negative eigenvalues, but for this they need strong assumptions.
	\item It is easy to construct examples of potentials $V\in L^1$ such that its Fourier transform $\hatt{V}$ is zero on a large centered ball. Simply take any spherically symmetric Schwartz function $\hatt{V}$ which is supported on a large enough centered annulus in Fourier space and let $V$ be the inverse Fourier transform of $\hatt{V}$.  In this case $\int V\, dx =\hatt{V}(0)=0$ and our Theorem \ref{thm:superduper1} shows that there exists at least one strictly negative eigenvalue under suitable conditions on the kinetic energy $T$, whereas the previous results in \cite{Pankrashkin, HainzlSeiringer-degenerate,LaptevSafranovWeidl} are not applicable. The only exception is the pioneering work of Simon, which for the Laplace operator $T(p)=p^2$ in one and two dimensions, shows that there are strictly negative eigenvalues if $\int V\, dx\le 0$, $V$ does not vanish identically, and, for some technical reasons, some high enough moments of $V$ are finite.   	
	\item If the kinetic energy symbol $T$ is continuous, then the compact set $M$ above can be chosen to be a subset of the zero set of $T$. In this case, the behaviour of $T$ near its zero set determines whether \eqref{eq:superduper1} holds or not. 
   \item The potential $V$ is relatively form compact if 
    		$ (T(p)+1)^{-1/2} |V| (T(p)+1)^{-1/2}$ is a compact operator on $L^2(\R^d)$, see  \cite[Section 6.3]{Teschl} or 
    		\cite[Section 7.5]{Simon-course2}. In this case, $V$ is automatically relatively 
    		form small with respect to $T(p)$ with relative bound zero,
    		\cite[Lemma 6.26]{Teschl}, so by the KLMN theorem, \cite[Theorem 6.24]{Teschl} or \cite[Theorem 7.5.7]{Simon-course2}, the sum $T(p) +V$ is well defined as a sum of quadratic forms on the from domain of $T(p)$   and by relative form compactness 
    		$\sigma_\ess(T(p)+V)=\sigma_\ess(T(p))$, \cite[Lemma 6.26]{Teschl} or \cite[Theorem 7.8.4]{Simon-course2}, in particular, 
    		$ \inf\ess(T(p)+V) = 0 $. 
      The assumptions $\calQ(T(p)^{1-\varepsilon})\subset\calQ(V)$, for some $ 0 < \varepsilon < 1, $ $V\in L^1$, and $\lim_{\eta\to\infty}T(\eta)=\infty$ already imply that $V$ is relatively form compact with respect to $T(p)$, see Lemma \ref{lem:rel-compactness1} in the appendix. A different relative compactness criterium, which still needs that $T$ diverges at infinity, is discussed in Lemma \ref{lem:rel-compactness2}. We give the simple arguments for both in the Appendix. 
\end{thmlist}
\end{remarks} 
The situation discussed in Theorem \ref{thm:superduper1} changes drastically when $\eta\mapsto T(\eta)^{-1}$ is integrable near the zero set of $T$, more precisely,  when 
$T^{-1}\id_{\{T<\delta \}}\in L^1(\R^d)$ for some $\delta>0$. Not only do weakly coupled bound states cease to exist but we have even a quantitative bound on the number of negative eigenvalues in this case. More precisely, introduce the function $G:[0,\infty]\to [0,\infty]$ by 
	\begin{equation}\label{eq:G}
		G(u)\coloneq u\int_{T(\eta)<u}T(\eta)^{-1}\, \frac{\d\eta}{(2\pi)^d} 
	\end{equation}
	for $u\ge 0$. It is clear from the definition that $G(u)<\infty$ if and only if $\int_{T<u} T(\eta)^{-1}\, \d\eta<\infty$ and if $G(u_0)<\infty$ for some $u_0>0$, then $G(u)<\infty$ for all $0\le u\le u_0$ and in this case $\lim_{u\to 0} G(u)=0$. The function $G$ is the central object in our quantitative bound on the number of bound states for the Schr\"odinger-type operator $T(p)+V$. 
\begin{theorem}\label{thm:superduper2}
	Let $T:\R^d\to [0,\infty)$ be measurable and assume that the potential $V$ is relatively form compact with respect to $T(p)$. Then $T(p)+V$ is well--defined as a quadratic form on $\calQ(T(p))$ and 
	$\sigma_\ess(T(p)+V)=\sigma_\ess(T(p))\subset[0,\infty)$. 
	Furthermore, one has, for all $0<\alpha<\tfrac{1}{4}$, the bound 
			\begin{equation}\label{eq:superduper2}
				N(T(p)+V)\le \frac{\alpha^2}{(1-4\alpha)^2}\int G(V_-(x)/\alpha^2)\, \d x ,
			\end{equation}
			where  $V_-=\max(0,-V)$ is the negative part of $V$ and $N(T(p)+V)$ is the number of eigenvalues of $T(p)+V$ which are strictly negative.  
\end{theorem}

\begin{remarks}\label{rem:superdooper2}
	\begin{thmlist}
	\item As we will see in Section \ref{sec:applications}, in many practical cases, even when the kinetic energy symbol $T$ is not homogeneous, the function $G$ from Theorem \ref{thm:superduper2} can be straightforwardly evaluated and in most cases the result of this evaluation agrees with the precise semi--classical guess up to a small factor.  
	In particular, when $T(\eta)=\eta^2$, we recover Cwikel's version of the CLR bound.  
	\item\label{thm:superduper2-sublevel} A straightforward argument shows that if for all $u>0$ the sublevel sets $\{T<u\}$ have finite Lebesgue measure, then 
	  \begin{equation*}
	  	\int_{T< u}T(\eta)^{-1}\, \d\eta<\infty \text{ for some } u>0 
	  	\Longleftrightarrow \int_{T< u }T(\eta)^{-1}\, \d\eta<\infty \text{ for all } u>0 .  
	  \end{equation*} 
	  In this case, $G(u)<\infty$ for all $u\ge 0$ is equivalent to $G(u_0)<\infty$ for some $u_0>0$. 
	  Moreover, in the case that the sublevel sets $\{T<u\}$ have finite Lebesgue measure, Lemma \ref{lem:rel-compactness2} in the appendix yields a non-trivial relative form compactness criterium of a 
	  potential $V$ which does not require that $T$ diverges to infinity at infinity. 
	  
	  Of course, in all the applications we know of one usually has 
	  $\lim_{\eta\to\infty}T(\eta)=\infty$ or, in the case of discrete Schr\"odinger operators, the range of possible momenta $\eta$ is a bounded subset of $\R^d$.  Thus in these applications one always 
	  has $G(u)<\infty$ for all $u>0$ once $G(u_0)<\infty$ for some $u_0>0$.  
	\item The function $G$ above has a nice semi-classical interpretation. We note 
	\begin{align*}
		G(u) &= \int\limits_{T/u<1} (T(\eta)/u)^{-1}\, \frac{\d\eta}{(2\pi)^d} 
		= \int\limits_{T/u<1} \frac{\d\eta}{(2\pi)^d} \int_0^\infty 
			s^{-2}\id_{\{ T/u<s\}}\,\d s \\
		&= \int_0^\infty  \int_{\R^d} 
			\id_{\{ T(\eta) -\min(1,s)u<0\}}\,\frac{\d\eta}{(2\pi)^d} \, \frac{\d s}{s^2} \\
		&=  \int_{\R^d} 
			\id_{\{ T(\eta) -u<0\}}\,\frac{\d\eta}{(2\pi)^d} 
			+ \int_0^1  \int_{\R^d} 
			\id_{\{ T(\eta) -\min(1,s)u<0\}}\,\frac{\d\eta}{(2\pi)^d} \, \frac{\d s}{s^2} \, .
	\end{align*}
	Thus with the classical phase-space volume, given by 
	\begin{align*}
		N^{\mathrm{cl}}(T+V)\coloneq \iint_{\R^d\times\R^d} \id_{\{ T(\eta)+V(x)<0\}}\, \frac{\d\eta \d x}{(2\pi)^d} 
			= \int_{\R^d} \left|\left\{ T(\cdot) +V(x)<0 \right\}\right| \,\frac{\d x}{(2\pi)^d} \, ,
	\end{align*}
	a reformulation of the bound in \eqref{eq:superduper2}  is 
	\begin{equation}
	\begin{split}\label{eq:quantum correction}
		N(T(p)&+V)\le 
		 \frac{\alpha^2}{(1-4\alpha)^2} N^{\mathrm{cl}}(T+V/\alpha^2) 
		+ \frac{\alpha^2}{(1-4\alpha)^2}\int_0^1 N^{\mathrm{cl}}\left( T(\eta) +sV(x)/\alpha^2 \right)\, \frac{\d s}{s^2} 
	\end{split} 
	\end{equation} 
	for all $0<\alpha<1/4$. 	The first part on the right hand side above is clearly related to the classical phase space volume guess suggested by the uncertainty principle and the second part can be considered as a \emph{quantum correction}.  
	This quantum correction is \emph{necessary}, since already for the usual one-particle Schr\"odinger operator it is well known that the famous CLR bound does not hold in all cases where the classical phase space volume is finite. Our bound from Theorem \ref{thm:superduper2} shows that a simple general quantitative bound on the number of bound states holds, if the contribution from the quantum corrections, i.e., the second term in \eqref{eq:quantum correction}, is finite. 
	\item\label{rem:spectrum-positive} As already mentioned, from the definition of $G$ together with a  simple monotonicity argument it is clear that if for some $u_0$ one has  $G(u_0)<\infty$ then $G(u)<\infty$ for all $0\le u\le u_0$ and also $\lim_{u\to 0+}G(u)=0$. Thus, if $G(u_0)<\infty$, or equivalently, $T\id_{T<u_0}\in L^1(\R^d)$, for some $u_0>0$,  a simple  construction yields a potential $V<0$ such that\footnote{for your favorite choice of $0<\alpha<1/4$, e.g., $\alpha=1/42$.} $\int_{\R^d} G(V_-(x)/\alpha^2)\, \d x<\infty$.  
	But then, replacing $V$ by $\lambda V$ for $0<\lambda\le 1$, the monotone convergence theorem gives $\lim_{\lambda\to 0+} \int_{\R^d} G(\lambda V_-(x)/\alpha^2)\, \d x =0$ and the bound provided by \eqref{eq:superduper2} shows 
		\begin{align*}
				 N(T(p)+\lambda V)
				&\le \frac{\alpha^2}{(1-4\alpha)^2}	\int_{\R^d} G(\lambda V_-(x)/\alpha^2)\, \d x<1
		\end{align*} 
	for all small enough $\lambda>0$. 
	So in this case there exist a strictly negative potential for which $T(p)+V$ has no strictly negative eigenvalues. 
	So it appears that Theorems \ref{thm:superduper1} and \ref{thm:superduper2} appear to be \emph{complementary}. More precisely, as  Theorem  \ref{thm:both} below shows, this is indeed the case under some slight additional global assumptions on the kinetic energy $T$, which seem to be  fulfilled in all physically relevant applications. 
	\item Theorem \ref{thm:superduper1} shows immediate appearance of strictly negative eigenvalues once the integral of $T(\eta)^{-1}$ diverges in a neighborhood of a compact set. On the other hand, Theorem \ref{thm:superduper2} yields a quantitative bound on the number of negative eigenvalues under the condition that  $T^{-1}\id_{\{T<\delta\}}$ is integrable for some $\delta>0$. Naturally, one can ask the question what could happen if  $T^{-1}$ is integrable over every compact set, but diverges globally. As an example of such a situation, we consider the operator 
	\begin{equation*}
		H= (p_1^2+p_2^2)^{1/2} - \lambda U(x_1,x_2,x_3)  \quad \text{on } L^2(\R^3)\, ,
	\end{equation*}
	for some $0\le U\in\calC_0^\infty(\R^3)$, $\lambda\ge 0 $. For sufficiently small 
	$\lambda$ this operator does not have negative spectrum (no weakly coupled bound states). On the other hand, after some critical $\lambda_{\mathrm{cr}}$ the infimum of the essential spectrum immediately goes down and discrete eigenvalues still do not exist. More precisely, one has 
	\begin{equation}
	\begin{split}
		\sigma(H)\cap (-\infty,0)&=\sigma_\ess(H)\cap (-\infty,0) \\
		&= \bigcup_{x_3\in\R } \sigma((p_1^2+p_2^2)^{1/2} -\lambda U(\cdot,x_3))\cap (-\infty,0)
	\end{split}
	\end{equation}
	and, for fixed $x_3\in\R $,  the negative spectrum of $(p_1^2+p_2^2)^{1/2} -\lambda U(\cdot,x_3)$, considered as an operator on $L^2(\R^2)$, is monotone in $\lambda>0$. Moreover, for fixed $ x_{3} \in \R $ Theorem \ref{thm:superduper2} is applicable, showing that for small enough 
	$\lambda$ the operator $(p_1^2+p_2^2)^{1/2} -\lambda U(\cdot,x_3)$  is  positive.   This example shows that for the existence of weakly coupled bound states one needs that the kinetic energy goes to zero fast enough near its zero set. 
	\end{thmlist}
\end{remarks} 
To show that under some slightly stronger conditions on $ T, $ Theorem \ref{thm:superduper1} and Theorem \ref{thm:superduper2} are complementary, we provide 
\begin{theorem}\label{thm:both} 
	Let $Z\coloneq\{T=0\}$ be the zero set of the kinetic energy symbol 
	$T:\R^d\to [0,\infty)$. Assume that $Z$ is compact and that $T$ is small only close to its zero set. 
	More precisely, let assume that for all $\delta>0$ there exists $u>0$ with 
	\begin{align}\label{eq:finiteness}
		\int_{\{T<u\}\cap Z_\delta^c} \frac{1}{T(\eta)}\, \d\eta <\infty \, ,
	\end{align}
   	where $Z_\delta^c$ is the complement of $Z_\delta$. Then  
	\begin{SLalpha}
		\item\label{thm:both-a}  \begin{align}\label{eq:divergence}
				\int_{Z_\delta} \frac{1}{T(\eta)}\,\d\eta =\infty \quad \text{for some } \delta>0
			\end{align}
			is equivalent to $T(p)+V$ having weakly coupled bound states for any non-trivial attractive potential $V$. That is, for any  $V\le 0$, $V\neq 0$, which is relatively form compact with respect to $T(p)$ and $\int V\,\d x\le 0$, the operator  $T(p)+V$ has at least one strictly negative eigenvalue. 
			
			Moreover, if in addition the kinetic energy symbol $T$ is locally bounded, then \eqref{eq:divergence} is also equivalent to  $T(p)+V$ having a strictly negative eigenvalue for non-trivial sign changing potentials $V$ in the sense that if $V\in L^1(\R^d)$, $V\neq 0$, is relatively form compact with respect to $T(p)$ and $\int V\,\d x\le 0$, then the 
	operator $T(p)+V$  has again at least one strictly negative eigenvalue. 
		\item\label{thm:both-b} 
			\begin{align}\label{eq:integrable}
					\int_{Z_\delta} \frac{1}{T(\eta)}\, \d\eta <\infty 
					\quad \text{for some } \delta>0
				\end{align}
			is equivalent to the existence of a quantitative bound on the number of bound states in the following sense: There exists a function $G_0:[0,\infty]\to [u,\infty]$ with $G_0(u)<\infty$ for all small enough $u>0$ and $\lim_{u\to0+}G_0(u)=0$ such that for any potential $V$ which is  relatively form compact with respect to $T(p)$ one has the bound 
			\begin{equation*}
				N(T(p)+V)\le \int G_0(V_-(x))\, \d x ,
			\end{equation*}
			where  $V_-=\max(0,-V)$ is the negative part of $V$ and $N(T(p)+V)$ is the number eigenvalues of $T(p)+V$ which are strictly negative.  
			Moreover, in this case one can take $G_0(u)= \frac{\alpha^2}{(1-4\alpha)^2}G(\alpha^{-2}u)$ with $G$ defined in \eqref{eq:G} and $0<\alpha<1/4$.  
	\end{SLalpha}
\end{theorem}
\begin{remarks} \label{rem:thm-both}
  \begin{thmlist}
    \item\label{rem:both-complementary} The existence of a quantitative bound 
    	on the number of strictly negative eigenvalues of $T(p)+V$ clearly implies that 
    	weakly coupled bound states do not exist, see Remark \ref{rem:superdooper2}.\ref{rem:spectrum-positive}. 
     On the other hand, we do not know of any result in the literature which shows that the absence of weakly coupled bound states implies the existence of a quantitative semi-classical type bound on the number of strictly negative bound states of $T(p)+V$. This is the main point of our Theorem \ref{thm:both}, which shows that under a rather weak regularity assumption on $T$, which is fulfilled in all physically relevant cases,  the two phenomena of weakly coupled bound states and the existence of  quantitative semi-classical type bound on the number of strictly negative bound states are indeed complementary.   		
	\item To see that \eqref{eq:finiteness} is, indeed, a weak growth and regularity on $T$ we note that \eqref{eq:finiteness} is fulfilled under the following two conditions on $T$:
		\begin{align}\label{eq:cond1}
			\text{there exists } 
				0<\veps,R<\infty 
			\text{ with } 
				T(\eta)\ge \veps \text{ for all } |\eta|> R 
		\end{align}
		and
		\begin{align}\label{eq:cond2} 
			 T \text{ is lower semi-continuous} \, .
		\end{align}
	Indeed, under the above two conditions one has 
	\begin{align}\label{eq:empty}
		\text{for any } \delta>0 \text{ there exists }u>0 \text{ with }\{T<u\}\cap Z_\delta^c = \emptyset \, , 
	\end{align} 
	 which clearly implies \eqref{eq:finiteness}. To see \eqref{eq:empty}, define $r(u)\coloneqq\sup\{\dist(\eta,Z):\, T(\eta)<u\}$. Then $\{T<u\}\subset Z_\delta$ is equivalent to $r(u)\le \delta$, so it is enough to show $\lim_{u\to 0+} r(u)=0$. Clearly, $0<u\mapsto r(u)$ is increasing.  Assume that there exists 
	$\veps_0>0$ with $r(u)\ge 2\veps_0$ for all $u>0$. Taking $u=1/n$, this yields the existence of a sequence $\xi_n\in\R^d$ with 
	$\dist(\xi_n,Z)>\veps_0$ and $T(\xi_n)<1/n$. Because of \eqref{eq:cond1} we have $\xi_n\le R$ for all large enough $n$. Thus we can take a subsequence $\xi_{n_l}$ such that $\eta=\lim_{l\to\infty}\xi_{n_l}$ exists. Because of the lower semicontinuity of $T$ one has $0\le T(\eta)\le \liminf_{l\to\infty} T(\xi_{n_l})=0$, so $\eta\in Z$, which contradicts $\dist(\eta,Z) = \lim_{l\to\infty}\dist(\xi_{n_l},Z)\ge \veps_0>0$. So $r(u)\to 0$ as $u\to 0$, which proves  \eqref{eq:empty}.
	\item \label{rem:equivalence}  Condition \eqref{eq:finiteness} ensures that 
  		we have the equivalences 
  			\begin{equation}
  				\begin{split}\label{eq:infinite}
					\int_{Z_\delta} \frac{1}{T(\eta)}\,\d\eta =\infty  \text{ for all } \delta>0
					&\Longleftrightarrow 
						\int_{Z_\delta} \frac{1}{T(\eta)}\, \d\eta =\infty 
						\text{ for some } \delta>0 \\
					&\Longleftrightarrow 
						\int_{T<u} \frac{1}{T(\eta)}\, \d\eta =\infty 
						\text{ for all } u>0 \, .
				  \end{split}
  			\end{equation}
		This clearly ensures that conditions \eqref{eq:divergence} and \eqref{eq:integrable} are complementary. The equivalence  \eqref{eq:infinite} follows from Lemma \ref{lem:equivalence} in Section \ref{sec:proof-thm-both}.  
  \end{thmlist}
\end{remarks}

We would like to stress that our assumptions on $T$ are very weak and fulfilled in all physically interesting cases.  
Our theorems have several applications, discussed in Section \ref{sec:applications}, including Schr\"odinger operators with fractional Laplacians, different types of Schr\"odinger type operators with degenerate kinetic energies such as pseudo-relativistic Schr\"odinger operators with positive mass and two-particle pseudo-relativistic Schr\"odinger operators with  different masses, including \emph{very} different masses, BCS--type operators, and discrete Schr\"odinger operators.

Our paper is organized as follows: We first address the question of existence of negative bound states. The applications of Theorem \ref{thm:superduper1} and Theorem \ref{thm:superduper2} are discussed in Section \ref{sec:applications}. 
The main idea in the proof of Theorem \ref{thm:superduper1} is first shown in a simple model case in Section \ref{sec:model case}. In Section \ref{sec:local case} we give the proof of Theoren \ref{thm:superduper1} and its refinement Theorem \ref{thm:many}  and their corollaries. In Section \ref{sec:cwikel} we give the proof of Theorem \ref{thm:superduper2} and in Section \ref{sec:proof-thm-both} the proof of Theorem \ref{thm:both}.

\section{Existence of bound states: the general setup}\label{sec:main results}
We consider operators of the form
\begin{align}
H = T(p) + V
\end{align}
where $p = -i\nabla$ is the quantum-mechanical momentum operator and the symbol of the kinetic energy $T$
is a measurable nonnegative function on $\R^d$. We define $T(p)$ as a Fourier multiplier, 
\begin{align}\label{eq:T}
  T(p)\varphi\coloneq \calF^{-1}[T(\cdot)\hatt{\varphi}(\cdot)]	
\end{align}
where we use the convention 
\begin{align*}
  \hatt{f}(\eta)\coloneq \calF(f)(\eta)
  &\coloneq \frac{1}{(2\pi)^{d/2}}	\int_{\R^d} e^{-i\eta\cdot x} f(x)\, dx \\
  \intertext{and}
  \widecheck{g}(x)\coloneq \calF^{-1}(g)(x)
  &\coloneq \frac{1}{(2\pi)^{d/2}}	\int_{\R^d} e^{i\eta\cdot x} g(\eta)\, d\eta  
\end{align*}
for the Fourier transform and its inverse. A-priori the above expressions are only defined when $f,g$ are Schwarz class, but they extend to unitary operators to all of $L^2(\R^d)$ by density of Schwarz functions in $L^2(\R)$, see \cite{LiebLoss,Simon-course1}. 


    
 Physical heuristic suggests that a weak attractive potential can create a bound state if $T$ is small close to its zero set. To make this precise we introduce a local version of this. 
\begin{definition}\label{def:thick}
 Let $T:\R^d\to [0,\infty)$ be measurable. $T$ has a thick zero set near $\omega\in Z$ if 
 \begin{align}\label{eq:thick near omega}
 	\int_{B_\delta(\omega)} T(\eta)^{-1}\, d\eta = \infty \text{ for all } \delta>0,
 \end{align}
 where $B_\delta(\omega)= \{\eta\in\R^d:\, \dist(\eta,\omega)<\delta\}$ is the open ball of radius $\delta$ centered at $\omega$. 
\end{definition}

%
%

In the following, we will assume, without mentioning all the time, that the assumptions of Theorem \ref{thm:superduper1} hold.  
A local version of Theorem \ref{thm:superduper1} is given by 

\begin{theorem}\label{thm:one}
	Suppose $V\neq 0$ is relatively form compact with respect to 
	$T(p)$, $\int V\, dx\le 0$, and $T$ has a thick zero set near 
	some point $\omega\in Z$. Then $T(P)+V$ has at 
	least one strictly negative eigenvalue.
\end{theorem}


A refinement of Theorem \ref{thm:one}, when the zero set of the kinetic energy $T$ has many disjoint thick parts is given by the next theorem, which also yields an easy criterion for the infinitude of weakly coupled bound states.

\begin{theorem}\label{thm:many}
 Assume that for some $k\in\N$ the kinetic energy $T$ has a thick zero set near $k$ pairwise distinct points  $\omega_1,\ldots,\omega_k\in Z$  
 and  $V \neq 0$ is relative form compact with respect to $T(p)$.
 \begin{SLalpha}
 	\item \label{thm:many-a} If $V\le 0$,   then $T(p)+V$ has at least $k$ strictly negative eigenvalues. 
 	\item \label{thm:many-b} If $V\in L^1(\R^d)$ has not a fixed sign and if the $k\times k$ 
 		  matrix $M=(\hatt{V}(\omega_l-\omega_m))_{l,m=1,\ldots,k }$, where $\hatt{V}$ is the Fourier transform of $V$,  is strictly negative definite, then $T(p)+V$ has at least $k$ strictly negative eigenvalues. 
 	\item \label{thm:many-c} If $V\in L^1(\R^d)$, the symbol $T$ is locally bounded, the matrix $M=(\hatt{V}(\omega_l-\omega_m))_{l,m=1,\ldots,k }$ is negative semi-bounded, and the eigenvalue 0 of this matrix is non-degenerate, then $T(p)+V$ has at least $k$ strictly negative eigenvalues. 
 \end{SLalpha}
\end{theorem}
\begin{remarks}
	\begin{thmlist}
		\item In the spirit of Theorem \ref{thm:many}, one can formulate a condition under which the operator $T(p)+V$ has at least $k$ eigenvalues for a semi-bounded matrix $M=(\hatt{V}(\omega_l-\omega_m))_{l,m=1,\ldots,k }$ with a \emph{degenerate} eigenvalue zero. We are not doing this for the sake of simplicity, but leave it to the interested reader. 
		\item  In \cite{BrueningGeylerPankrashkin,Pankrashkin} the authors   also had the condition that the matrix 
				$M=(\hatt{V}(\omega_l-\omega_m))_{l,m=1,\ldots,k }$ is strictly negative definite, but the conditions  on the zero set of the kinetic energy are much stronger than ours in Theorem \ref{thm:many}.  Moreover, we can also handle the case of a non-degenerate zero eigenvalue  of $M$. 
	\end{thmlist}
\end{remarks}
Useful corollaries are  
\begin{corollary}\label{cor:isolated-points} 
 Assume that $V\neq 0$ is relatively form compact with respect to $T(p)$ and that there are $k$ isolated points $\omega_1,\ldots,\omega_k$ and that near 
 $A=\{\omega_1,\ldots,\omega_k\}$, i.e., in an open neighborhood $\calO$ 
 containing $A$,  the kinetic energy symbol obeys the bound 
 \begin{align}
 	T(\eta)\lesssim \dist(\eta,A)^\gamma \quad\text{for all } \eta\in \calO
 	\text{ and some } \gamma\ge d\, .
 \end{align}
 \begin{SLalpha}
	\item If $V\le 0$ or if $V\in L^1(\R^d)$ and the matrix $M= (\hatt{V}(\omega_l-\omega_m))_{l,m=1,\ldots,k }$ is strictly negative definite, then $T(p)+V$ has at least $k$ strictly negative eigenvalues.
	\item If $\int V\, dx\le 0$, then $T(p)+V$ has at least one strictly negative eigenvalue.
 \end{SLalpha}
\end{corollary}
\begin{corollary}\label{cor:submanifold}
 Suppose $V \neq 0$ is relatively form compact with respect to $T(p)$.  
 Also assume that there is a 
 $\calC^2$ submanifold $\Sigma$ of codimension $1\le n\le d-1$ and that 
 near $\Sigma$, i.e., in an open neighborhood $\calO$ containing $\Sigma$,  the kinetic energy symbol obeys the bound 
 \begin{align}
 	T(\eta)\lesssim \dist(\eta,\Sigma)^\gamma \quad\text{for all } \eta\in \calO
 	\text{ and some } \gamma\ge n\, .
 \end{align}
  \begin{SLalpha}
	\item If $V\le 0$ then $T(p)+V$ has infinitely many strictly negative eigenvalues.
	\item If $\int V\, dx\le 0$, then $T(p)+V$ has at least one strictly negative eigenvalue.
 \end{SLalpha}
\end{corollary}

\begin{remark}
	In most applications $T$ is continuous and the zero set of $T$ is either a point, a collection of points, or a smooth submanifold in $\R^d$. So Corollaries \ref{cor:isolated-points} and \ref{cor:submanifold} would be enough to cover all applications we can think of. However, we find that the proof for the general case is so simple, that adding further structure to its assumptions only obscures the simplicity of the proof. So we prefer to state Theorem \ref{thm:superduper1} and its local versions, Theorems \ref{thm:one} and  \ref{thm:many}, in their generality.  
\end{remark}
Some of the most interesting applications will be considered next.

\section{Applications}\label{sec:applications}
In the current section we discuss the following applications of Theorems \ref{thm:superduper1} and \ref{thm:superduper2} and their corollaries.  
\begin{SE}
\item[\ref{sec:fractional Schroedinger}] Schr\"odinger operators with a fractional Laplacian.
  \item[\ref{sec:relativistic one-particle}] Relativistic one--particle operators with positive mass.
  \item[\ref{sec:relativistic pair operators}] Relativistic pair operators with positive mass. 
  \item[\ref{sec:ultra-relativistic pair}] Ultra--relativistic pair operators. 
  \item[\ref{sec:relativistic pair mixed}] Relativistic pair operators: one heavy and one extremely light particle.
    \item[\ref{sec:BCS}] Operators arising in the mathematical treatment of the Bardeen--Cooper--Schrieffer  theory of superconductivity (BCS).
  \item[\ref{sec:discrete Schroedinger}] Discrete Schr\"odinger operators on a lattice.
\end{SE}
We would like to stress the fact that in all cases, except one, where we get a finite bound for the number of bound states, these bounds agree, up to constants, exactly with what one would guess from semi--classics. The one example, where we do not get a semi-classical type bound, is considered in Theorem \ref{thm:relativistic pair operator d=3} in section \ref{sec:relativistic pair operators} and we do \emph{expect to have corrections} to the semi--classical picture, since this is a \emph{critical case}. 

We would also like to emphasise that Theorem \ref{thm:superduper2} easily allows to get these semi--classical bounds even for kinetic energies which are \emph{not homogenous}! 
\subsection{Schr\"odinger operators with a fractional Laplacian}\label{sec:fractional Schroedinger}
 We consider the operator $ (-\Delta)^{\gamma/2} + V = |p|^{\gamma} + V $ in $ \mathbb{R}^{d} $ assuming that $ V $ satisfies the conditions of Theorems \ref{thm:superduper1} and \ref{thm:superduper2}. It follows immediately from Corollary \ref{cor:submanifold} that 
 \begin{theorem}
Suppose $V \neq 0$ is an attractive potential in the sense that $\int V dx \leq 0$. Then for $ \gamma \geq d $ the operator $ |p|^{\gamma} + V $ has at least one striclty negative eigenvalue.
 \end{theorem}
 On the other hand, Theorem \ref{thm:superduper2} implies for this operator that 
 \begin{theorem}\label{thm:bound states-fractional}
 Assume that $ \gamma < d $ then the number of negative eigenvalues of the operator $ |p|^{\gamma} + V $ satisfies 
 \begin{align} \label{eq:bound states-fractional}
 	N(|p|^{\gamma} + V) \leq \left(\frac{4d}{d-\gamma}\right)^{2d/\gamma} \frac{d-\gamma}{d(4\gamma)^{2}} \frac{|B_1^d|}{(2\pi)^{d}} \int_{\mathbb{R}^{d}} V_{-}(x)^{d/\gamma} \, dx,
 \end{align}
where $ |B_1^d| $ denotes the the volume of the unit ball in $ \R^{d}$.
 \end{theorem}
\begin{proof} 
Inequality \eqref{eq:bound states-fractional} follows from (\ref{eq:superduper2}) for the optimal choice of $ \alpha = \frac{d-\gamma}{4d}. $
\end{proof}

 Of course, if $\gamma\ge d$, then there \emph{cannot} be a quantitative bound on the number of negative eigenvalues, but for $\beta>0$ the  number of eigenvalues below $-\beta$, which we denote by $N(|p|^\gamma +V, -\beta)$, can be bounded. 
 
 Define, for $\gamma>d$ and $\beta>0$,  
 \begin{align}
 	G_{d,\gamma,\beta}(u)= \beta^{\frac{d-\gamma}{\gamma}}\frac{|B_1^d|}{(2\pi)^d}\,  u \min\left(\Big(\frac{u}{\beta}-1\Big)_+^{\frac{d}{\gamma}}, \frac{\pi d/\gamma}{\sin(\pi d/\gamma)}\right)
 \end{align}
 and, for $\gamma=d$,  
 \begin{align}
 	G_{d,d,\beta}(u)= \frac{|B_1^d|}{(2\pi)^d}\,  u \ln\left(1+\Big(\frac{u}{\beta}-1\Big)_+\right) 
 \end{align}
 for $u\ge 0$. 
 Then we have 

 \begin{theorem} \label{thm:bound states beta} 
 	Assume that $\gamma\ge d$ and $\beta>0$. Then for all $0<\alpha<1/4$
 	\begin{align*}
 		N(|p|^\gamma +V, -\beta)
 			\le \frac{\alpha^2}{ (1-4\alpha)^{2}} 
 			\int_{\R^d} G_{d,\gamma,\beta}(\alpha^{-2}V_-(x))\, dx\, .
 	\end{align*}
 \end{theorem}

\begin{remarks}
  \begin{thmlist}
	 \item If $V_-\le \beta$, then $|p|^\gamma +V$ cannot have any  spectrum below $-\beta$ and the bound from Theorem \ref{thm:bound states beta} reflects this.
	 \item Note that $L_{0,d}^\mathrm{cl}=\frac{|B_1^d|}{(2\pi)^d}$ is one of the so-called classical Lieb-Thirring constants	\cite{Hundertmark-Simonfest, LiebThirring}.
  \end{thmlist}
\end{remarks}
\begin{proof}
 Of course, $N(|p|^\gamma +V, -\beta)= N(\beta+|p|^\gamma +V) = N(T(p)+V)$ with $T(\eta)= \beta+|\eta|^\gamma$ by a simple  shifting argument. 
 Thus Theorem \ref{thm:superduper2} shows that a bound  of the form of \eqref{eq:superduper2} holds for $N(|p|^\gamma +V, -\beta)$ with $G$ given by 
 \begin{align}
  G(u)= 
 	 u\int_{T<u} \frac{1}{T(\eta)}\,\frac{\d\eta}{(2\pi)^d} 
 	 =  u \int_{|\eta|^\gamma<u-\beta} \frac{1}{\beta +|\eta|^\gamma}\,\frac{\d\eta}{(2\pi)^d}\, .
 \end{align}
 We have  
 \begin{align*}
 	\int_{|\eta|^\gamma<\beta s} \frac{1}{\beta +|\eta|^\gamma}\,\d\eta
 	= \beta^{\frac{d}{\gamma}-1} 
 		\int_{|\eta|^\gamma<s} \frac{1}{1+|\eta|^\gamma}\,\d\eta
 	= |S^{d-1}| \beta^{\frac{d}{\gamma}-1}  \int_0^{s^\frac{1}{\gamma}} \frac{r^{d-1}}{1+r^\gamma}\, \d r
 \end{align*}
 by scaling and going to spherical coordinates, $|S^{d-1}|$ is the surface area of the unit sphere in $\R^d$. 
 Thus, if $\gamma=d$, then 
 \begin{align*}
 	G(u)  = \frac{|S^{d-1}|}{d(2\pi)^d} \, u\ln \left(1+\Big( \frac{u}{\beta}-1 \Big)_+ \right)
 \end{align*}
 and since $|B_1^d|=|S^{d-1}|/d$ this proves the claim for $\gamma=d$.  
 If $\gamma>d$, then 
 \begin{align*}
 	\int_0^\infty \frac{r^{d-1}}{1+r^\gamma}\, \d r 
 	= \frac{1}{\gamma} \int_0^\infty t^{\frac{d}{\gamma}-1} (1+t)^{-1}\, \d t
 	= \frac{\pi/\gamma}{\sin( \pi d/\gamma )}
 \end{align*}
 by a contour integral over a keyhole contour encircling the positive real axis \cite{BakNewman}.  
 On the other hand, 
 $
 	\int_0^{s^\frac{1}{\gamma}} \frac{r^{d-1}}{1+r^\gamma}\, \d r 
 	\le \int_0^{s^\frac{1}{\gamma}} r^{d-1}\, \d r 
 	= \frac{s^{\frac{d}{\gamma}}}{d} 
 $, so 
 \begin{align*}
 	\int_{|\eta|^\gamma<u-\beta} \frac{1}{\beta +|\eta|^\gamma}\,\frac{\d\eta}{(2\pi)^d}
 	\le \frac{|S^{d-1}|}{d(2\pi)^d}
 		\min\left(\Big(\frac{u}{\beta}-1\Big)_+^{\frac{d}{\gamma}}, \frac{\pi d/\gamma}{\sin(\pi d/\gamma)}\right)
 \end{align*}
 which shows $G(u)\le G_{d,d,\beta}(u)$ if $\gamma=d$. This proves the theorem.

\end{proof}

\subsection{Relativistic one--particle operators with positive mass}\label{sec:relativistic one-particle}
If one wants to include relativistic effects, one is often lead to pseudorelativistic operators where the kinetic energy is of the form $T_{c^2m}(p)=\sqrt{|p|^2+c^4m^2} -c^{2}m$ \cite{BalinskyEvans, Fefferman-de-la-Llave, HardenkopfSucher, Herbst, LiebSeiringer, LiebYau}. Here $m$ is the mass of the particle and $c$ is the velocity of light. In the limit of $c\to 0$, i.e., the ultra-relativistic limit, and in the limit of vanishing mass, i.e., massless particles, this becomes simply the operator $|p|$, which was already discussed in Section \ref{sec:fractional Schroedinger}. For non--vanishing mass, one can set $c=1$ by absorbing $c^2$ into $m$ with a simple scaling argument.   We have 
\begin{theorem}\label{thm:relativistic one-particle small d}
	Let $d=1$, or $d=2$, $V\in L^1(\R^d)$ is relatively form compact with respect to $T_{m}(p)$ and an attractive potential in the sense that $V\neq0$ 
	and $\int V\, dx \le 0$, then the operator $ \sqrt{|p|^2+m^2} -m + V $ has at least one strictly negative eigenvalue.
\end{theorem}
\begin{proof} Since 
$\sqrt{|\eta|^2+m^2} -m = \tfrac{|\eta|^2}{2m} +\Oh(\frac{|\eta|^4}{m^3})$, the claim follows from Corollary \ref{cor:isolated-points}. 
\end{proof}
For larger dimensions, we have a quantitative bound on the number of negative eigenvalues, counting multiplicity.  
\begin{theorem}\label{thm:number of bound states-relativistic one-particle}
 For $d\ge 3$ and $m\ge 0$ let 
 \begin{align*}
 	G_{d,m}(u) \coloneq \frac{d}{d-2}\frac{|B^d_1|}{(2\pi)^d} u^{d/2}(u+2m)^{d/2}\, ,
 \end{align*}
 where $|B^d_1|$ denotes the volume of the ball of radius one in $\R^d$. 
 Then the number of negative eigenvalues of $ \sqrt{|p|^2+m^2} -m + V  $ on $L^2(\R^d)$ satisfies, for any $0<\alpha<\tfrac{1}{4}$, 
 \begin{equation} \label{eq:bound states-relativistic one-particle}
 \begin{split}
 	N(\sqrt{|p|^2+m^2} -m + V) 
 	\leq 
 	\frac{\alpha^2}{(1-4\alpha)^2}  
 	\int_{\mathbb{R}^{d}} G_{d,m}(\alpha^{-2}V_{-}(x))\, dx \, .
 \end{split}
 \end{equation} 
 \end{theorem}

\begin{remarks}
\begin{thmlist}
	\item We have  
 		\begin{align*}
 			G_{d,m}(u) = \frac{d}{d-2} \frac{|B_1^d|}{(2\pi)^d} \big| \big\{\eta\in\R^d:\, T_m(\eta)<u\big\} \big|\, ,
 		\end{align*}
	where $|A|$ denotes the volume of a Borel set $A\subset\R^d$. So up to a factor of 
	$d/(d-2)$, the function $G_m$ is exactly what one would expect  from a semi--classical guess. 
	\item In the limit $m\to 0$, one recovers, with a slightly worse constant, the bound from Theorem \ref{thm:bound states-fractional} as long as  $d\ge 3$. 
	\item Physical intuition suggests that for bound states large (negative) values of the potential correspond to a large momentum and small values to a small momentum. Since  
	$\sqrt{|\eta|^2+m^2} -m\simeq |\eta|$ for large and  $\sqrt{|\eta|^2+m^2} -m\simeq \tfrac{|\eta|^2}{2m}$ for small momentum, physical heuristics thus suggests that a pseudo-relativistic system has a \emph{finite} number of bound states if the ``large values" of $V_-$ are in $L^d(\R^d)$ and the ``small values" are in $L^{d/2}(\R^d)$.   
	It is easy to see that 
	\begin{align*}
		 G_{d,m}(u) 
		\lesssim
		 \min(u,1)^{d/2} + \max(u,1)^{d} \, , 
	\end{align*}
	where the implicit constants depend only on $m$ and $d$, so our bound \eqref{eq:bound states-relativistic one-particle} corroborates this physical heuristic argument quantitatively.	
\end{thmlist}
\end{remarks}
 
\begin{proof}[Proof of Theorem \ref{thm:number of bound states-relativistic one-particle}: ]
 With $T(\eta)=T_m(\eta)= \sqrt{|\eta|^2+m^2} -m$ we rewrite $G$ from \eqref{eq:G} as 
 \begin{align}\label{eq:calculate G-1}
 	G(u)&= u\int_{T<u} \frac{1}{T(\eta)}\,\frac{\d\eta}{(2\pi)^d} 
 	 = \int_{T/u<1} \int_0^\infty s^{-2}\id_{\{T(\cdot)/u<s\}}(\eta) \,\d s \frac{\d\eta}{(2\pi)^d} \nonumber\\
 	&= 
 		\frac{1}{(2\pi)^d} \int_0^\infty s^{-2}
 		\big|\{T<\min(1,s)u\}\big| \,\d s \nonumber \\
 	&=\frac{1}{(2\pi)^d} \left[\big|\{T<u\}\big| + u\int_0^u s^{-2}
 		\big|\{T<s\}\big| \,\d s \right]\, .
 \end{align}
 Since $\big|\{T<u\}\big|= \int_{|\eta|<(u(u+2m))^{1/2}}\, \d\eta =|B^d_1|u^{d/2}(u+2m)^{d/2}$,  we get from \eqref{eq:calculate G-1} 
 \begin{align*} 
    G(u) 
 	&= \frac{|B^d_1|}{(2\pi)^d} \left[ u^{d/2}(u+2m)^{\frac{d}{2}} + u\int_0^u s^{\frac{d}{2}-2}
 		(s+2m)^{d/2} \,\d s \right]\, .
 \end{align*}
 Using the simple bound 
 \begin{align*}
 	u\int_0^u s^{\frac{d}{2}-2}
 		(s+2m)^{\frac{d}{2}} \,\d s 
 		\le u(u+2m)^{\frac{d}{2}}\int_0^u s^{\frac{d}{2}-2}\,\d s 
 		=  \frac{2}{d-2} u^{\frac{d}{2}}(u+2m)^{\frac{d}{2}}\, ,
 \end{align*}
 we get the upper bound
 \begin{align*}
   	G(u) 
 	\le  \frac{d}{d-2}\frac{|B_1^d|}{(2\pi)^d} u^{\frac{d}{2}}(u+2m)^{\frac{d}{2}} = G_{d,m}(u)
 \end{align*}
 and  Theorem \ref{thm:superduper2} applies. 

\end{proof}

\subsection{Relativistic pair operators with positive masses} \label{sec:relativistic pair operators}
Considering two relativistic particles of masses $m_\pm$ interacting with each other one is lead to study the operator 
\begin{equation*} \sqrt{p_-^2+m_+^2}-m_+ + \sqrt{p_+^2+m_-^2}-m_- +V(x_+-x_-)\, 
\end{equation*}	
 on $L^2(\R^{2d})$, with $p_\pm=-i\nabla_{x_\pm}$, the momenta of the first ($+$) and second particle ($-$), where $\R^{2d}\ni x=(x_+,x_-)\in \R^d\times\R^d$. See, for example, \cite{LewisSiedentopVugalter}, where they study the essential spectrum, the extension of the famous HVZ Theorem to semi-relativistic particles, for an arbitrary but fixed number of particles. 

It turns out that, due to the fact that this operator does not transform in a simple way under Gallilei transformations, this system has some unusual features. Transforming this two-particle  operator into center of mass coordinates, one gets a direct integral decomposition $\int_{\R^d}^{\oplus} H_{\text{rel}, m_\pm}(P)\, \d P$, see  \cite{LewisSiedentopVugalter} and  \cite{VugalterWeidl}, where $H_{\text{rel},m_\pm}(P)$ is the \emph{pair-operator} 
\begin{align}\label{eq:relativistic pair}
	H_{\text{rel}, m_\pm}(P)&\coloneqq 
		\sqrt{|\mu_+P-q|^2+\mu_+^2M^2} +\sqrt{|\mu_-P+q|^2+\mu_-^2M^2}  
		- \sqrt{P^2+M^2} +V(y)\nonumber\\
	&=: T_{P,M,\mu_\pm}(q) +V(y)	
\end{align}
 on $L^2(\R^d)$. Here $y\in \R^d$ is the relative coordinate and $q=-i\nabla_y$ the relative momentum of the two particles, $M=m_- + m_+$ is the total mass, $P\in\R^3$ the total momentum, and we set $\mu_\pm\coloneqq m_\pm/M$. 
 
 Note that the dependence on the total momentum is much more complicated than in the non--relativistic case, where the two particle operator 
 \begin{align*}
 	\frac{p_+^2}{2m_+} + \frac{p_-^2}{2m_-} + V(x_+-x_-)
 \end{align*}
 on $L^2(\R^d)$ is unitarily equivalent to a direct integral 
 $\int_{\R^d}^{\oplus} H_{\text{non--rel}}(P)\, \d P$ with the non-relativistic pair--operator 
 \begin{align*}
 	H_{\text{non--rel}}(P)\coloneqq 
		\frac{M}{2m_+ m_-} p^2 + V(y)  + \frac{P^2}{2M} \, .
 \end{align*}
 Here the term $\frac{P^2}{2M}$ is simply the kinetic energy of the center of mass frame and the shift by $\frac{P^2}{2M}$ corresponds to the covariance  of non--relativistic Schr\"odinger operators under Galilei transformations.  

Bounds for the number of bound states for the relativistic pair--operator\eqref{eq:relativistic pair} were considered in \cite{VugalterWeidl}. Here we want to show how their results are an easy consequence of our approach. Moreover, in the following section, we will consider the ultra--relativistic pair--operator and prove a conjecture made in \cite{VugalterWeidl} concerning the limit of vanishing masses, when both particles are ultra-relativistic in Section \ref{sec:ultra-relativistic pair}. Moreover, we will also see how within our approach one can easily discuss a mixed, relativistic--ultra relativistic, case, where one particle has positive mass while the other one has zero mass, see Section \ref{sec:relativistic pair mixed}.    
 
For positive masses and low dimensions we have 
\begin{theorem}\label{thm:relativistic pair-operator small d}
	Suppose that  $V \neq 0$, $V\in L^1(\R^d)$ is relatively form-compact with respect to $T_{P,M,\mu_\pm}(p)$ and is an attractive potential in the sense that $\int V dx \leq 0$. Then for $  d=1,2 $ and any $P\in\R^d$, $m_\pm>0$ the operator $ H_{\text{rel}, m_\pm}(P) $ has at least one negative eigenvalue.
\end{theorem}
\begin{proof}
	Using the Taylor expansion  $\sqrt{1+x}= 1+\tfrac{x}{2} - \tfrac{x^2}{8}+\Oh (x^3)$ and $\mu_++\mu_-=1$ one sees  
	\begin{align*}
		T_{P,M,\mu_\pm}(\eta)& =  \sqrt{P^2+M^2}\left(\mu_+\sqrt{1+  \frac{\eta^2-2\mu_+P\cdot\eta}{\mu_+^2(P^2+M^2)}} +\sqrt{1+  \frac{\eta^2+2\mu_-P\cdot\eta}{\mu_-^2(P^2+M^2)}} \,  
		- 1\right) \\
			& = \frac{(P^2+M^2)\eta^2 - (P\cdot\eta)^2}{2 \mu_+ \mu_- (P^2 + M^2)^{3/2}}  +\Oh\left(\frac{|\eta|^4+\mu_\pm\eta^2|P\eta|}{\mu_\pm^3(P^2+M^2)^{3/2}}\right) \\
			& =  \frac{M^2\eta^2 + P^2 \eta_\perp^2}{2 \mu_+ \mu_- (P^2 + M^2)^{3/2}} +\Oh\left(\frac{|\eta|^4+\mu_\pm\eta^2|P\eta|}{\mu_\pm^3(P^2+M^2)^{3/2}}\right)
	\end{align*}
 where $M=m_+ + m_->0$, $\mu_\pm=m_\pm/M>0$, and $\eta= sP+\eta_\perp$ with $s\in\R$ and $\eta_\perp$ orthogonal to $P$ if $P\neq 0$.  Thus the kinetic energy vanishes quadratically near $\eta=0$ and Corollary \ref{cor:isolated-points} applies. 
\end{proof}

To give a quantitative bound on the number of negative bound states we need a little bit more notation. 
Let $d\in\N$, $P\in\R^d$, $M\ge 0$, $\gpm =\sqrt{P^2+M^2}$, $-1/2\le \wti{\mu}\le 1/2$, and define 
 \begin{align} 
 	G^d_{P,M,\wti{\mu}}(u)\coloneqq \frac{3|B^d_1|}{(4\pi)^d}\frac{u^{d/2}(u+\gpm )\left( u +2\gpm  \right)^{d/2}}{( u^2+2u\gpm  +M^2 )^{1/2}} 
 	\left( \frac{u^2+2u\gpm  +(1-4\wti{\mu}^2)M^2}{u^2+2u\gpm  +M^2} \right)^{d/2}\, .
 \end{align}
 \begin{remark} The function $G^d_{P,M,\wti{\mu}}$ is quite natural. 
Up to a factor of three it is \emph{}exactly what one would guess semi--classically, that is, 
\begin{align*}
	G^d_{P,M,\wti{\mu}}(u)= \frac{3}{(2\pi)^d}\left| \{\eta\in\R^d:\,T_{P,M,\mu_\pm}(\eta) < u\} \right|
\end{align*} 
whith $\wti{\mu}=(\mu_--\mu_+)/2$. 
For the convenience of the reader, we sketch the calculation of $\left| \{\eta\in\R^d:\,T_{P,M,\mu_\pm}(\eta) < u\} \right|$ from \cite{VugalterWeidl} in Appendix \ref{app:calculation}.
\end{remark}

In four and more dimensions we have a simple bound for the number of bound states of the relativistic pair--operator.
 \begin{theorem}[Bound states in high dimension]\label{thm:relativistic pair operator large d}
 If $d\ge 4$, then the number of negative eigenvalues of the relativistic pair--operator $H_{\text{rel},m_\pm}(P) $ on $L^2(\R^d)$ satisfies for  $0<\alpha<\tfrac{1}{4}$   
 \begin{equation} \label{eq:bound states-relativistic pair large d}
 \begin{split}
 	N(H_{\text{rel},m_\pm}(P)) 
 	\leq 
 	\frac{\alpha^2}{(1-4\alpha)^2}
 	\int_{\mathbb{R}^{d}} G^d_{P,M,\wti{\mu}}\left(\alpha^{-2}V_{-}(x)\right) \, dx\, .
 \end{split}
 \end{equation}
 
 \end{theorem}
\begin{remark}
 For positive total mass $M>0$ we have 
 \begin{align*}
 	G^d_{P,M,\wti{\mu}}(u)\lesssim \min(u,1)^{d/2} + \max(u,1)^d
 \end{align*}
 where the implicit constants depend only on $P\in\R^d$ and $M>0$.  
 So the number of relativistic pair-bound states is finite if the interaction potential is locally $L^d(\R^d)$ and globally $L^{d/2}(\R^d)$, as in the single particle case. 
 
 In the limit of zero total mass one has 
 \begin{align*}
 	G^d_{P,0}(u)=\lim_{M\to 0}G^d_{P,M,\wti{\mu}}(u)= \frac{3|B_1^d|}{(4\pi)^d}u^{(d-1)/2}(u+|P|)(u+2|P|)^{(d-1)/2}\, .
 \end{align*}
 Compared to the massive case, there is (half) a \emph{loss of one dimension} in the massless case: 
 the negative part of the interaction potential $V_-$ has to be locally in 
 $L^{d}(\R^d)$ and globally in $L^{(d-1)/2}(\R^d)$  to have finitely many bound states. 
 
 Informally, setting $d=3$, one sees that $G^3_{P,0}$ contains a term \emph{linear} in $u$. This hints at the fact that in this case a quantitative bound on the number of bound states should not exist. As Theorem \ref{thm:relativistic pair operator d=3} below shows, this is indeed the case. Moreover, any bound for positive masses $m_\pm>0$ should diverge as $M=m_++m_-\to 0$. We will see in the next theorem, that this divergence is at most logarithmic in the limit of vanishing total mass $M$ when $d=3$.  That such a  divergent term is necessary is shown in Theorem    \ref{thm:ultra-relativistic d=3}. 
\end{remark}
\begin{proof}[Proof of Theorem \ref{thm:relativistic pair operator large d}:]
  With 
  \begin{align} 
  T(\eta)=T_{P,M,\mu_\pm}(\eta) = \sqrt{|\mu_+P-\eta|^2+\mu_+^2M^2} +\sqrt{|\mu_-P+\eta|^2+\mu_-^2M^2}  - \sqrt{P^2+M^2}    	
  \end{align}
   we rewrite $G$ from \eqref{eq:G} as 
 \begin{align}\label{eq:calculate G-2}
 	G(u)&= u\int_{T<u} \frac{1}{T(\eta)}\,\frac{\d\eta}{(2\pi)^d} 
 	 = 
 		\frac{1}{(2\pi)^d} \int_0^\infty s^{-2}
 		\big|\{T_{P,M,\mu_\pm}<\min(1,s)u\}\big| \,\d s \nonumber \\
 	&=\frac{1}{(2\pi)^d} \left\{\big|\{T_{P,M,\mu_\pm}<u\}\big| + u\int_0^u s^{-2}
 		\big|\{T_{P,M,\mu_\pm}<s\}\big| \,\d s \right\}
 \end{align}
 From Appendix \ref{app:calculation}, see formula \eqref{eq:volume-symbol-rel-pair}, we know 
 \begin{align*}\big|\{T_{P,M,\mu_\pm}<u\}\big|= \frac{|B^d_1|}{2^d}\frac{u^{d/2}(u+\gpm )(u+2\gpm )^{d/2}}{(u^2+2u\gpm +M^2)^{1/2}}
 	\left( \frac{u^2+2u\gpm +(1-4\wti{\mu}^2)M^2}{u^2+2u\gpm +M^2} \right)^{d/2}
 \end{align*}	
 where $\wti{\mu}=(\mu_--\mu_+)/2$ and $\gpm =\sqrt{P^2+M^2}$. So it is enough to show that 
 \begin{align}\label{eq:great d ge 4}
 	u\int_0^u s^{-2}
 		\big|\{T_{P,M,\mu_\pm}<s\}\big| \,\d s 
 		\le 2 \big|\{T_{P,M,\mu_\pm}<u\}\big|. 
 \end{align}
 Since the map 
 \begin{align*}
 	0\le  u\mapsto \frac{u^2+2u\gpm +(1-4\wti{\mu}^2)M^2}{u^2+2u\gpm +M^2}
 \end{align*}
 is increasing, we have for $d\ge 4$  
 \begin{align*}
 	&2^d|B^d_1|^{-1}\int_0^u s^{-2} 
 		\big|\{T_{P,M,\mu_\pm}<s\}\big| \,\d s 
 	 \\
 	& \le u^{\frac{d}{2}-2} (u+2\gpm )^{d/2} \left( \frac{u^2+2u\gpm +(1-4\wti{\mu}^2)M^2}{u^2+2u\gpm +M^2} \right)^{d/2} 
 		\int_0^u \frac{s+\gpm }{(s^2+2s\gpm +M^2)^{1/2}}\, ds 
 \end{align*}
 and 
 \begin{align*}
 	\int_0^u &\frac{s+\gpm }{(s^2+2s\gpm +M^2)^{1/2}}\, ds 
 		= (u^2+2u\gpm +M^2)^{1/2} -M \\
 	&= \frac{u^2+2u\gpm }{(u^2+2u\gpm +M^2)^{1/2}+M} 
 	  \le \frac{2u(u+\gpm )}{(u^2+2u\gpm +M^2)^{1/2}}\, . 
 \end{align*}
 Putting the above bounds together shows \eqref{eq:great d ge 4} and with Theorem \ref{thm:superduper2} we conclude the proof.

\end{proof}

The following, in tandem with Theorem \ref{thm:relativistic pair operator d=3}, shows  that the relativistic pair-operator is ``critical" in three dimensions.
\begin{theorem}[Bound states in dimension $3$]\label{thm:relativistic pair operator d=3} 
	Let $P\in\R^3$, $m_\pm> 0$ and set $ M=m_++m_-$, $\mu_\pm= m_\pm/M$, 
 $\wti{\mu} =(\mu_- - \mu_+)/2$, $\gpm =\sqrt{P^2+M^2}$ and 
 \begin{align*} 
 	G^\text{mod}_{P,M,\wti{\mu}}(u)&\coloneqq 
 	 G^3_{P,M,\wti{\mu}}(u)+ u\frac{|B^3_1|}{2^{1/2}(2\pi)^3} 	 (P^2+M^2)\ln \left(\sqrt{1+P^2/M^2} + \sqrt{2+P^2/M^2}\right)\, .
 \end{align*} 
 Then the number of negative eigenvalues of the relativistic pair--operator $H_{\text{rel},m_\pm}(P) $ on $L^2(\R^3)$ satisfies for $0<\alpha<\tfrac{1}{4}$ 
 \begin{equation} \label{eq:bound states-relativistic d=3}
 \begin{split}
 	N(H_{\text{rel},m_\pm}(P)) 
 	\leq 
 	\frac{\alpha^2}{(1-4\alpha)^2}
 	\int_{\mathbb{R}^{3}} G^\text{mod}_{P,M,\wti{\mu}}\left(\alpha^{-2}V_{-}(x)\right) \, dx \, .
 \end{split}
 \end{equation}
\end{theorem}
\begin{proof} As in the proof of Theorem \ref{thm:relativistic pair operator large d}, Theorem \ref{thm:superduper2} yields a bound on the number of bound states with 
 \begin{align*}
 	G(u)
 	&=\frac{1}{(2\pi)^3} \left\{\big|\{T_{P,M,\mu_\pm}<u\}\big| + u\int_0^u s^{-2}
 		\big|\{T_{P,M,\mu_\pm}<s\}\big| \,\d s \right\}
 \end{align*}
 where now $|A|$ denotes the volume of a Borel set $A\subset \R^3$. So it is enough to show
 \begin{equation}\label{eq:great d=3}
 \begin{split}
 	u\int_0^u s^{-2}  &\big|\{T_{P,M,\mu_\pm}<s\}\big| \,\d s  \\
 	&\le  2 \big|\{T_{P,M,\mu_\pm}<u\}\big| 
 		+u  2^{5/2}g^{2} |B^3_1| \ln(g/M+ \sqrt{1+(g/M)^2})
 \end{split}
 \end{equation}
 where we abbreviated $g=\sqrt{P^2+M^2}$ and $\tau =|P|/M$. Using \eqref{eq:volume-symbol-rel-pair} for 
 $\big|\{T_{P,M,\mu_\pm}<s\}\big|$, we have, similarly as 
 in the proof of Theorem \ref{thm:relativistic pair operator large d}, 
  \begin{align*}
 	\frac{2^3}{|B^3_1|}\int_0^u &s^{-2} 
 		\big|\{T_{P,M,\mu_\pm}<s\}\big| \,\d s 
 	 \\
 	& \le  \left( \frac{u^2+2ug+(1-4\wti{\mu}^2)M^2}{u^2+2ug+M^2} \right)^{3/2} 
 		\int_0^u \frac{s^{-1/2} (s+2g)^{3/2}(s+g)}{(s^2+2sg+M^2)^{1/2}}\, ds 
 \end{align*}
 and an integration by parts shows
 \begin{align*}
 	\int_0^u &\frac{s^{-1/2} (s+2g)^{3/2}(s+g)}{(s^2+2sg+M^2)^{1/2}}\, ds 
 	= \left[s^{-1/2} (s+2g)^{3/2}\left((s^2+2sg+M^2)^{1/2}-M\right)\right]_0^u \\
 	&\phantom{=~} 
 		-  \int_0^u \left( -\frac{1}{2}s^{-3/2} (s+2g)^{3/2} + s^{-1/2}\frac{3}{2}(s+2g)^{1/2} \right)\left((s^2+2sg+M^2)^{1/2}-M\right)\, ds \\
 	&= \frac{u^{1/2} (u+2g)^{5/2}}{(u^2+2ug+M^2)^{1/2}+M} 
 		+ \int_0^u s^{-3/2}(s+2g)^{1/2}\frac{(g-s)s(s+2g)}{(s^2+2sg+M^2)^{1/2}+M}\, ds \\
 	&\le \frac{2u^{1/2} (u+g)(u+2g)^{3/2}}{(u^2+2ug+M^2)^{1/2}} 
 		+ \int_0^g \frac{s^{-1/2}(s+2g)^{3/2}(g-s) }{(s^2+2sg+M^2)^{1/2}}\, ds \,.
 \end{align*}
 Since $(s+M^2/g)(s+g)\le s^2+2sg +M^2$, we have 
 \begin{align*}
 	\int_0^g &\frac{s^{-1/2}(s+2g)^{3/2}(g-s) }{(s^2+2sg+M^2)^{1/2}}\, ds 
 	\le \int_0^g \frac{s^{-1/2}(s+2g)^{3/2}(g-s) }{(s+g)^{1/2}(s+M^2/g)^{1/2}}\, ds \\
 	&\le 2^{3/2}g^{2} \int_0^g s^{-1/2}(s+M^2/g)^{-1/2}\, ds 
 		= 2^{5/2}g^{2} \int_0^{\sqrt{g}} (s^2+M^2/g)^{-1/2} \, ds \\
 	&= 2^{5/2}g^{2} \int_0^{g/M} (s^2+1)^{-1/2} \, ds 
 		= 2^{5/2}g^{2} \ln(g/M+ \sqrt{1+(g/M)^2})
 \end{align*}
where we also used that $s\mapsto \frac{(s+2g)^{3/2}(g-s) }{(s+g)^{1/2}}$ is decreasing on $[0,g]$. The last three bounds together with the trivial bound $\frac{u^2+2ug+(1-4\wti{\mu}^2)M^2}{u^2+2ug+M^2}\le 1$ show that \eqref{eq:great d=3} holds, which finishes the proof.
\end{proof}

\subsection{Ultra--relativistic pair operators}\label{sec:ultra-relativistic pair}
In the ultra--relativistic limit one takes the velocity of light to zero. Equivalently, one takes the limit of vanishing masses. The kinetic energy symbol of the pair operator becomes 
\begin{align}
	T_{P,0,\mu_\pm}(\eta) = \lim_{M\to 0} T_{P,M,\mu_\pm}(\eta) 
	=  |\mu_+P-\eta| +|\mu_-P+\eta|  - |P|  \, .
\end{align}
 The triangle inequality shows	$T_{P,0,\mu_\pm}(\eta)=  |\mu_+P-\eta| +|\mu_-P+\eta| -|P|\ge |P|- |P|=0$, so $T$ is positive\footnote{This also follows from the fact that $T_{P,0,\mu_\pm}$ is the limit of positive terms.}. 
 It was noted already in \cite{VugalterWeidl} that the kinetic energy $T_{P,0,\mu_\pm}$ is, for $P\neq 0$, zero on a `large set': 
  If $\eta$ is parallel to $P$ we can write it as $\eta= sP$ with $s\in\R $, and then, assuming also $-\mu_-\le s\le \mu_+$, one has  
  \begin{align*}
  	T_{P,0,\mu_\pm}(\eta)= |\mu_+P-sP| +|\mu_-P+sP| -|P| = (\mu_+-s)|P| + (\mu_-+s)|P| - |P|=0 \, ,
  \end{align*} 
 since $\mu_-+\mu_+=1$, so the kinetic energy symbol has the whole segment $[-\mu_-P,\mu_+P]$ as a zero set. This observation led to the conjecture in \cite{VugalterWeidl} that 
 the ultra-relativistic pair operator  $T_{P,0,\mu_\pm}+V$ should possess weakly coupled bound states in three dimensions whenever the total momentum does not vanish. Our next theorem confirms this.  
 \begin{theorem}\label{thm:ultra-relativistic d=3}
 	Assume that $d=3$, the total momentum $P\in\R^3\setminus\{0\}$, $V$ is relatively form compact with respect to $T_{P,0,\mu_\pm}$ and an attractive potential in the sense that $V\neq0$ and $\int V\, dx \le 0$. Then $T_{P,0,\mu_\pm} + V$ has at least one strictly negative eigenvalue and if, in addition, $V\le 0$ then it has  infinitely many strictly negative eigenvalues. 
 \end{theorem}
\begin{proof}
 We already convinced ourselves that $T_{P,0,\mu_\pm}$ is non-negative and zero on the line segment $[-\mu_-P,\mu_+P]$. 
  
  Now let $\eta= sP+\eta_\perp$, where $-\mu_-< s< \mu_+$ and $\eta_\perp $ is perpendicular to $P$. Then 
  \begin{align*}
  	T_{P,0,\mu_\pm}(\eta)&= \sqrt{(\mu_+-s)^2P^2+\eta_\perp^2} + \sqrt{(\mu_-+s)^2P^2+\eta_\perp^2}- |P| \\
  		&=  \frac{ \eta_\perp^2}{2(\mu_+-s)(\mu_-+s)|P|} + \Oh\left(\frac{|\eta_\perp|^4}{|\mu_\pm \mp s|^3|P|^3}\right)
  \end{align*} 
  using the Taylor expansion $\sqrt{1+x}= 1+\tfrac{x}{2}+ \Oh (x^2)$. 
  
  Since  the line segment $[-\mu_-P,\mu_+P]$ on which the kinetic energy vanishes has codimension 2 in $\R^3$, we can apply Theorem \ref{cor:submanifold} to conclude that $T_{P,0,\mu_\pm} + V$ has infinitely many negative eigenvalues. 
\end{proof}

In dimension $d\ge 4$, there is a useful bound on the number of bound states even in the ultra--relativistic limit. It has the interesting feature that even though the kinetic energy of the massless pair still remembers, 
through $\mu_\pm$, the ratio of the two masses before taking the limit of vanishing total mass, the semiclassical bound is independent of this. 
\begin{theorem}\label{thm:ultra-relativistic pair operator large d}	Let $d\ge 4$ and define 
	\begin{align*}
			G^d_{P,0}(u)\coloneqq \frac{(d-1)|B^d_1|}{(d-3)(4\pi)^d} u^{(d-1)/2}(u+|P|)\left( u +2|P| \right)^{(d-1)/2} .
	\end{align*}
 Then the number of negative bound states of the ultra--relativistic pair operator $T_{P,0,\mu_\pm}+V$ on $L^2(\R^d)$ is bounded by 
  \begin{equation} \label{eq:ultra relativistic pair operator large d}
 \begin{split}
 	N(T_{P,0,\mu_\pm}(q)+V) 
 	\leq 
 	\frac{\alpha^2}{(1-4\alpha)^2}
 	\int_{\mathbb{R}^{d}} G^d_{P,0}\left(\alpha^{-2}V_{-}(x)\right) \, dx,
 \end{split}
 \end{equation}
 for all $0<\alpha<1/4$.
\end{theorem}
\begin{remarks}
\begin{thmlist}
	\item Note  $G^d_{P,0}(u)= \frac{d-3}{d-1} \lim_{M\to0}G^d_{P,M,\wti{\mu}}(u)< \lim_{M\to0}G^d_{P,M,\wti{\mu}}(u)$. So Theorem \ref{thm:ultra-relativistic pair operator large d} improves upon the $M\to 0$ limit in Theorem \ref{thm:relativistic pair operator large d}.
	\item Since $G^d_{P,0}(u)\lesssim \min(u,1)^{(d-1)/2}+ \max(u,1)^d$, one needs $\min(V_-,1)\in L^{(d-1)/2}$ and $\max(V_-,1)\in L^d$ in order that the right hand side of \eqref{eq:ultra relativistic pair operator large d} is finite.
\end{thmlist}
\end{remarks}

\subsection{Relativistic pair operators: one heavy and one massless particle}\label{sec:relativistic pair mixed}
Considering a pair of relativistic particles in the center of mass frame when the  particles have \emph{very} different masses, say the first one is much heavier than the other, it makes sense to consider the idealized limit where $m_+=m$ is kept fixed, while $m_-\to 0$. In this case $\mu_+=m_+/(m_++m_-)\to 1$, $\mu_- m_-/(m_+m_-)\to 0$, and $\wti{\mu}=(\mu_--\mu_+)/2\to -1/2$, so the kinetic energy of the pair becomes  
\begin{align}
	T_{P,m,1,0}(\eta)\coloneqq \lim_{\mu_-\to0}T_{P,M,\mu_\pm}(\eta) =  \sqrt{|P-\eta|^2+m^2} +|\eta| 
		- \sqrt{P^2+m^2} \, .
\end{align}
In this case 
\begin{align}\label{eq:classical-phase-volume-heavy-light}
	\left| \{ T_{P,m,1,0}<u \} \right| &= \lim_{\mu_+\to 1} \left|\{ T_{P,M,\mu_\pm}<u\}\right| 
	=\frac{|B^d_1|}{2^d}\frac{u^{d/2}(u+g_{P,m} )\left( u +2g_{P,m}  \right)^{d}}{( u^2+2ug_{P,m}  +m^2 )^{(d+1)/2}} 
\end{align}
where we recall $g_{P,m} =\sqrt{P^2+m^2}$. We define
\begin{align}
	G^d_{P,m,1}(u) \coloneqq 
		\frac{|B^d_1|}{(4\pi)^d}\left\{ \frac{u^{d/2}(u+g_{P,m} )(u+2g_{P,m} )^d}{(u^2+2ug_{P,m} +m^2)^{(d+1)/2}} 
			+ 2^d\left(\frac{g_{P,m} }{m}\right)^{d}\frac{u^{d/2}(u+2g_{P,m} )}{(u^2+2ug_{P,m} +m^2)^{1/2}+m} 
				\right\}\, .
\end{align}
With this function we have 
\begin{theorem}\label{thm:relativistic pair heavy and light large d} 	
	For all $d\ge 2$ the number of negative bound states of the relativistic pair operator $T_{P,m,1,0}(q)+V$  on $L^2(\R^d)$, describing one heavy and one massless particle, is bounded by 
  \begin{equation} \label{eq:relativistic pair heavy and light large d}
 \begin{split}
 	N(T_{P,m,1,0}(q)+V) 
 	\leq 
 	\frac{\alpha^2}{(1-4\alpha)^2}
 	\int_{\mathbb{R}^{d}} G^d_{P,m,1}\left(\alpha^{-2}V_{-}(x)\right) \, dx,
 \end{split}
 \end{equation}
  for all $0<\alpha<1/4$.
\end{theorem}
\begin{proof} Using a by now familiar argument, Theorem \ref{thm:superduper2} yields a bound for $N(T_{P,m,1,0}(q)+V)$ with 
 \begin{align*}
 	G(u)
 	&=\frac{1}{(2\pi)^d} \left\{\big|\{T_{P,m,1,0}<u\}\big| + u\int_0^u s^{-2}
 		\big|\{T_{P,m,1,0}<s\}\big| \,\d s \right\}
 \end{align*}
 and we have \eqref{eq:classical-phase-volume-heavy-light} for $\big|\{T_{P,m,1,0}<u\}\big|$.  Since the map $0\le s\mapsto \frac{\left( s +2g_{P,m}  \right)^{2}}{( s^2+2sg_{P,m}  +m^2 )}$ is decreasing, 
 \begin{align*}
 	|B^d_1|^{-1}\int_0^u & s^{-2}
 		\big|\{T_{P,m,1,0}<s\}\big| \,\d s 
 		= \frac{1}{2^d} \int_0^u\frac{s^{\frac{d}{2}-2}(s+g_{P,m} )\left( s +2g_{P,m}  \right)^{d}}{( s^2+2sg_{P,m}  +m^2 )^{(d+1)/2}} \, ds \\
 		&\le u^{\frac{d}{2}-2}\left(\frac{g_{P,m} }{m}\right)^d 
 			\int_0^u\frac{s+g_{P,m} }{( s^2+2sg_{P,m}  +m^2 )^{1/2}} \, ds \\
 		&= u^{\frac{d}{2}-2}\left(\frac{g_{P,m} }{m}\right)^d \left( ( u^2+2ug_{P,m}  +m^2 )^{1/2}-m 
 			\right)\\
 		&= u^{\frac{d}{2}-1}\left(\frac{g_{P,m} }{m}\right)^d \frac{u+2g_{P,m} }{ ( u^2+2ug_{P,m}  +m^2 )^{1/2}+m} \, .
 \end{align*}
  So $G(u)\le G^d_{P,m,1}(u)$ for all $u\ge 0$. 
\end{proof}

\subsection{BCS type operators}\label{sec:BCS}
In the Bardeen-Cooper-Schrieffer theory of super-conductivity the symbol of the kinetic energy is given by the function 
\begin{align*}
  K_{\beta,\mu}(\eta) = (\eta^2-\mu)\frac{e^{\beta(\eta^2-\mu)}+1}{e^{\beta(\eta^2-\mu)}-1} 
\end{align*}
where $\beta=\tfrac{1}{T}$ is the inverse temperature and $\mu$ the chemical potential, i.e., the Fermi energy of the system. A-priori, the function 
$\R^d\ni\eta \mapsto K_{\beta,\mu}(\eta)$ is only defined for $\eta^2\neq\mu$, but we can extend it to $\eta^2=\mu$, by setting 
$K_{\beta,\mu}(\eta)=2\beta^{-1}$ whenever $\eta^2=\mu$. Extended in this way, $K_{\beta,\mu}$ is even $\calC^\infty(\R^d)$, see the proof of Theorem \ref{thm:BCS} below, and 
\begin{align*}
   \lim_{\eta\to\infty}K_{\beta,\mu}(\eta)=\infty, \, 	K_{\beta,\mu}(\eta)\ge 2\beta^{-1} \text{ for all }\eta\in\R^d\,\text{ and } K_{\beta,\mu}(\eta) =2\beta^{-1} \Leftrightarrow |\eta|=\sqrt{\mu}. 
\end{align*}
Hence $\sigma(K_{\beta,\mu}(p))=\sigma_\ess(K_{\beta,\mu}(p))=[2\beta^{-1},\infty)$.


The function $K_{\beta,\mu}$ decreases pointwise 
in $\beta>0$ and the limit of the kinetic energy as $\beta\to\infty$, i.e., the zero temperature limit, is given by 
\begin{align*}
  K_{\infty,\mu}(p) 	=|p^2-\mu|\, .
\end{align*}
In In BCS theory the operator
 \begin{align*}
  	K_{\beta,\mu}(p) + V, 
\end{align*}
  models the binding of Cooper pairs of electrons, where $V$ describes the interaction of electrons. The free kinetic energy is modified to the above form to take into account the filled Fermi sea and finite temperature effects, $\mu$ is the Fermi energy.
 
 The critical inverse temperature is given by 
 \begin{align}\label{eq:critical inverse temperature}
   T_{\mathrm{cr}}(V)^{-1}\coloneq\beta_{\mathrm{cr}}(V)\coloneq 
   		\sup\{\beta>0:\, \inf\sigma(K_{\beta,\mu}(p) + V)\ge 0\} \,. 	
 \end{align}
 It was shown in \cite{HainzlHamzaSeiringerSolovej} that the BCS gap equation has a non-trivial solution if $\beta>\beta_{\mathrm{cr}}(V)$ while for $0\le\beta\le \beta_{\mathrm{cr}}(V)$ it does not. Thus the phase transition from a normal state to the superconducting state is determined by $\beta_{\mathrm{cr}}(V)$ and there is a phase transition at positive temperature if and only if $\beta_{\mathrm{cr}}(V)>0$. Our method yields a painless simple criterion for it.  

\begin{theorem}\label{thm:BCS} Assume that $V\neq0$, $\int V\,\d x\le 0$, and 
the Fermi energy $\mu>0$. Then 
\begin{thmlist}
\item\label{thm:BCS1} The operator $|p^2-\mu|+V$ in $\R^d$, $d\ge 2$ has at least one  strictly negative eigenvalue and if, in addition, $V\le 0$, then it has infinitely many strictly negative eigenvalues. 
\item\label{thm:BCS2} For all $\beta>0$ the operator $K_{\beta,\mu}(p)+V$ has at least one  eigenvalue strictly below $2\beta^{-1}$ and if, in addition, $V\le 0$, then it has infinitely many  eigenvalues strictly below $2\beta^{-1}$.   Its ground state eigenvalue is strictly decreasing and becomes strictly negative for large enough $\beta>0$. 
\item\label{thm:BCS3} $\beta_{\mathrm{cr}}(V)$defined in \eqref{eq:critical inverse temperature} is finite, hence the critical temperature $T_{\mathrm{cr}}(V)>0$. Moreover, if $0<\beta<\beta_{\mathrm{cr}}$, then $\inf\sigma(K_{\beta,\mu}(p) + V)> 0$, and 
	if $\beta>\beta_{\mathrm{cr}}$, then $\inf\sigma(K_{\beta,\mu}(p) + V)< 0$. 
\end{thmlist}
\end{theorem}
\begin{remarks}
 \begin{thmlist}
  \item Our theorem shows that for arbitrary weak attractive interaction, even in the limiting case where $V\neq0$ but $\int V\,\d x=0$, Cooper pairs will always bind. This is a key ingredient for the BCS theory of superconductivity.  
  \item In \cite{FrankHainzlNabokoSeiringer} and \cite{HainzlSeiringer-degenerate} a criterion was proven which implies positivity of the critical temperature for the BCS model for potentials $\lambda V$ for arbitrary small coupling 
  $\lambda$ once a suitable integral operator has a strictly negative eigenvalue. However, their approach, modelled after the one of Simon \cite{Simon}, identifies a singular part of the Birman-Schwinger operator, which forces weakly eigenvalues to exists. But, in order that this `singular part' is not vanishing, this requires $\hatt{V}$ to be 
  non--vanishing  in a centered ball 
  $B_{2\sqrt{\mu}}(0)=\{\eta\in\R^d:\, |\eta|< 2\sqrt{\mu}\}$. On the other hand, our Theorem shows that the full Birman-Schwinger operator is singular, even for potentials whose Fourier transform vanishes on $B_{2\sqrt{\mu}}(0)$: Just take for $\hatt{V}$ any spherically symmetric Schwartz function which is supported on a centered annulus disjoint from  $ B_{2\sqrt{\mu}}(0)$. Then our Theorem   \ref{thm:BCS} shows the existence of weakly coupled bound states, whereas the criteria in \cite{FrankHainzlNabokoSeiringer} and \cite{HainzlSeiringer-degenerate} are not applicable. 
  \item According to Theorem \ref{thm:superduper1} the operator $|p^2-\mu|+V$ has  weakly coupled bound states for arbitrarily small $ V $ also in the one-dimensional case. However, in this case the number of negative eigenvalues is finite.
  \item One can generalize the results of Theorem \ref{thm:BCS} i) and consider the operator $|p^2-\mu|^{\gamma}+V$ for $ \gamma > 0 $ and $ d \geq 2. $ Theorem \ref{thm:superduper1} implies in this case that the operator has an infinite number of weakly coupled bound states for all $ \gamma \geq 1, $ independently of $ d \geq 2. $ On the other hand, for $ \gamma < 1 $ Theorem \ref{thm:superduper2} implies a quantitive bound on the number of negative eigenvalues.
 \end{thmlist}
\end{remarks} 

\begin{proof}
  Of course, part \eqref{thm:BCS3} follows from part \eqref{thm:BCS2}, since the supremum on the right hand side of \eqref{eq:critical inverse temperature} is finite once the lowest eigenvalue of
  $K_{\beta,\mu}(p)+V$ is strictly negative for large enough $\beta>0$. 
    
	For  part \eqref{thm:BCS1} we simply note that the zero set of $\eta\mapsto|\eta^2-\mu|$ 
	is equal to the centered sphere $S^{d-1}_{\sqrt{\mu}}$of radius $\sqrt{\mu}>0$ and 
	\begin{align*}
	  |\eta^2-\mu| 
	  = (|\eta|-\sqrt{\mu})(|\eta|+\sqrt{\mu}) \sim \dist(\eta, S^{d-1}_{\sqrt{\mu}})\, .	
	\end{align*}
 Since $S^{d-1}_{\sqrt{\mu}}$ has codimension $1$ in $\R^d$ Theorem \ref{cor:submanifold} applies. 
 
 Instead of using  Theorem \ref{cor:submanifold}, we could use Theorem \ref{thm:many}, since for any point $\omega\in S^{d-1}_{\sqrt{\mu}}$ one can easily see that 
 $\int_{B_\delta(\omega)} |\eta^2-\mu|^{-1}\, \d\eta =\infty$ for any $\delta>0$. 
 
 For the proof of part \eqref{thm:BCS2} consider the map $\R\ni a\mapsto a\frac{e^{\beta a}+1}{\e^{\beta a}-1}$, at first defined only for $a\neq0$. Since 
 \begin{align*}
   	a\frac{e^{\beta a}+1}{\e^{\beta a}-1} = \beta^{-1}\frac{e^{\beta a}+1}{\frac{\e^{\beta a}-1}{\beta a}} \to 2\beta^{-1} \text{ as } a\to 0, 
 \end{align*}
 we can extend this to all $a\in\R $ by continuity. Moreover, this map is clearly infinitely often differentiable for $a\neq0$, close to zero a short calculation reveals 
 \begin{align}\label{eq:near sphere}
   	a\frac{e^{\beta a}+1}{\e^{\beta a}-1} -2\beta^{-1}
   	= \frac{\beta a (e^{\beta a}+1)-2\frac{e^{\beta a}-1}{\beta a}}{\frac{\e^{\beta a}-1}{\beta a}} 
   	= \frac{1}{6\beta}\big( (\beta a)^2 + \Oh((\beta a)^3)\big)
 \end{align}
 and one sees that it is even infinitely often differentiable for all $a\in\R $ and growing linearly in $a$ for large $a$. Moreover, 
 \begin{align*}
   \frac{\partial}{\partial a }\left( a\frac{e^{\beta a}+1}{e^{\beta a}-1}\right)
   = \frac{1}{2}e^{-\beta a} \frac{\sinh(\beta)-\beta a}{(\sinh(\beta a))^2}
   = \left\{  \begin{array}{ccl}
  				>0&\text{for}&a>0\\
  				<0 &\text{for}& a<0 
  			\end{array} 
 \right.	\, ,
 \end{align*}
 so 
 \begin{align*}
   	K_{\beta,\mu}(\eta)\ge 2\beta^{-1} \text{ for all }\eta\in\R^d\,\text{ and } K_{\beta,\mu}(\eta) =2\beta^{-1} \Leftrightarrow |\eta|=\sqrt{\mu} \, , 
\end{align*}
 that is,  $\R^d\ni\eta\mapsto K_{\beta,\mu}(\eta)$ attains its minimal value $2\beta^{-1}$  exactly at the sphere $S^{d-1}_{\sqrt{\mu}}$. 
 
 Furthermore, 
 \begin{align*}
 \frac{\partial}{\partial \beta} a\frac{e^{\beta a}+1}{e^{\beta a}-1}
 =\left\{  \begin{array}{ccl}
  				-2\frac{ a^2e^{\beta a}}{(e^{\beta a}-1)^2} &\text{if}&a\neq0\\
  				-2\beta^{-2} &\text{if}& a=0 
  			\end{array}
 \right.
 <0 \, .
 \end{align*}
 So the symbol $K_{\beta,\mu}(\eta)$ is strictly decreasing in $\beta>0$. In particular, 
  $K_{\beta,\mu}(\eta)> |\eta^2-\mu|$ for all $\beta>0$ and $\eta\in\R^d$.
 
 The asymptotics \eqref{eq:near sphere} shows 
 \begin{align*}
   K_{\beta,\mu}(\eta)-2\beta^{-1} = \frac{\beta}{6} \dist(\eta, S^{d-1}_{\sqrt{\mu}})^2 
   + \Oh(\beta ^2 \dist(\eta, S^{d-1}_{\sqrt{\mu}})^3)	\, ,
 \end{align*}
 and again Corollary \ref{cor:submanifold}, or Theorem \ref{thm:many}, show that 
 $K_{\beta,\mu}(p)+V$ has infinitely many eigenvalues strictly below $2\beta^{-1}$. 
 
 Let $E_\beta \coloneq\inf\sigma(K_{\beta,\mu}(p)+V)$ be the ground state energy of 
 $K_{\beta,\mu}(p)+V$. We claim that it is strictly decreasing in $\beta>0$ and 
 $\lim_{\beta\to\infty}E_\beta = E_\infty=\inf(|P^2-\mu|+V)<0 $. 
 
 Let $\beta_2>\beta_1>0$ and let  $\psi_1$ be an eigenfunction of  $K_{\beta_1,\mu}(p)+V$ corresponding to  the ground state energy $E_{\beta_1}$. Then the variational principle 
 and the strict monotonicity of the symbol $K_{\beta,\mu}(\eta)$ in $\beta>0$ implies 
 \begin{align*}
   	E_{\beta_2} - E_{\beta_1} &\le \la \psi_1, K_{\beta_2,\mu}(p)+V)\psi_1 \ra - \la \psi_1, K_{\beta_1,\mu}(p)+V)\psi_1 \ra \\
   	&= \la \hatt{\psi}_1, (K_{\beta_2,\mu}(\cdot)- K_{\beta_1,\mu}(\cdot))\hatt{\psi}_1 \ra 
   	  < 0 \, ,
 \end{align*}
 so the ground state energy is strictly decreasing in $\beta>0$. Moreover this implies 
 $E_\beta> E_\infty\coloneq \inf\sigma(|p^2-\mu|+V)$, the ground state energy of 
 $|p^2-\mu|+V$ and, again by the variational principle, letting $\psi_\infty$ be a ground 
 state of  $|p^2-\mu|+V$, we have 
 \begin{align*}
   E_\beta \le \la \psi_\infty, (K_{\beta,\mu}+V)\psi_\infty \ra 	
   \to  \la \psi_\infty, (|p^2-\mu|+V)\psi_\infty \ra = E_\infty<0 \text{ as } \beta\to\infty.
 \end{align*}
 So $E_\beta$ decreases strictly to $E_\infty<0$.  In particular, there is a unique $\beta_{\mathrm{cr}}>0$ such that  $\inf\sigma(K_{\beta,\mu}(p) + V)> 0$ for all $0<\beta<\beta_{\mathrm{cr}}$  and  $\inf\sigma(K_{\beta,\mu}(p) + V)< 0$ for all 
 $\beta>\beta_{\mathrm{cr}}$.  
\end{proof}

\subsection{Discrete Schr\"odinger operators}\label{sec:discrete Schroedinger}
 We give the details of the method for discrete Schr\"odinger operators on 
 $\Z^d$. In principle, one can consider a general d-dimensional lattice. 
 One just has to use the dual lattice and adjust the notion of the discrete Fourier transform accordingly. 

On $l^2(\Z^d)$ we consider operators $T(p)$ similar to \eqref{eq:T} given by 
\begin{equation}
	T(p)\coloneq \calF^{-1} T \calF, 
\end{equation}
where for this section $\calF$ now denotes the discrete Fourier transform given by 
\begin{align*}
	\calF h(\eta) = \sum_{n\in\Z^d} e^{-i\eta n} h(n) 
\end{align*}
for $\eta\in \Gamma^*$, the d-dimensional Brillion zone\footnote{Here simply the torus} 
$\Gamma^*= [-\pi,\pi)^d$. The inverse Fourier transform is given by 
\begin{align*}
	\calF^{-1} g(n) = \int_{\Gamma^*} e^{i\eta n} g(\eta)\,\frac{\d\eta}{(2\pi)^d} \,,
\end{align*}
where $\d\eta/(2\pi)^d$ is the normalized Haar measure on the torus. 
A priori they are defined when $h\in l^1(\Z^d)$ and $g\in L^1(\Gamma^*)$, and it 
is well-known that $\calF$ extends to a unitary operator 
$\calF:l^2(\Z^d)\to L^2(\Gamma^*)$ with adjoint $\calF^{-1}$. 
We call the function $T$ admissible if $T(\eta)\in [0, T_\text{max}]$ for some finite $T_\text{max}$ and $\inf(\sigma(T(p))) = 0$. This is, for example, the case if $T$ is continuous with 
$\min_{\eta\in \Gamma^*}T(\eta)=0$ and $T_\text{max}\coloneq \max_{\eta\in\Gamma^*}T(\eta)$. In this case one even has $\sigma(T(p)) = [0, T_\text{max}]$.  

Our Theorems \ref{thm:superduper1} and \ref{thm:superduper2} easily extend to this setting, yielding
\begin{theorem}[Weakly coupled bound states for discrete Schr\"odinger operators]\label{thm:discrete1}
	Let $T:\Gamma^*\to [0,T_\text{max}]$ be measurable.  
	Assume that there exists a compact set $M\subset \Gamma^*$ such that 
	\begin{equation}\label{eq:discrete1-1}
		\int_{M_\delta} T(\eta)^{-1}\, \d\eta = \infty \text{ for all } \delta>0\, ,
	\end{equation}
	where $M_\delta\coloneq\{\eta\in \Gamma^*:\, \dist(\eta,M)\le \delta\}$. 
	Then   
	$\inf\sigma(T(p))=0$ and, if the potential $V\in l^1(\Z^d)$, also $\inf\sigma_\ess(T(p)+V)=0$. Moreover, if $V\not = 0$ with $\sum_{n\in\Z^d} V(n)\le 0$, then  
	the operator $T(p)+V$ has at least one strictly negative eigenvalue. 
	
	Assume that there exists a compact set $N\subset \Gamma^*$ such that 
	\begin{equation}\label{eq:discrete1-2}
		\int_{N_\delta} (T_\text{max}-T(\eta))^{-1}\, \d\eta = \infty \text{ for all } \delta>0\, ,
	\end{equation}
	where $N_\delta\coloneq\{\eta\in \Gamma^*:\, \dist(\eta,N)\le \delta\}$. 
	Then   
	$\sup\sigma(T(p))=T_\text{max}$, and, if the potential $V\in l^1(\Z^d)$, also 
	$\sup\sigma_\ess(T(p)+V)=T_\text{max}$. Moreover, if $V\not = 0$ with $\sum_{n\in\Z^d} V(n)\ge  0$, then 
	the operator $T(p)+V$ has at least one eigenvalue strictly greater than $T_\text{max}$.
\end{theorem}

\begin{remarks}
  \begin{thmlist}
	\item The proof of this theorem is a literal translation of the continuous case to the discrete setting. We leave the proof as an exercise to the interested reader. 
	\item In the case of the usual discrete Schr\"odinger operator $\Delta_d$ on $l^2(\Z^d)$ one has 
	$T(\eta)= \sum_{j=1}^d 2 (1-\cos(\eta_j)) = \sum_{j=1}^d 4\sin^2(\eta_j/2)$. 
	So in this case we can take $M=\{0\}$ and $N= \{(\pi,\ldots,\pi)^t\}$. Since $T$ behaves quadratically near $M$ and $N$, one sees that \eqref{eq:discrete1-1} and \eqref{eq:discrete1-2} hold in dimension $d\le 2$. 
	\item Our Theorem \ref{thm:many} and Corollaries \ref{cor:isolated-points} and \ref{cor:submanifold} have a natural counterpart in the discrete world with virtually the same proofs as in the continuous setting.  
  \end{thmlist}
\end{remarks}

Theorem \ref{thm:superduper2} has also a counterpart in the discrete setting:
\begin{theorem}\label{thm:discrete2}
	Let $T:\Gamma^*\to [0,T_\text{max}]$ be measurable,  
	\begin{align}
		G^-_\alpha(u)&\coloneq u\int_{T<u/\alpha^2}T(\eta)^{-1}\, \frac{\d\eta}{(2\pi)^d} \text{ for } u\ge 0\, , \label{eq:G+} \\
		\intertext{and}
		G^+_\alpha(u)&\coloneq u\int_{T_\text{max}-T<u/\alpha^2}\big(T_\text{max}-T(\eta)\big)^{-1}\, \frac{\d\eta}{(2\pi)^d} \text{ for } u\ge 0. \label{eq:G-}
	\end{align}
	Then 
	\itemthm\label{thm:discrete2-1} $G^-_\alpha(u)<\infty $ for all $\alpha,u>0$ 
		$\Leftrightarrow$ $  T^{-1}\id_{\{ T<\delta \}}\in L^1(\Gamma^*)$ for some $\delta>0$.\\
	\itemthm\label{thm:discrete2-2} $G^+_\alpha(u)<\infty $ for all $\alpha,u>0$ 
		$\Leftrightarrow$ $  (T_\text{max}-T)^{-1}\id_{\{ T_\text{max}-T<\delta \}}\in L^1(\Gamma^*)$ for some $\delta>0$.\\
	\itemthm\label{thm:discrete2-3} Assume that the potential $V$ is bounded and 
		$\inf\sigma_\ess(T(p)+V)\ge 0$. 
 	 Then  one has, for all $0<\alpha<\tfrac{1}{4}$, the bound 
			\begin{equation}\label{eq:CLRdiscrete1}
				N_-(T(p)+V)\le \frac{1}{(1-4\alpha)^2}\sum_{n \in \mathbb{Z}^{d}} G^-_\alpha(V_-(n)),
			\end{equation}
			where $V_-= \max(0,-V)$ is the negative part of $V$ and  $N_-(T(p)+V)$ is the number eigenvalues of $T(p)+V$ which are strictly negative.  
			
	Similarly, 	if $\sup\sigma_\ess(T(p)+V)\le T_\text{max}$,  then  one has, for all $0<\alpha<\tfrac{1}{4}$, the bound 
			\begin{equation}\label{eq:CLRdiscrete2}
				N_+(T(p)+V)\le \frac{1}{(1-4\alpha)^2}\sum_{n \in \mathbb{Z}^{d}} G^+_\alpha(V_+(n)),
			\end{equation}
			where $V_+= \max(0,V)$ is the positive  part of $V$ and  $N_+(T(p)+V)$ is the number of eigenvalues of $T(p)+V$ which are strictly above $T_\text{max}$.  
\end{theorem}

\begin{remarks}
  \begin{thmlist}
	\item If $G^\pm_\alpha(u)$ is finite for some $\alpha$ and $u$, then it is finite for all $\alpha,u>0$. Moreover, since the integration in \eqref{eq:discrete1-1} and \eqref{eq:discrete1-2} is over a subset of the compact set $\Gamma^*$, the functions $G^\pm_\alpha(u)$ behave \emph{linearly} in $u$ for $u$ large, once they are finite. We can improve this a little bit, see Corollary \ref{cor:discreteSchrodingerIproved}.
	\item As in the continuous case, one can reformulate the bounds on the discrete spectrum as 
		\begin{align}
			N_\pm (T(p)+V)
			&\le 
				(1-4\alpha)^{-2}\int_{0}^\infty N_\pm^{\text{cl}} (\max(\alpha^2,s)T +V)\,\d s  
		\end{align}
		with the semiclassical expressions
		\begin{align*}
			N_-^\text{cl} (T +V)
			&\coloneq \sum_{n\in\Z^d} \int_{\Gamma^*} \id_{\{T(\eta)+V<0\}} \frac{\d\eta}{(2\pi)^d}\\
			\intertext{and}
			N_+^\text{cl} (T +V)
			&\coloneq \sum_{n\in\Z^d} \int_{\Gamma^*} \id_{\{T(\eta)+V>T_\text{max}\}} \frac{\d\eta}{(2\pi)^d} \, .
		\end{align*}
  \end{thmlist}
\end{remarks}

The proof of Theorem \ref{thm:discrete2} is a straightforward adaptation of the proof in the continuous case. We leave it to the interested reader.

In the case of the usual discrete Schr\"odinger operator $ \Delta_{d} $ on $ l^{2}(\mathbb{Z}^{d}) $ Theorem \ref{thm:discrete2} gives us the following explicit bounds. 
	\begin{theorem} \label{thm:discreteSchrodinger}
		The numbers $ N_{-}(-\Delta_{d} + V) $ and $ N_{+}(-\Delta_{d} + V) $ of eigenvalues of $ -                          		\Delta_{d} + V $ below 0 and above $ 4d, $ respectively, satisfy for any $ \alpha \in (0,1/4) $
			\begin{align} \label{eq:discreteSchrodingerNegative}
				N_{-}(-\Delta_{d} + V) \leq (1 + 4\alpha)^{-2} \frac{|S^{d-1}|}{2^{2d}(d-2)\alpha^{d-2}} 				\sum_{n \in \mathbb{Z}^{d}} V_{-}(n)\min(V_{-}(n),4d\alpha^{2})^{d/2-1},
			\end{align}
		and
			\begin{align} \label{eq:discreteSchrodingerPositive}
				N_{-}(-\Delta_{d} + V) \leq (1 + 4\alpha)^{-2} \frac{|S^{d-1}|}{2^{2d}(d-2)\alpha^{d-2}} 				\sum_{n \in \mathbb{Z}^{d}} V_{+}(n)\min(V_{+}(n),4d\alpha^{2})^{d/2-1}.
			\end{align}
	\end{theorem}
	\begin{proof}
		 Using that 
			\begin{align*}
				\sin(x/2) \geq \frac{x}{\pi}
			\end{align*}
		for $ x \in [-\pi,\pi] $ we estimate
			\begin{align*}
				T(\eta) \geq \frac{4}{\pi^2} |\eta|^{2}.
			\end{align*}
		To get the first bound (\ref{eq:discreteSchrodingerNegative}) we have to estimate $ G^{-}_{\alpha}		(u). $ We have
			\begin{align} \label{eq:G-discrete}
				G^-_\alpha(u) &= u\int_{\{ \eta \in [-\pi,\pi]^{d} \mid \sum_{j=1}^d 4\sin^2(\eta_j/2)<u/				\alpha^2 \}}\bigl(\sum_{j=1}^d 4\sin^2(\eta_j/2)\bigr)^{-1}\, \frac{\d\eta}{(2\pi)^d} \\
				&\leq u \int_{|\eta| \leq \pi\min(\frac{u^{1/2}}{2\alpha},d^{1/2})} \frac{\pi^{2}}{4}|					\eta|^{-2}\, \frac{\d \eta}{(2\pi)^{d}} \notag \\
				&= \frac{|S^{d-1}|}{2^{2d}(d-2)\alpha^{d-2}} u \min(u,4d\alpha^{2})^{\frac{d}{2}-1}. \notag
			\end{align}
		Now, the claimed inequality follows immediately from (\ref{eq:CLRdiscrete1}).\\
		For the proof of the second bound (\ref{eq:discreteSchrodingerPositive}), we need to investigate 
			\begin{align} \label{eq:G+discrete}
				G^{+}_{\alpha}(u) &= u \int_{\{ \eta \in [0,2\pi] \mid 4d - \sum_{j=1}^d 4\sin^2(\eta_j/				2)<u/\alpha^2 \}}\bigl(4d - \sum_{j=1}^d 4\sin^2(\eta_j/2)\bigr)^{-1}\, \frac{\d\eta}					{(2\pi)^d}.
			\end{align}
		But by the change of variables $ \eta_{j} = \zeta_{j} + \pi, 1 \leq j \leq d, $ we see that (\ref{eq:G-discrete}) 		and (\ref{eq:G+discrete}) are exactly the same.
		Now, the claim follows from (\ref{eq:CLRdiscrete2}).
	\end{proof}
Copying a simple trick from \cite{MolchanovVainberg}, that exploits special properties of the discrete setting, we can improve our result a little bit.
	\begin{corollary} \label{cor:discreteSchrodingerIproved}
		The numbers $ N_{-}(-\Delta_{d} + V) $ and $ N_{+}(-\Delta_{d} + V) $ of eigenvalues of $ -                          		\Delta_{d} + V $ below 0 and above $ 4d, $ respectively, satisfy for any $ \alpha \in (0,1/4) $
			\begin{align} \label{eq:discreteSchrodingerNegativeMV}
				N_{-}(-\Delta_{d} + V) \leq (1 + 4\alpha)^{-2} \frac{|S^{d-1}|}{2^{2d}(d-2)\alpha^{d-2}} 				\sum_{V_{-} \leq 4d \alpha^{2}} V_{-}(n)^{d/2} + \#\{n \in \mathbb{Z}^{d} \mid V_{-}(n) > 				4d \alpha^{2}\},			
			\end{align}
		and
			\begin{align} \label{eq:discreteSchrodingerPositiveMV}
				N_{+}(-\Delta_{d} + V) \leq (1 + 4\alpha)^{-2} \frac{|S^{d-1}|}{2^{2d}(d-2)\alpha^{d-2}} 				\sum_{V_{+} \leq 4d \alpha^{2}} V_{+}(n)^{d/2} + \#\{n \in \mathbb{Z}^{d} \mid V_{+}(n) > 				4d \alpha^{2}\}.
			\end{align}
	\end{corollary}
	\begin{proof}
		We split the potential $ V = V_{+} - V_{-} = V_{+,1} + V_{+,2} - V_{-,1} - V_{-,2}, $ where
			\begin{align*}
				V_{\pm,1}(n) \leq h, \quad \text{and} \quad 
				V_{\pm,2}(n) > h
			\end{align*}
		for all $ n \in \mathbb{Z}^{d} $ and some $ h > 0. $ Furthermore, let
			\begin{align*}
				n^{\pm}(h,V) = \#\{n \in \mathbb{Z}^{d} \mid V_{\pm}(n) > h \}
			\end{align*}
		denote the number of values $ n \in \mathbb{Z}^{d} $ for which the positive and negative part of 		the potential, respectively, are greater than $ h. $ Assuming that $ V $ is suitably summable, 			both $ n^{+}(h,V) $ as well as $ n^{-}(h,V) $ are finite and $ V_{-,2} $ and $ V_{+,2} $ can be 		considered finite rank perturbations.\\
		Thus, it is for any $ \varepsilon > 0 $
			\begin{align*}
				N^{-}(-\Delta_{d} + V) &\leq N^{-}(-(1-\varepsilon)\Delta_{d} - V_{-,1}) + N^{-}(-  					\varepsilon\Delta_{d} - V_{-,2})\\
				&\leq N^{-}(-(1-\varepsilon)\Delta_{d} - V_{-,1}) + n^{-}(h,V),
			\end{align*}
		and
			\begin{align*}
				N^{+}(-\Delta_{d} + V) &\leq N^{+}(-(1-\varepsilon)\Delta_{d} + V_{+,1}) + N^{+}(-  					\varepsilon\Delta_{d} + V_{+,2})\\
				&\leq N^{+}(-(1-\varepsilon)\Delta_{d} + V_{+,1}) + n^{+}(h,V).
			\end{align*}
		Now, we choose $ h = 4d\alpha^{2} $ and apply Theorem \ref{thm:discreteSchrodinger} to $ -(1-		\varepsilon)\Delta_{d} - V_{-,1} $ and $ -(1-\varepsilon)\Delta_{d} + V_{+,1}. $ Since the 				resulting estimates are valid for any $ \varepsilon > 0, $ we pass to the limit $ \varepsilon 			\rightarrow 0 $ and end up with inequalities (\ref{eq:discreteSchrodingerNegativeMV}) and 				(\ref{eq:discreteSchrodingerPositiveMV}).
	\end{proof}
\begin{remark}
Of course, with virtually the same proof a version of Corollary \ref{cor:discreteSchrodingerIproved} holds also for more general kinetic energies $ T $ than just the discrete Laplacian.
\end{remark}

\section{Existence of bound states: a simple model case}\label{sec:model case}
We want to construct a test function $\varphi$ such that 
$\la\varphi, (T(p)+V) \varphi \ra<0 $. Once one has such a state together with $\sigma_\ess (T(p)+V)\subset [0,\infty)$, 
the Rayleigh--Ritz variational principle shows that there must be strictly negative discrete spectrum. The catch is, of course, how to guess such a variational state $\varphi$. 

To motivate our construction for our general set--up, we will discuss here the simple model case, where $T(\eta)= |\eta|^\gamma$, i.e., $T(p)= (-\Delta)^{\gamma/2}$ is a fractional Laplacian.  

\subsection{The case $\int V\,\d x<0$: Learning from failure.}\label{sec:integral negative}
We will work mainly in Fourier--space and, for simplicity, consider  
$\int V\, dx <0$ first.  
In order to make the kinetic energy small, a natural first guess for the test function would be 
\begin{align*}
  \wti{w}_\delta \coloneq \id_{A_\delta}	
\end{align*}
for a suitably chosen set $A_\delta$ of finite positive measure which concentrates around zero.  It turns out that this is not a good guess, however, and it is instructive to see why. 
We want that our 
testfunction approaches a constant, so its Fourier transform should approach a delta--function at zero. Thus we need to normalize $\wti{w}_\delta$ and are led to consider 
\begin{align*}
  \wti{\varphi}_\delta \coloneq \frac{\wti{w}_\delta}{\|\wti{w}_\delta\|_{L^1_\eta}} \, .	
\end{align*}
Note that this choice fulfills two crucial assumptions: Let 
\begin{align*}
  \kappa_\delta(x)\coloneq \calF^{-1}(\wti{\varphi}_\delta)(x)
  			=\frac{1}{(2\pi)^{d/2}}\int_{\R^d} e^{i\eta\cdot x} 
  			\wti{\varphi}_\delta(\eta)\, d\eta
\end{align*}
be the inverse Fourier transform of $\wti{\varphi}_\delta$. We always have 
\begin{align*}
  |\kappa_\delta(x)|\le \frac{\|\wti{\varphi}_\delta\|_{L^1_\eta}}{(2\pi)^{d/2}}	
  	= \frac{1}{(2\pi)^{d/2}}	
\end{align*}
by our normalization of $\wti{\varphi}_\delta$. Moreover, as long as $A_\delta$ concentrates  to the single point zero in a suitable way, we also have 
\begin{align*}
  \lim_{\delta\to 0} \kappa_\delta(x) = \frac{1}{(2\pi)^{d/2}}	
  	= \frac{1}{(2\pi)^{d/2}}	
\end{align*}
for all $x\in\R^d$. Since, by assumption the potential $V$ is integrable, we conclude with Lebesgue's dominated  convergence theorem 
\begin{align}
  \lim_{\delta\to 0}\la \kappa_\delta, V\kappa_\delta \ra	
  =
    \frac{1}{(2\pi)^d} \int V\, dx  \, .
\end{align}
So if $ \int V\, dx <0$, the choice for $\varphi_\delta$ yields a test function which makes the potential contribution strictly negative. 

It only remains to see whether the kinetic energy vanishes and, since the set $A_\delta$ concentrates near zero, this should be the case, but there is a catch: 
Note that 
\begin{align*}
	\la \kappa_\delta, (-\Delta)^{\gamma/2}\kappa_\delta \ra 
	 &= \la \wti{\varphi}_\delta, |\eta|^\gamma \wti{\varphi}_\delta \ra  
	 	= \frac{1}{\|\wti{w}_\delta\|_{L^1_\eta}^2} \int_{A_\delta} |\eta|^\gamma\,d\eta \, .
\end{align*}
Since $\|\wti{w}_\delta\|_{L^1_\eta}= |A_\delta|$, the Lebesgue measure of the set $A_\delta$, we can use rearrangement inequalities, see, e.g., \cite{LiebLoss},  to make the kinetic energy smallest by chosing 
$A_\delta$ to be centered ball of radius $\delta$, say. In this case $|A_\delta|\sim \delta^d$ and thus  
\begin{align}
	\la \kappa_\delta, (-\Delta)^{\gamma/2}\kappa_\delta \ra 
	 	\sim  \frac{1}{\delta^{2d}} \int_{0}^\delta r^{\gamma +d-1}\, dr
	 	\sim \delta^{\gamma-d}
\end{align}
and this goes to zero as $\delta\to 0$ \emph{only if}  $\gamma>d$ and it misses the critical case where $\gamma=d$. 
\smallskip

So we have to modify the testfunctions. A better choice, which also works for $\gamma=d$, 
turns out\footnote{We will further slightly generalize this ansatz in Section \ref{sec:local case}. } to be given by 
\begin{align}\label{eq:better guess}
  \hatt{w}_\delta(\eta) \coloneq |\eta|^{-\gamma}\,\id_{A_\delta}(\eta)	
\end{align}
where now the set $A_\delta$ has to stay away from zero to make  $\hatt{w}_\delta(\eta)$ normalizable. Note that $|\eta|^{-\gamma}$ is just the inverse of the symbol 
$T(\eta)=|\eta|^\gamma$. 

We further set, as before,  
\begin{align*}
	\hatt{\varphi}_\delta \coloneq \frac{\hatt{w}_\delta}{\|\hatt{w}_\delta\|_{L^1_\eta}} .
\end{align*}
With this choice  
\begin{align*}
  \la w_\delta, (-\Delta)^{\gamma/2} w_\delta\ra 
  		= 	\la \hatt{w}_\delta, |\eta|^\gamma \hatt{w}_\delta \ra
  		= \int_{A_\delta} |\eta|^{-\gamma}\, d\eta = \|\hatt{w}_\delta\|_{L^1_\eta},
\end{align*}
hence 
\begin{align}\label{eq:kinetic energy model case}
	\la \varphi_\delta, (-\Delta)^{\gamma/2}\varphi_\delta \ra
		= \frac{1}{\|\hatt{w}_\delta\|_{L^1_\eta}}\, .
\end{align}
As before, we still have 
\begin{align*}
	\la \varphi_\delta, V\varphi_\delta  \ra \to \frac{1}{(2\pi)^d} \int V\, dx 
	\quad \text{as } \delta\to 0 
\end{align*}
as soon as $A_\delta$ concentrates near a point in the limit $\delta\to 0$.  Since  the function $\R^d\ni\eta\mapsto |\eta|^{-\gamma}$ has a \emph{non-integrable singularity} near zero for $\gamma \ge d$, we can make $A_\delta$ concentrate near zero, thus having 
$\|\hatt{w}_\delta\|_{L^1_\eta}$ blow up and,  because of \eqref{eq:kinetic energy model case}, we then have  
$\lim_{\delta\to 0} \la \varphi_\delta, (-\Delta)^{\gamma/2}\varphi_\delta \ra =0$. 

Explicitly, choosing $A_\delta$ to be the annulus 
\begin{align*}
  A_\delta\coloneq \{ r_{1,\delta}<|\eta|<r_{2,\delta}\}	
\end{align*}
we have 
\begin{align*}
  	\|\hatt{w}_\delta\|_{L^1_\eta} 
  		\sim \int_{r_{1,\delta}}^{r_{2,\delta}} r^\gamma r^{d-1}\, dr 
  		= \left\{
  			\begin{array}{ccl}
  				\ln(\frac{r_{2,\delta}}{r_{1\,\delta}}) &\text{if}&\gamma=d\\
  				\frac{1}{\gamma-d}[r_{1,d}^{-(\gamma-d)}- r_{2,\delta}^{-(\gamma-d)}] &\text{if}& \gamma>d
  			\end{array}
  		  \right.
\end{align*}
and choosing $r_{1,\delta}=\delta^2$ and $r_{2,\delta}= \delta $ we see 
$\lim_{\delta\to 0} \|\hatt{w}_\delta\|_{L^1_\eta} =\infty $. 
With  \eqref{eq:kinetic energy model case}
\begin{align*}
	\lim_{\delta\to 0}\la  \varphi_\delta, ((-\Delta)^{\gamma/2}+V) \varphi_\delta\ra 
	= \frac{1}{(2\pi)^d} \int V\, dx ,
\end{align*} 
follows. So bound states with strictly negative energy exist once $\int V\, dx<0$.  

\subsection{The case of $\int V\, dx =0$.}\label{sec:integral zero}
To include the case where $V$ does not vanish identically but $\int V\, dx=0$, we have to further modify the testfunction. Second order perturbation theory suggest that the testfunction should be modified by adding a suitable multiple of the potential $V$. 
This suggest the ansatz 
\begin{align}
  \varphi_\delta +\alpha \phi 	
\end{align}
 for some $\alpha\in \R $ and a suitably nice function $\phi$, to be determined  later,  as a trial state for the computation of the energy. Using this  we get, with 
 $T(p)=|p|^\gamma= (-\Delta)^{\gamma/2}$,  
\begin{align*}
	E(\delta,\alpha)&\coloneq \la \varphi_\delta +\alpha \phi, (T(p)+V)(\varphi_\delta +\alpha \phi) \ra \\
	&= \la \varphi_\delta, T(p)\varphi_\delta \ra + \la \varphi_\delta, V\varphi_\delta \ra 
		+ 2\alpha\re(\la \varphi_\delta, T(p)\phi \ra) 
		+ 2\alpha\re(\la \varphi_\delta, V\phi \ra) \\
	&\phantom{=~} + \alpha^2 \la \phi, T(p)\phi \ra 
		+ \alpha^2 \la \phi, V\phi \ra \, .
\end{align*}
From the discussion above we know that 
\begin{equation}
	\begin{split}
	\lim_{\delta\to 0} \la \varphi_\delta, T(p)\varphi_\delta \ra 
	&= 0,  \\
	\lim_{\delta\to 0}  \la \varphi_\delta, V\varphi_\delta \ra 
	&= \frac{1}{(2\pi)^d} \int V\,dx =0,   \\	
	\lim_{\delta\to 0}  \la \varphi_\delta, V\phi \ra 
	&= \frac{1}{(2\pi)^{d/2}} \int V\phi\,dx \, , 
	\end{split}
\end{equation}
and, since $T(p)$ is a positive operator, we also have 
\begin{align}
	|\la \varphi_\delta, T(p)\phi \ra|
		\le \la \varphi_\delta, T(p)\varphi_\delta \ra^{1/2} \la \phi, T(p)\phi \ra^{1/2} 
	\to 0 \quad\text{as }\delta\to 0. 
\end{align}
Thus 
\begin{align*}
	E(\alpha)\coloneq \lim_{\delta\to 0} E(\delta,\alpha)
	 = 2\alpha\frac{1}{(2\pi)^{d/2}} \re \int V\phi\,dx + \alpha^2 \la \phi, T(p)\phi \ra 
		+ \alpha^2 \la \phi, V\phi \ra
\end{align*}
and 
\begin{align*}
	\lim_{\alpha\to 0} \frac{E(\alpha)}{\alpha} 
	 = \frac{2}{(2\pi)^{d/2}} \re \int V\phi\,dx \, .
\end{align*}
This shows that we will have $E(\delta,\alpha)<0$ for some finite $\delta>0$ and $\alpha>0$, if we can find a Schwarz function $\phi$ such that $\int V\phi\,dx <0$. 

Split $V= V_+-V_-$, the positive and negative parts of $V$. By assumption,  $\int V_-\, dx = \int V_+\, dx>0$.  Take a big centered ball $B$ such that $\int_B V_-\, dx >0$ and consider the set 
\begin{align*}
	D\coloneq B\cap \{ V_->0 \} \, .
\end{align*} 
Let $\kappa_\veps\in\calC^\infty_0(\R^d)$ be an approximate delta--function and set 
\begin{align*}
	\phi_\veps\coloneq \kappa_\veps*\id_{D} \, .
\end{align*}
This is a nice infinitely often differentiable function with compact support, hence it is in the form domains of $T(p)$ and $V$, and by the properties of convolutions \cite{LiebLoss} we have $0\le \phi_\veps\le 1$ and $\phi_\veps\to \id_{D}$ in $L^1$ for $\veps\to 0$,  hence, after taking a subsequence, also pointwise almost everywhere. 
With slight abuse of notation we denote this subsequence still by $\phi_\veps$. 
With the help of Lebesgue's dominated convergence theorem one sees 
\begin{align*}
	\lim_{\veps\to 0} \int V\phi_\veps\,dx = -\int_B V_-\, dx <0  \, ,
\end{align*}
so using $\phi_\veps$ instead of $\phi$ for some small enough $\veps>0$ in the above argument shows that there are $\alpha,\delta,\veps>0$ such that 
\begin{align*}
	\la \varphi_\delta +\alpha \phi_\veps, (T(p)+V)(\varphi_\delta +\alpha \phi_\veps) \ra <0.  
\end{align*}
Hence the variational principle shows that we have a strictly negative eigenvalue of $T(p)+V$ also in the case where $V$ does not vanish identically but $\int V\, dx=0$.

\section{Existence of bound states: proof of the general case}\label{sec:local case}
In this section we give the proof of Theorems \ref{thm:superduper1} and \ref{thm:many} and Colloraries \ref{cor:isolated-points} and \ref{cor:submanifold}. To prepare for this, 
we give a lemma first, which is a convenient replacement for \eqref{eq:kinetic energy model case}. We have to be a little bit careful here, to ensure that the constructed function is normalizable. The construction in \eqref{eq:better guess} worked, because there the kinetic energy was bounded away from zero in any open set not containing zero. In the general case, where $T$ is just measurable, this is not so clear. As an easy way out, we simply  cut the kinetic  energy close to zero.
\begin{lemma}\label{lem:kinetic energy bound 1} Let $T:\R^d\to [0,\infty)$ and $0\le \chi\le 1$ be measurable. For $\tau>0$ define the function $\hatt{w}_\tau$ by 
  \begin{align*}
    \hatt{w}_\tau(\eta)= \max(T(\eta),\tau)^{-1} \chi(\eta) \quad \text{for }\eta\in\R^d. 	
  \end{align*}
  Then $w_\tau\coloneq\calF^{-1}(\hatt{w}_\tau)\in \calQ(T(p))$, the quadratic form domain of $T(p)$,  and we have the following bound for its kinetic energy,  
  \begin{align*}
  	\la w_\tau, T(p) w_\tau \ra \le  \|\hatt{w}_\tau\|_{L^1_\eta} \, .
  \end{align*}
\end{lemma}
\begin{remark}
	At first sight the bound provided by Lemma \ref{lem:kinetic energy bound 1} seems surprising, since the left hand side of the bound scales quadratically in $w$ but the right hand side is linear in $\hatt{w}$. This is not a contradiction, though, since we assume that $\hatt{w}_\tau= \max(T,\tau)^{-1}\chi$ and $0\le \chi\le 1$, which breaks the scaling. 
\end{remark}
\begin{proof}[Proof of Lemma \ref{lem:kinetic energy bound 1}:] 
 This is a simple calculation. Since $T$ is positive and by Plancherel,  
 \begin{align*}
   	\la w_\tau, T(p) w_\tau \ra 
   	 &=\la  \sqrt{T(p)}w_\tau, \sqrt{T(p)} w_\tau\ra
   	  = \la \hatt{w}_\tau, T \hatt{w}_\tau \ra 
   	  = \int_{\R^d} T(\eta) \max(T(\eta),\tau)^{-2}\chi(\eta)^2\, d\eta \\
   	 &\le  \int_{\R^d} \max(T(\eta),\tau)^{-1}\chi(\eta)^2\, d\eta 
   	  \le \int_{\R^d} \max(T(\eta),\tau)^{-1}\chi(\eta)\, d\eta
   	  = \|\hatt{w}_\tau\|_{L^1_\eta} \, , 
 \end{align*}
 since, by assumption $0\le\chi\le 1$, thus also  $0\le \chi^2\le \chi$. 
\end{proof}
Now we come to the 
\begin{proof}[Proof of Theorem \ref{thm:one}:] 
  Let $\omega\in \R^d$ be such that 
  \begin{align*}
    \int_{B_{1/n}(\omega)} T(\eta)^{-1}\, d\eta = \infty 
  \end{align*}
  for every $n\in\N$. By monotone convergence, we have 
  \begin{equation*}
  	\lim_{\tau\to 0} \int_{B_{1/n}(\omega)} \max(T(\eta),\tau)^{-1} \, \d\eta 
  	=
  	\int_{B_{1/n}(\omega)} T(\eta)^{-1}\, d\eta = \infty \, ,
  \end{equation*}
  so there exists a sequence $\tau_{n+1}<\tau_n\to 0$, for $n\to\infty$, with 
  \begin{equation}\label{eq:unbounded} 
  	\lim_{n\to\infty} \int_{B_{1/n}(\omega)} \max(T(\eta),\tau_n)^{-1}\,\d\eta = \infty\, .
  \end{equation}
 Define the functions $\hatt{w}_n$ and $\hatt{\varphi}_n$ by  
  \begin{align*}
    \hatt{w}_n(\eta)\coloneq \max(T(\eta),\tau_n)^{-1} \id_{B_{1/n}(\eta)} \text{ and }  
    \hatt{\varphi}_n(\eta) \coloneq \frac{\hatt{w}_n(\eta)}{\|\hatt{w}_n\|_{L^1_\eta}}	
  \end{align*}
  for every $\eta\in \R^d$. Note that $\hatt{w}_n\in L^1_\eta$, so $\varphi_n$ is non-trivial.  
  Because of Lemma \ref{lem:kinetic energy bound 1}, $w_n\in\calQ(T(p))\subset\calQ(V)$. 
  By construction,  
  \begin{align*}
    \|\hatt{w}_n\|_{L^1_\eta} \to \infty \text{ as }n\to\infty\, .
  \end{align*}
  In addition, since the sets $B_{1/n}(\omega)$ concentrate around $\omega$ and 
  $\hatt{\varphi}_n$ is $L^1$ normalized, we also have that $\hatt{\varphi}_n$ is a sequence of approximate delta--functions which concentrates  at $\omega$. 
  Thus we have the uniform bound 
  \begin{align*}
    |\varphi_n(x)|\le \frac{1}{(2\pi)^{d/2}}\|\hatt{\varphi}\|_{L^1_\eta} = \frac{1}{(2\pi)^{d/2}}	
  \end{align*}
  and  the pointwise limit 
  \begin{align*}
    \lim_{n\to\infty} \varphi_n(x) = \frac{1}{(2\pi)^{d/2}} e^{i\omega\cdot x}	\text{ for all } x\in\R^d \, .
  \end{align*}
  Using Lebesgue's dominated convergence theorem this shows  
  \begin{align*}
  	\lim_{n\to\infty} \la \varphi_n, V\varphi_n \ra = \frac{1}{(2\pi)^d} \int_{\R^d} V\, dx\, .
  \end{align*}
  For the kinetic energy we simply note that Lemma \ref{lem:kinetic energy bound 1} yields 
  \begin{align*}
    \la \varphi_n, T(p)\varphi_n \ra 
    = \frac{1}{\|\hatt{w}_n\|_{L^1_\eta}^2}	\la w_n, T(p)w_n \ra 
	\le  \frac{1}{\|\hatt{w}_n\|_{L^1_\eta}} \to 0 \text{ as } n\to\infty\, .
  \end{align*}
  So if $\int V\, dx<0$ we can immediately conclude 
  \begin{align*}
  \lim_{n\to\infty} 	\la \varphi_n, (T(p)+V)\varphi_n \ra 
    = \frac{1}{(2\pi)^d} \int_{\R^d} V\, dx <0
  \end{align*}
  and the variational principle implies that there is a 
  strictly negative eigenvalue of $T(p)+V$.  
  \smallskip
  
  In the case that $\int V_+\, dx = \int V_-\, dx >0 $, so $\int V\, dx =0$, we use the construction of Section \ref{sec:integral zero} to see that there is a positive  function  $\phi\in\calC^\infty_0(\R^d)$ with 
  \begin{align*}
    \int_{\R^d} V\phi \, dx <0  \, . 	
  \end{align*}
  Similarly to the discussion in Section \ref{sec:integral zero}, we modify the trial state in the form 
  \begin{align*}
  \varphi(x)= \varphi_n(x) + \alpha e^{i\omega\cdot x} \phi(x) \text{ for } x\in\R^d. 	
  \end{align*}
  Setting $\wti{\phi}(x)= e^{i\omega\cdot x} \phi (x)$ we have, analogously to the calculation in Section \ref{sec:integral zero}, 
  \begin{align*}
    \lim_{\alpha\to 0 } \alpha^{-1} &\lim_{n\to\infty} 
    \la \varphi_n + \alpha \wti{\phi}, (T(p)+V)(\varphi_n + \alpha \wti{\phi})\ra	
    = \frac{2}{(2\pi)^{d/2}} \re \int e^{-i\omega\cdot x}V(x)\wti{\phi}(x)\,dx \\
    &= \frac{2}{(2\pi)^{d/2}}  \int V(x)\phi(x)\,dx <0 \, .
  \end{align*}
 So for all large enough $n\in\N$ and small enough $\alpha>0$ 
 \begin{align*}
   	\la \varphi_n + \alpha \wti{\phi}, (T(p)+V)(\varphi_n + \alpha \wti{\phi})\ra <0\, ,
 \end{align*}
 which, by the variational principle, implies the existence of at least 
 one negative eigenvalue for $T(p)+V$. 
\end{proof}
 
 \smallskip
\begin{proof}[Proof of Theorem \ref{thm:many}: ]  We will first prove part \ref{thm:many-b}. 
  Assume that there are $k$ distinct points $\omega_1,\ldots,\omega_k$ such that 
   \begin{align}\label{eq:infinity ball 2}
    \int_{B_\delta(\omega_l)} T(\eta)^{-1}\, d\eta = \infty 
  \end{align}
  for all $l=1,\ldots,k$ and all $\delta>0$. Using the previous construction, we see 
  that for each $l=1,\ldots,k$ there are functions $\varphi_{l,n}$ where the support of $\hatt{\varphi}_{l,n}$ concentrates in Fourier space near 
  $\omega_l$. Then 
  \begin{align*}
    & \lim_{n\to\infty} \la \sum_{l=1}^k c_l\varphi_{l,n}, (T(p)+V)(\sum_{l=1}^k c_l\varphi_{l,n})\ra 
    	= \lim_{n\to\infty} \sum_{r,s=1}^k \ol{c_r} c_s   \la \varphi_{r,n}, (T(p)+V)\varphi_{s,n}\ra \\
    	 &= \frac{1}{(2\pi)^d} \int_{\R^d} V(x) e^{-ix(\omega_r-\omega_s)}\, \d x 
    	 	= \frac{1}{(2\pi)^{d/2}} \sum_{r,s=1}^k\hatt{V}(\omega_r-\omega_s) \ol{c_r}c_s
    	 	= \frac{1}{(2\pi)^{d/2}} \la c,Mc \ra_{\C^k}
  \end{align*}
  with the Matrix $M= (\hatt{V}(\omega_r-\omega_s))_{r,s=1,\ldots,k}$. 
  If this matrix is strictly negative definite, then  
  $\la c,Mc \ra_{\C^k} <0$ for all $c\neq 0$, thus $T(p)+V$ will be strictly  negative on the subspace $N_{k,n}=\text{span}(\varphi_{r,n}, r=1,\ldots,k)$.
  For large $n$ the $\varphi_{r,n}$ do not overlap in Fourier-space, thus 
  $\dim N_{k,n}=k$ for all large enough $n$ and this gives the existence of at least $k$ strictly negative eigenvalues of $T(p)+V$ by the usual variational arguments, see \cite{BerezinShubin, BirmanSolomyak, Simon-course2}.  
  \smallskip
  
  To prove part \ref{thm:many-a}, we simply note that if $V\le 0$ and $V\neq 0$, 
  then we can multiply it with a suitable cutoff function  $0\le \psi\le 1$, setting  
  $\wti{V}= \psi V$ to get $\wti{V}\le 0$ and $\wti{V}\neq 0$. By the min-max principle, the eigenvalues of $T(p)+V$ are bounded from above by the eigenvalues of $T(p)+\wti{V}$, so without loss of generality, we can assume that  $V\in L^1$. 
  
  So assume that $V\in L^1$ with $V\le 0$ and $v\neq 0$. Then the matrix $M$ above will be strictlty negative definite. Indeed, for $|c|^2=\sum_{r=1}^k |c_r|^2=1$, we have 
  \begin{align*}
   	\la c,Mc \ra_{\C^k} 
   		=  \frac{1}{(2\pi)^{d/2}} \int_{\R^d} V(x)\,  \big|\sum_{r=1}^k c_r e^{-ix\omega_r}\big|^2\, \d x <0 
  \end{align*}
 since $x\mapsto\sum_{r=1}^k c_r e^{-ix\omega_r}$ is analytic and thus not zero on any open set of positive Lebesgue measure and $V\le 0$, $V\neq 0$.  Moreover, let $S^{k-1}$ be the unit sphere in $\C^k$ and note that the map 
   $S^{k-1}\ni c\mapsto \la c,Mc\ra_{\C^k}$ is continuous and $M$ is a hermitian Matrix. Thus by the above, we see that the largest eigenvalue of $M$ is strictly negative.  So $M$ is strictly negative definite and by part \ref{thm:many-b}, we conclude that $T(p)+V$ has at least $k$ strictly negative eigenvalues. 
  \smallskip
  
  To prove part \ref{thm:many-c} let $M$ be the $k\times k$ matrix as above and let $a=(a_1, \ldots,a_k)^t$ be the normalized eigenvector corresponding to the eigenvalue zero, which we assume is non-degenerate. Let $V_{k-1}$ be the $k-1$ dimensional orthogonal complement of $a$ in $\C^k$. Furthermore, we set   
  \begin{align*}
  	N_{k-1,n}\coloneq \left\{ \sum_{r=1}^k c_r\varphi_{r,n}:\, c\in V_{k-1} \right\}\, .  
  \end{align*}
  Note that $N_{k-1,n}$ is a $k-1$-dimensional subspace of $L^2(\R^d)$ for large enough $n$ since then $\varphi_{r,m}$, $r=1,\ldots,k$, have disjoint support in Fourier space, and  
  \begin{align*}
  	\wti{\varphi}\coloneq \sum_{r=1}^k a_k \varphi_{r,n} + \alpha \phi \, . 
  \end{align*}
 Furthermore, put $N_{k,\alpha,n}\coloneq \text{span}\{ N_{k-1,n}, \wti{\varphi}\}$. Our goal is to show that the dimension of $N_{k,n}$ is $k$ and that $T(p)+V$ is strictly negative on $N_{k,n}\setminus\{0\}$ for large enough $n$ and a suitable choice for $\phi$ and $\alpha\in \R$. 
 Any vector in $ N_{k,n}$ can be written as $\psi_{n,\alpha}= \sum_{r=1}^k c_r\varphi_{r,n} + \alpha \phi$ with $c\in \C^k$. We have  
 \begin{align*}
 	\la \psi_{n,\alpha}, (T(p)+V) \psi_{n,\alpha}\ra 
 		&= \sum_{r,s=1}^k \ol{c_r} c_s \la\varphi_{r,n}, (T(p)+V) \varphi_{s,n}  \ra + 2\alpha \sum_{r=1}^k\re\left( c_r\la \phi, (T(p) +V) \varphi_{r,n} \ra \right) \\
 		& \phantom{= } \, + \alpha^2 \la \Phi,(T(p) +V)\Phi\ra \, ,
 \end{align*}
 so 
 \begin{align*}
 	\lim_{n\to\infty} \la & \psi_{n,\alpha}, (T(p) +V) \psi_{n,\alpha}\ra \\
 		&= \la c, M c  \ra_{\C^k} + \frac{2\alpha}{(2\pi)^{d/2}} \re \int_{\R^d} \ol{\phi(x)} V(x) \sum_{r=1}^k c_r\,e^{i\omega_r x}\, \d x 
 		  + \alpha^2 \la \Phi,(T(p) +V)\Phi\ra \\
 		&\le \frac{2\alpha}{(2\pi)^{d/2}} \re \int_{\R^d} \ol{\phi(x)} V(x) \sum_{r=1}^k c_r\,e^{i\omega_r x}\, \d x 
 		 +  \alpha^2 \la \Phi,(T(p) +V)\Phi\ra \,
 \end{align*}
 since $M$ is negative semi-definite. Thus  
 \begin{align}\label{eq:punchline2}
 	\lim_{\alpha \to 0+}\alpha^{-1}\lim_{n\to\infty} \la \psi_{n,\alpha}, (T(p)+V) \psi_{n,\alpha}\ra 
 		&\le  \frac{2}{(2\pi)^{d/2}} \re \int_{\R^d}  \ol{\phi(x)} V(x) \sum_{r=1}^k c_r\,e^{i\omega_r x}\, \d x \, .
 \end{align}
 Note that the convergence on the left hand side of \eqref{eq:punchline2}  
 in $\alpha$ and $n$  is \emph{uniform} for $c$ in bounded subsets of $\C^k$. 
 
 Set $\phi_0(x)=\min(V_-(x),1)\sum_{r=1}^k \, e^{i\omega_r x}$. Then 
 \begin{align*}
 	\re \int_{\R^d}  \ol{\phi_0(x)} V(x) \sum_{r=1}^k c_r\,e^{i\omega_r x}
 	\, \d x 
 		= -\int_{\R^d} \min(V_-(x)^2,1) \big|\sum_{r=1}^k c_r\,e^{i\omega_r x}\big|^2 \, \d x <0, 
 \end{align*}
 since $x\mapsto \sum_{r=1}^k c_r\,e^{i\omega_r x}$ is not zero on sets of positive Lebesgue measure, by analyticity, and $V_-\neq 0$ with $\int V_- \, \d x>0$, by assumption. Cutting off $\phi_0$ outside of a large enough ball and mollifying it with a smooth approximate delta-function, one gets $\phi_1\in\calC^\infty_0(\R^d)$ such that 
 $ \re \int_{\R^d} \big( \ol{\phi_1(x)} V(x) \sum_{r=1}^k c_r\,e^{i\omega_r x}\big)\, \d x <0$. 
 
 Now let $0\le \kappa\le 1$ be a smooth cutoff function with $\kappa(\eta)=1$ for $|\eta|\le 1$ and $\kappa(\eta)=0$ for $|\eta|\ge 2$ and set 
 $\kappa_m(\eta)= \kappa(\eta/m)$ and 
 $\hatt{\phi_{2,m}}= \hatt{\phi_1}\kappa_m$. Since $\hatt{\phi_1}$ is a Schwartz function, $\hatt{\phi_1}\kappa_m$ converges in $L^1$ to $\hatt{\phi_1}$ as $m\to\infty$, and thus $\phi_{2,m}$ converges to $\phi_1$ uniformly as $m\to\infty$. In particular, since $V\in L^1(\R^d)$, 
 \begin{align*}
 	\lim_{m\to\infty} \re \int_{\R^d}  \ol{\phi_{2,m}(x)} V(x) \sum_{r=1}^k c_r\,e^{i\omega_r x}\, \d x 
 	=
 	\re \int_{\R^d}  \ol{\phi_1(x)} V(x) \sum_{r=1}^k c_r\,e^{i\omega_r x}\, \d x < 0 
 \end{align*} 
 so setting $\phi_{2}= \phi_{2,m}$ for some large enough $m$ gives 
  $\re \int_{\R^d}  \ol{\phi_{2}(x)} V(x) \sum_{r=1}^k c_r\,e^{i\omega_r x}\, \d x <0$. 

 Now we come to the final step in the construction of a suitable $\phi$:  
 Recall that each $\hatt{\varphi}_{r,n}$ has support inside the ball $B_{1/n}(\omega_r)$. Take a smooth  cutoff function $0\le \theta_n\le 1$ such that $\theta_n=1$ on $B_{1/n}(\omega_r)$ for each $r=1,\ldots,k$ and $\theta_n=0$ on $\bigcup_{r=1}^k B_{2/n}(\omega_r)^c$, which is possible at least for large enough $n$. 
 
 Set $\hatt{\phi_{3,m}} = \hatt{\phi_2}(1-\theta_m)$. Note that  $\theta_m $ converges to zero in $L^1(\R^d)$ for $m\to\infty$, so again $\phi_{3,m}$ converges to $\phi_2$ uniformly, and thus 
 \begin{align*}
 	\lim_{m\to\infty} \re \int_{\R^d} \ol{\phi_{3,m}(x)} V(x) \sum_{r=1}^k c_r\,e^{i\omega_r x} \, \d x 
 	=
 	\re \int_{\R^d} \ol{\phi_2(x)} V(x) \sum_{r=1}^k c_r\,e^{i\omega_r x} 
 		\, \d x < 0 \, .
 \end{align*}
 We set $\phi=\phi_{3,m}$ for some large enough $m$. It is clear from the construction that  $\hatt{\phi}$ is bounded and has compact support, so  it is in the domain of $T(p)$, and the supports of $\hatt{\varphi_{r,n}}$, $r=1,\ldots,k$,  and $\hatt{\phi}$ are pairwise disjoint for large enough $n$. 
 \smallskip
 
 Take $\phi$ as constructed above. We get from \eqref{eq:punchline2} that for all small enough $\alpha>0$ and large enough $n\in\N$ one has  
 \begin{align*}
 	\la \psi_{n,\alpha}, (T(p)+V) \psi_{n,\alpha}\ra <0 
 \end{align*}
   where $\psi_{n,\alpha}= \sum_{r=1}^k c_r\varphi_{r,n} + \alpha \phi$ and this bound holds for all  $c\in\{\C^k:\, |c|=1  \}$. By scaling, we then also have 
 \begin{align*}
 	\la \psi, (T(p)+V) \psi\ra <0 \quad \text{ for all } \psi\in N_{k,\alpha,n}
 \end{align*}
 and the usual min-max variational arguments show that $T(p)+V$ has no less than  
 $\dim(N_{k,\alpha,n})$ strictly negative eigenvalues.  
 
 It remains to show that $\dim(N_{k,\alpha,n})=k$, at least for large enough $n$. Given the somewhat complicated construction of $\phi$, this is now easy: 
 Let $(c^1, c^2, \ldots, c^{k-1})$ be an orthonormal basis for $V_{k-1}$ and extend it to an orthonormal  basis for $C^k$ by putting $c^k=a$, where $a$ is the normalized eigenvector of $M$ corresponding to the non-degenerate eigenvalue zero.  An arbitrary element of $N_{k,\alpha,n}$ can be written as 
 \begin{align*}
 	\sum_{s=1}^{k-1} d_s \big(\sum_{r=1}^k c^s_r \varphi_{r,n}\big) + d_k \sum_{r=1}^k c^{k}_r \varphi_{r,n} + \alpha d_k \phi 
 		= \sum_{r=1}^k (Ad)_r \varphi_{r,n} + \alpha d_k \phi 
 \end{align*}
 with the $k\times k$ unitary matrix $A= (c^1, c^2, \ldots, c^k)$, $Ad$ is the usual matrix vector product and $(Ad)_r$ is the $r^\text{th}$ entry of the vector $Ad$. Now assume that $\sum_{s=1}^{k-1} d_s \big(\sum_{r=1}^k c^s_r \varphi_{r,n}\big) + d_k \sum_{r=1}^k c^{k}_r \varphi_{r,n} + \alpha d_k \phi =0$. Taking the Fourier transform, this gives 
 \begin{align*}
 	\sum_{r=1}^k (Ad)_r \hatt{\varphi}_{r,n}(\eta) + \alpha d_k \hatt{\phi}(\eta) = 0 \quad \text{ for all }\eta\in\R^d. 
 \end{align*}
 By construction, the supports of  $\hatt{\phi}$ and $\hatt{\varphi}_{r,n}$, $r=1,\ldots, k$, are \emph{pairwise disjoint} for all large enough $n$, in particular, they are linearly independent functions. So $(Ad)_r =0$ for $r=1,\ldots,k$, which implies that $d$ is the zero vector in $\C^k$, since $A$ is unitary, hence invertible. Thus $N_{k,\alpha,n}$ is the span of $k$ linearly independent vectors, that is,  $\dim(N_{k,\alpha,n})=k$, for all large enough $n$. 
 This finishes the proof.
\end{proof}
 
Now we come to the 
\begin{proof}[Proof of Corollary \ref{cor:isolated-points}] 
 A simple calculation shows that the assumption of Theorem \ref{thm:many} is fulfilled at the $k$ distinct points $\omega_1,\ldots,\omega_k$. So Theorem \ref{thm:many} applies.
\end{proof}

For the proof of Corollary \ref{cor:submanifold} the following Lemma is helpful, 
which gives the so--called nearest point projection parametrization of a suitable open neighborhood of $\Sigma$.

\begin{lemma}\label{lem:transform} Suppose $\Sigma$ is a $\calC^2$ submanifold of 
 $\R^d$ with codimension $1\le n\le d-1$. 
 Then for each point $\omega\in \Sigma$, there exists a neighborhood $\calO$ of 
 $\omega$ in $\R^d$ and neighborhoods $\calU_1$ in $\R^{d-n}$ and  $\calU_2$ in $\R^{n}$ both containing zero and a $\calC^1$ diffeomorphism  
 $\Psi: \calU_1 \times \calU_2 \to \calO$, such that 
 \begin{align*}
  \Psi(0, 0) = \omega \text{ and } \Psi(y, 0) \in \Sigma \text{ for all } y\in \calU_1 \, .
 \end{align*}
 Moreover,
 \begin{align*}
  \dist(\Psi(y, t),\Sigma) = t.
 \end{align*}
\end{lemma}
This result is well--known to geometers and can be found, in the analytic category, for example in \cite{LSimon}. For convenience of the reader, we give the 
proof of this Lemma in the appendix.
We now come to the 
\begin{proof}[Proof of Corollary \ref{cor:submanifold}: ] 
  Assume that $\Sigma$ has codimension $1\le n\le d-1$. 
  Pick a point $\omega\in\Sigma$ and let $\calO$, $\calU_1$, $\calU_2$ be the neighboorhoods and $\psi$ the $\calC^1$ diffeomorphism from  Lemma \ref{lem:transform}.
  Since $\calO$ is open there exists $\delta_0>0$ such that $B_\delta(\omega)\subset\calO$ for all $0<\delta\le\delta_0$. Fix such a $\delta$ and choose $A_1\subset \calU_1$  
  and $A_2\subset \calU_2$ both centered closed balls in $\R^{d-n}$, respectively, $\R^n$, 
  with
  \begin{align*}
    \psi(A_1\times A_2)\subset B_\delta(\omega). 	
  \end{align*}
  We use $\psi$ to change coordinates in the calculation of a lower bound for 
  $\int_{B_\delta(\omega)}T(\eta)^{-1}\, \d\eta $. Parametrize $\eta$ as 
  $\eta=\psi(y,t)$, then the change of variables formula gives 
  \begin{align*}
  	\int_{B_\delta(\omega)}T(\eta)^{-1}\, \d\eta 
  	&\ge \int_{\psi(A_1\times A_2)}T(\eta)^{-1}\, \d\eta 
  		= \iint_{A_1\times A_2} T(\psi(y,t)) \, |\det(D(\psi(y,t)))|\,\d y \d t\\
  	&\gtrsim \iint_{A_1\times A_2} |t|^{-\gamma} \,\d y \d t
  		\sim  \int_{A_1} \Big( \int_0^{\diam(A_2)} r^{-\gamma+d-1}\, \d r \Big)\, \d y
  		= \infty \int_{A_1}\, \d y = \infty \, ,
  \end{align*} 
  where in the second inequality we used the assumption on the symbol $T$, and the fact that $D\psi$ is continuous, so $|\det(D\psi(y,t))|\gtrsim 1$ on the compact set 
  $A_1\times A_2$. In the last steps we simply used $\gamma\ge n$.  
  Together with the $k=1$ case of part \ref{thm:many-c} of Theorem \ref{thm:many} this shows that $T(p)+V$ has at least one strictly negative eigenvalue if V is relatively form compact with respect to $T(p)$, $V\in L^1(\R^d)$, and $\int V\, dx \le 0$.

  Of course, we can pick arbitrarily many distinct points $\omega_l\in\Sigma$ and then the above shows that for arbitrarily many distinct points $\omega_l\in\Sigma$ one has 
  \begin{align*}
    	\int_{B_\delta(\omega_l)}T(\eta)^{-1}\, \d\eta =\infty 
  \end{align*}
  for all small enough $\delta$, hence by monotonicity also for all $\delta>0$. Thus  if $V\neq 0$ and $V\le 0$, the assumption of part \ref{thm:many-a} of Theorem \ref{thm:many} is fulfilled for any $k\in\N$ and so  
  $T(p)+V$ has infinitely many strictly negative bound states in this case. 
\end{proof}

Finally, we will prove Theorem \ref{thm:superduper1}, by reducing it to the $k=1$ case of Theorem \ref{thm:many}.\ref{thm:many-c}. For this, the following Lemma is useful.
\begin{lemma}\label{lem:reduction}
	Assume that $T:\R^d\to [0,\infty)$ is measurable and that there exists a 
	compact set $M\subset \R^d$ such that \eqref{eq:superduper1} holds. 
	Then there exists a point $\omega\in M$  such that $T$ has a thick zero set near $\omega$. 
\end{lemma}
\begin{proof}
  By assumption we know that there exist a compact subset $M\subset \R^d$ with  
  \begin{align*}
    \int_{M_\delta} T(\eta)^{-1}\, d\eta = \infty	
  \end{align*}
  for all $\delta>0$, where $M_\delta$ is the closed $\delta$--neighborhood 
  $
    M_\delta = \{\eta\in\R^d:\, \dist(\eta,M)\le \delta\}. 	
  $
  
  Assume, by contradiction, that for every $\omega\in M$ 
  there exists $\delta_\omega>0$ with 
  \begin{align*}
    \int_{B_{\delta_\omega}(\omega)} T(\eta)^{-1}\, d\eta <\infty. 	
  \end{align*}
  We clearly have 
  \begin{align*}
    M\subset \bigcup_{\omega\in M} B_{\delta_\omega}(\omega)	
  \end{align*}
  and by compactness of $M$ there exist a finite subcover, i.e., $N\in\N$ and points 
  $\omega_l\in M $, $l=1,\ldots,N$, such that 
  \begin{align*}
    \calO\coloneq \bigcup_{l=1}^N B_{\delta_{\omega_l}}(\omega_l) \supset M \, . 
  \end{align*}
  Clearly 
  \begin{align}\label{eq:contradiction 1}
    \int_{\calO} T(\eta)^{-1}	\, d\eta <\infty 
  \end{align}
  by construction of $\calO$. Since $M$ is compact and contained in the open set $\calO$, 
  it has a strictly positive distance from the closed set  $\calO^c$. Thus there exists $\delta>0$ such that 
  $M_\delta\subset \calO$, but then with \eqref{eq:contradiction 1} we arrive at the contradiction 
  \begin{align*}
    \infty>\int_{\calO} T(\eta)^{-1}	\, d\eta 
    \ge \int_{M_\delta} T(\eta)^{-1}	\, d\eta = \infty\, .
  \end{align*}
   Hence there exists $\omega\in M$ for which \eqref{eq:thick near omega} holds.
\end{proof}
Now we can give the short 
\begin{proof}[Proof of Theorem \ref{thm:superduper1}:] 
  From the assumption of the Theorem and Lemma \ref{lem:reduction} we have that there exists a point 
  $\omega\in\R^d$ such that $T$ has a thick zero set near $\omega$ and hence the $k=1$ case  of Theorem \ref{thm:many}.\ref{thm:many-c} applies. 
\end{proof}

\section{Quantitative bounds}\label{sec:cwikel}
\subsection{Proof of Theorem \ref{thm:superduper2}}
\label{sec:proofof-thm-superdooper2}
Our approach is inspired by Cwikel's proof of the Cwikel--Lieb--Rosenblum inequality. 
Since $T(p)+V\ge T(p)-V_-$, in the sense of quadratic forms,  the variational principle shows 
\begin{align*}
	N(T(p)+V)\le N(T(p)-V_-)
\end{align*}
where $N(A)$ denotes the number of negative eigenvalues of an operator $A$. Thus it is enough to bound the number of strictly negative eigenvalues of $T(p)-U$, where $U\ge 0$. 

Since $U$ is relatively form compact with respect to $T(p)$, the operator $\sqrt{U}(T(p)+E)^{-1/2}$ is compact for all $E>0$, see \cite[Lemma 6.28]{Teschl}. 
Let $A$ be a compact operator with singular values $s_j(A)$, $j\in\N$, and let 
\begin{align*}
   	n(A;1)\coloneq \#\{j\in\N:\, s_j(A)\ge 1\}
\end{align*}
be the number of singular values of $A$ greater or equal to one. Furthermore, for $E>0$ let $N(T(p)-U,-E)$ be the number of eigenvalues of $T(p)-U$ which are $\le -E$. 
The Birman--Schwinger principle, \cite[Theorem 7.9.4]{Simon-course2}, shows 
\begin{align}\label{eq:BS-principle}
	N(T(p)-U,-E) = n(K_E;1) 
\end{align}
with the so-called Birman-Schwinger operator $K_E= \sqrt{U}(T(p)+E)^{-1}\sqrt{U}$, which is also a compact operator for any $E>0$. 
Factorizing $K_E= A_E A_E^*$ with $A= f(x)g_E(p)$, where we introduced the multiplication operator $f=\sqrt{U}$ and the Fourier--multiplier 
$g_E(p)= (T(p)+E)^{-1/2}$, the Birman-Schwinger principle shows 
\begin{align*}
	N(T(p)-U,-E) = n(A_E;1) 
\end{align*}
since the singular values of $A_E$ are just the square roots of the positive eigenvalues of $K_E$. Since $N(T(p)-U)=\lim_{E\to 0+}N(T(p)-U,-E)$, we have to control $n(A_E;1)$ for small $E>0$. For convenience, we will write $g$ for $g_E$ below.    
Following Cwikel, we decompose $f$ and $g$ as 
\begin{align*}
	f&=\sum_{n\in\Z} f_n \quad \text{and } g= \sum_{n\in\Z} g_n   \, , \\
	\text{ where }f_n&\coloneq f\id_{\{\alpha r^{n-1}<f\le \alpha r^n\}}  \text{ and }
	g_n\coloneq g\id_{\{r^{n-1}<g\le r^n\}}
\end{align*}
for some $\alpha>0$ and $r>1$ and introduce the operators 
\begin{align*}
	B_{\alpha,r} \coloneq \sum_{k+l\le 1} f_k(x) g_l(p), \quad 
	H_{\alpha,r} \coloneq \sum_{k+l\ge 2} f_k(x) g_l(p)\, .
\end{align*}

We have the bounds
\begin{lemma}\label{lem:cwikel-type-a-priori-bounds} 
For any $\alpha >0$ and $r>1$ and any functions $f,g\ge 0$ the operator $B_\alpha$ is 
bounded and its operator norm is bounded by 
\begin{align*}
	\|B_{\alpha,r}\|&\le \alpha\frac{r^2}{r-1} .
\end{align*}
Moreover, define $\wti{G}_\alpha(u)$ for $u,\alpha>0$ by  
\begin{align*}
	\wti{G}_\alpha(u)\coloneq u^2\int\limits_{ug(\eta)>\alpha} g(\eta)^{2}\,\frac{\d\eta}{(2\pi)^d}.  
\end{align*}
If for $\alpha>0$ we have $\int_{\R^d} G_\alpha(f(x))\, \d x<\infty$, then  $H_{\alpha,r}$ is a Hilbert-Schmidt 
operator for all $ r>1$ and its Hilbert-Schmidt norm is bounded by    
\begin{align*}
	\|H_{\alpha,r}\|_{HS}^2 \le \int_{\R^d} \wti{G}_\alpha(f(x))\, dx\, .
\end{align*}
\end{lemma}
\begin{remarks} 
	\itemthm If $g$ is `locally' $L^2$ in the sense that $g\id_{\{g>\alpha\}}\in L^2(\R^d)$ for any $\alpha>0$, then $\wti{G}_\alpha(u)<\infty$ for all $u,\alpha>0$ and $\lim_{u\to 0} \wti{G}_\alpha(u)=0$ for any $\alpha>0$. \\
	\itemthm Note that the right hand side of the bound on the operator norm of $B_{\alpha,r}$ is minimized by the choice $r=2$ and the bound for the Hilbert-Schmidt norm of $H_{\alpha,r}$ is \emph{independent }  of $r>1$. So we can and will use the choice $r=2$ later. 
\end{remarks}
Before we give the proof of the lemma, we state and prove an immediate consequence.
\begin{corollary}\label{cor:compactness}
	If $\int_{\R^d} \wti{G}_\alpha(f(x))\, \d x<\infty$ for all $\alpha>0$, then the operator $f(x)g(p)$ is compact. 
\end{corollary}
\begin{proof}
	By Lemma \ref{lem:cwikel-type-a-priori-bounds} we have  $f(x)g(p)= B_\alpha+H_\alpha$, where $B_\alpha$ is bounded and $H_\alpha$ is a Hilbert-Schmidt, in particular, a compact operator. 
	Since the operator norm of $B_\alpha$ is bounded by $\|B_\alpha\|\le 4\alpha$, where we chose $r=2$ for convenience, we see that $f(x)g(p)$ is the norm limit, as $\alpha\to 0$, of the compact operators $H_\alpha$, so it must be compact. 
\end{proof}

\begin{proof}[Proof of Lemma \ref{lem:cwikel-type-a-priori-bounds}: ]
The proof of the bound for the operator norm of $B_{\alpha,r}$ follows Cwikel's ideas closely, with the difference that we defined $B_{\alpha,r}$ slightly differently\footnote{This slight change in definition is the reason why the bound on the Hilbert-Schmidt norm of $H_{\alpha,r}$ is \emph{independent} of $r>1$, so we can freely optimize the bound on $B_{\alpha,r}$ from Lemma \ref{lem:cwikel-type-a-priori-bounds} in $r>1$. This was already noticed in \cite{Hundertmark}. } than Cwikel in \cite{Cwikel}. We give the  short proof for the convenience of the reader. Let $\Psi,\Phi\in L^2(\R^d)$. Then
\begin{align*}
	|\la \Psi, B_{\alpha,r}\Phi \ra| 
	&\le \sum_{k+l\le 1} |\la f_k(x)\psi, g_l(p)\Phi \ra|
		\le \alpha\sum_{k+l\le 1} r^{k+l}\|(\alpha r^k)^{-1} f_k\Psi\| \|r^{-l} g_l\hatt{\Phi} \| \\
	&\le \alpha \left(\sum_{k+l\le 1} r^{k+l}\|(\alpha r^k)^{-1} f_k\Psi\|^2\right)^{1/2}  
			\left(\sum_{k+l\le 1} r^{k+l} \|r^{-l} g_l\hatt{\Phi} \|^2\right)^{1/2} \, .
\end{align*}
Moreover,  the $f_k$ have disjoint supports, so $\sum_{k\in\Z}(\alpha r^k)^{-2} f_k^2\le 1$ pointwise, 
hence
\begin{align*}
	\sum_{k+l\le 1} r^{k+l}\|(\alpha r^k)^{-1} f_k\psi\|^2
	= \sum_{n\le 1} r^{n} \la\Psi, \sum_{k\in\Z}(\alpha r^k)^{-2} f_k^2 \Psi \ra 
	\le \frac{r^2}{r-1} \|\Psi\|^2. 
\end{align*}
For the same reason 
\begin{align*}
	\sum_{k+l\le 1} r^{k+l} \|r^{-l} g_l\hatt{\Phi} \|^2 
	\le  \frac{r^2}{r-1} \|\Phi\|^2,
\end{align*}
so we get the bound 
\begin{align*}
	\|B_{\alpha,r}\| \le \alpha\frac{r^2}{r-1}
\end{align*}
for the operator norm of $B_{\alpha,r}$. 

The bound of the Hilbert-Schmidt norm of $H_\alpha$ is a simple calculation. 
It is convenient to consider the operator
$\wti{H}_{\alpha}= f(x)\calF^{-1}g(\eta)$, since $H_\alpha^*H_\alpha$ is unitarily equivalent to $\wti{H}_{\alpha}^*\wti{H}_{\alpha}$, so their Hilbert-Schmidt norms are the same.  The advantage is that one can easily read off the kernel of $\wti{H}_{\alpha}$, for which we have the bound
\begin{align*}
	|\wti{H}_{t}(x,\eta)|&
		\le (2\pi)^{-d/2} \sum_{k+l\ge 2} f_k(x) g_l(\eta) 
		= (2\pi)^{-d/2} \sum_{k+l\ge 2} f_k(x) g_l(\eta)\id_{\{f(x)g(\eta)>\alpha\}} \\
		&\le (2\pi)^{-d/2}f(x) g(\eta)\id_{\{f(x)g(\eta)>\alpha\}},
\end{align*}
since the supports of $f_k$, $g_l$, respectively, are pairwise disjoint and for 
$(x,\eta)$ in the support of $f_k g_l$  we have 
$ f(x) g(\eta) = f_k(x)g_l(\eta)> tr^{k+l-2}\ge t $ by construction of $f_k$ and $g_l$ 
and since $k+l\ge 2$ in the above sum. Thus with Tonelli's theorem one sees 
\begin{align*}
	\|H_\alpha\|_{HS}^2 &= \|\wti{H}_\alpha\|_{HS}^2 
		= \iint |\wti{H}_\alpha(x,\eta)|^2\, \d x\d\eta  
		\le  (2\pi)^{-d}\iint\limits_{f(x)g(\eta)>\alpha}  f(x)^2g(\eta)^2\, \d\eta\d x \\
	&= \int f(x)^2 \int\limits_{f(x)g(\eta)>\alpha}g(\eta)^2\,\frac{\d\eta \d x}{(2\pi)^d} 
		=\int_{\R^d} \wti{G}_\alpha(f(x))\, \d x 
\end{align*}
with 
\begin{align*}
	\wti{G}_\alpha(u)=   u^2\int\limits_{ug(\eta)>\alpha} g(\eta)^2\,\frac{\d\eta}{(2\pi)^d}. 
\end{align*}
as claimed. 
\end{proof}
Now we come to the 
\begin{proof}[Proof of Theorem \ref{thm:superduper2}: ]  
  The usual arguments, see \cite[Lemma 6.26]{Teschl} or \cite[Theorem 7.8.3]{Simon-course2}, show that the essential spectrum does not change, 
  $\sigma_\text{ess}(T(p)+V)=\sigma_\text{ess}(T(p))$, when 
  $V$ is a relatively form compact perturbation of $T(p)$.  
  That $\sigma_\ess(T(p))=\sigma(T(p))\subset [0,\infty)$ is clear, since $T(p)$ is a Fourier multiplier with a positive symbol $T$.   
  
   It remains to prove the bound \eqref{eq:superduper2}: As already discussed in the beginning of this section, setting $f=V_-^{1/2}$ and $g=g_E=(T+E)^{-1/2}$, the Birman Schwinger principle and the variational theorem yield
  \begin{equation}\label{eq:punchline1}
  \begin{split}
  	N(T(p)+V, -E)&= n(\sqrt{V_-}(T(p)+E)^{-1/2};1) = \#\{n:\, s_n(f(x)g(p))\ge 1\} \\
  	&\le \sum_{n\in\N} \frac{(s_n(f(x)g(p))-\mu)_+^2}{(1-\mu)^2} 
  		= \sum_{n\in\N} \frac{(s_n(B_\alpha+H_\alpha)-\mu)_+^2}{(1-\mu)^2} 
  \end{split}
  \end{equation}
  for any $0\le \mu<1$, where the inequality follows from the simple bound 
  $(s-\mu)_+^2/(1-\mu)^2\ge 1$ for all $s\ge 1$ and where we split $ f(x)g(p) = B_\alpha+H_\alpha$, with the optimal choice of $r=2$. 
  
  Ky--Fan's inequality for the singular values and the first part of Lemma \ref{lem:cwikel-type-a-priori-bounds} gives 
  \begin{align*}
  	s_n(B_\alpha+H_\alpha)\le s_1(B_\alpha)+s_n(H_\alpha)= \|B_\alpha\|+s_n(H_\alpha)\le 4\alpha  +s_n(H_\alpha)\, .
 \end{align*}
 So choosing $\mu=4\alpha$ in \eqref{eq:punchline1}, we arrive at  the bound 
 \begin{align*}
 	N(T(p)+V, -E)\le (1-4\alpha)^{-2} \|H_\alpha\|_{HS}^2 \le (1-4\alpha)^{-2} \int_{\R^d} \wti{G}_\alpha(f(x))\, dx\, ,
 \end{align*}
  for all $0<\alpha<1/4$. Since $g=g_E=(T+E)^{-1/2}$ and $f=\sqrt{V_-}$ a straightforward calculation and a simple monotonicity argument shows 
  \begin{align*}
  	\wti{G}_\alpha(u) = u^2\int_{T+E<u^2/\alpha^2} \frac{1}{T(\eta)+E}\frac{\d\eta}{(2\pi)^d}
  		\le u^2\int_{T<u^2/\alpha^2} \frac{1}{T(\eta)}\frac{\d\eta}{(2\pi)^d}
  		= \alpha^2 G(\frac{u^2}{\alpha^2})
  \end{align*}
  with $G$ from \eqref{eq:G}. 
  So, since $f(x)=\sqrt{V_-(x)}$, we have 
  \begin{align}
  	N(T(p)+V,-E)\le \frac{\alpha^2}{(1-4\alpha)^2} \int_{\R^d} G(V_-(x)/\alpha^2)\, \d x
  \end{align}
  and letting  $E\to 0$ finishes the proof. 
\end{proof}

\subsection{Proof of Theorem \ref{thm:both}}
\label{sec:proof-thm-both}
We start with 
\begin{lemma}\label{lem:equivalence}
	Under the conditions of Theorem \ref{thm:both} we have 
	\begin{align*}
		\int_{Z_\delta}\frac{1}{T}\, \d\eta =\infty \text{ for some }\delta>0
		&\Longrightarrow 
		\int_{T<u} \frac{1}{T}\, \d\eta =\infty \text{ for all }u>0 \\
		&\Longrightarrow 
		\int_{Z_\delta} \frac{1}{T}\, \d\eta =\infty \text{ for all }\delta>0\, .
	\end{align*}
\end{lemma}
\begin{remark}
	Lemma \ref{lem:equivalence} clearly shows  
	\begin{align*}
		\int_{Z_\delta}\frac{1}{T}\, \d\eta =\infty \text{ for some }\delta>0
		&\Longleftrightarrow 
		\int_{Z_\delta}\frac{1}{T}\, \d\eta =\infty \text{ for all }\delta>0
	\intertext{and } 
		\int_{Z_\delta}\frac{1}{T}\, \d\eta <\infty \text{ for some }\delta>0
		&\Longleftrightarrow 
		\int_{T<u}\frac{1}{T}\, \d\eta =\infty \text{ for some } u>0\, ,
	\end{align*}
	which explains Remark \ref{rem:thm-both}.\ref{rem:equivalence}. 
\end{remark}

\begin{proof}[Proof of Lemma \ref{lem:equivalence}: ]
  Note the simple identity  
\begin{equation}
	\begin{split}\label{eq:basic-finiteness}
	\int_{T<u} \frac{1}{T}\, \d\eta 
		& = \int_{\{T<u\}\cap Z_\delta} \frac{1}{T} \, \d\eta 
			+ \int_{\{T< u\}\cap Z_\delta^c} \frac{1}{T} \, \d\eta \\
		& = \int_{Z_\delta} \frac{1}{T} \, \d\eta  
			- \int_{\{T\ge u\}\cap Z_\delta} \frac{1}{T} \, \d\eta 
			+ \int_{\{T< u\}\cap Z_\delta^c} \frac{1}{T} \, \d\eta 
	\end{split}
  \end{equation}
  where $\int_{\{T<u\}\cap Z_\delta^c} \frac{1}{T} \, \d\eta <\infty$ for all 
  $\delta>0$ and all small enough $u>0$,  depending on $\delta$, because of \eqref{eq:finiteness}. Also $\int_{\{T\ge u\}\cap Z_\delta^c} \frac{1}{T} \, \d\eta \le |Z_\delta|/u<\infty$ by assumption.  
  So the left hand side of \eqref{eq:basic-finiteness} is infinite for all small enough $u>0$ if for some $\delta>0$ we have  $\int_{Z_\delta}\frac{1}{T}\, \d\eta =\infty$. But by monotonicity, then also  $\int_{T<u} \frac{1}{T}\, \d\eta =\infty$ for all $u>0$, which proves the first implication in Lemma \ref{lem:equivalence}. 
  
  On the other hand, once $\int_{T<u} \frac{1}{T}\, \d\eta =\infty$ for all $u>0$, one sees from  \eqref{eq:basic-finiteness} that 
  $\int_{ Z_\delta} \frac{1}{T} \, \d\eta =\infty$ for any $\delta>0$, since the last two terms in \eqref{eq:basic-finiteness} are finite for all small enough $u>0$. 
 \end{proof}

Now we come to the 
\begin{proof}[Proof of Theorem \ref{thm:both}: ]
	For part \eqref{thm:both-a} we note that by Lemma \ref{lem:equivalence} one has 
	\begin{align*}
		\int_{Z_\delta}\frac{1}{T}\, \d\eta =\infty \text{ for some }\delta>0
		&\Longleftrightarrow 
		\int_{Z_\delta}\frac{1}{T}\, \d\eta =\infty \text{ for all }\delta>0
	\end{align*}
	so one can use Theorem \ref{thm:superduper1} to see that one weakly coupled bound states exist once \eqref{eq:divergence} holds. 
	
	On the other hand, assume that \eqref{eq:divergence} fails. Then, again by Lemma \ref{lem:equivalence}, we have $\int_{T<u} \frac{1}{T}\, \d\eta <\infty $ for all small enough $u>0$. Thus $G(u)$ defined in \eqref{eq:G} is finite for all small enough $ u>0$ and $\lim_{u\to 0+}G(u)=0$.  Then a simple argument, see Remark \ref{rem:superdooper2}.\ref{rem:spectrum-positive},  yields a strictly negative potential $V$ such that  $T(p)+V$ has no negative spectrum. Thus condition \eqref{eq:divergence} is equivalent to having weakly coupled bound states. 
	
	For part \ref{thm:both-b} we simply note that Lemma \ref{lem:equivalence} shows that $\int_{Z_\delta}\frac{1}{T}\d\eta<\infty$ for some $\delta>0$ implies $\int_{T<0}\frac{1}{T}\d\eta<\infty$ for all small enough $u>0$. Then Theorem \ref{thm:superduper2} shows that a quantitative bound on the number of strictly negative eigenvalues of $T(p)+V$ in the form \eqref{eq:superduper2} holds. 
	
	Conversely, assume that \eqref{eq:integrable} fails. Then Theorem \ref{thm:superduper1} applies. Thus weakly coupled bound states always exist for any non-trivial attractive potentials, hence no quantitative bound on the number of strictly negative bound states can exist.  	
\end{proof}

\appendix
\setcounter{section}{0}
\renewcommand{\thesection}{\Alph{section}}
\renewcommand{\theequation}{\thesection.\arabic{equation}}
\renewcommand{\thetheorem}{\thesection.\arabic{theorem}}
%
\section{Relative form compactness}\label{app:relative-form-compactness}
The following Lemma was used in Remark \ref{rmk:superduper1}: 
\begin{lemma}\label{lem:rel-compactness1}
  Let $V\in L^1(\R^d)$ and $T:\R^d\to [0,\infty)$ be measurable with $ \lim_{\eta \rightarrow \infty} T(\eta) = \infty. $ Furthermore, assume that the we have $\calQ(T(p)^{1-\veps})\subset\calQ(V)$ for some $0<\veps<1$.  Then $V$ is relatively form-compact with respect to $T(p)$.
\end{lemma}
For the second compactness criterium we need one more notation. Let 
\begin{align}
	G_1(u)=u\int_{T+1<u} \frac{1}{T(\eta)+1}\frac{\d\eta}{(2\pi)^d}
\end{align}
for $u\ge 0$. 
\begin{lemma}\label{lem:rel-compactness2}
	Assume that 
	\begin{align*}
		\int_{\R^d} G_1(s |V(x)|)\d x <\infty \text{ for all } s>0. 
	\end{align*}
	Then $V$ is relatively form compact with respect to  $T(p)$. 
\end{lemma}

\begin{proof}[Proof of Lemma \ref{lem:rel-compactness1}: ]
We need to show that  $|V|^{\frac{1}{2}}(T(p)+1)^{-1/2}$ is compact. Since we have $\calQ(T(p)^{1-\veps})\subset\calQ(V)$, the closed graph theorem shows that 
$|V|^{1/2}(T(p)+1)^{(1-\veps)/2}$ is a bounded operator. 
Let $\id_{\le L}=\id_{\{|\eta|\le L\}}$ be the characteristic function of the closed centered ball of radius $L$ in momentum space,  
$\id_{\le L}(p)$ the corresponding Fourier multiplier, and 
$\id_{>L}(p) = \id - \id_{\leq L}(p) $.  
First observe that 
\begin{equation*}
A_L= |V|^{\frac{1}{2}} (T(p)+1)^{-\frac{1}{2}} \id_{\le L}(p)
\end{equation*}
is compact, in fact, it is a Hilbert-Schmidt operator. To see this, it is enough to show that 
$B_L= |V|^{\frac{1}{2}} \calF^{-1} (T(\eta) +1 )^{-\frac{1}{2}}$ is a Hilbert-Schmidt operator, since $A_L^*A_L$ is unitarily equivalent to $B_L^*B_L$. The operator $B_L$ has the kernel 
\begin{equation*}
B_L(x,\eta)= (2\pi)^{-\frac{d}{2}}|V(x)|^{\frac{1}{2}} e^{i x\cdot \eta} (T(\eta)+1)^{-\frac{1}{2}} \id_{\le L}(\eta)
\end{equation*}
and from $V\in L^1(\R^d)$ it follows that the kernel of $B_L$ is square-integrable with respect to $(x, \eta)\in\R^d\times\R^d$. 
Since 
	\begin{align*}
		\Vert V^{\frac{1}{2}}(T(p) + &1)^{-\frac{1}{2}} - V^{\frac{1}{2}}(T(p) + 1)^{-\frac{1}{2}} 				\id_{\le L}(p) \Vert  \\
		&= \Vert V^{\frac{1}{2}}(T(p) + 1)^{-\frac{1}{2}(1 - \varepsilon)}(T(p) + 1)^{-\frac{1}					{2}\varepsilon} \id_{\geq L}(p) \Vert \\
		&\leq \Vert V^{\frac{1}{2}}(T(p) + 1)^{-\frac{1}{2}(1 - \varepsilon)} \Vert \sup_{|\eta| \geq L}			(T(\eta) + 1)^{-\frac{1}{2}\varepsilon} \rightarrow 0 
	\end{align*}
as $L\to\infty$ since  $T(p)\to \infty$ as $|p|\to \infty$, whereas $V^{\frac{1}{2}}(T(p) + 1)^{-\frac{1}{2}(1 - \varepsilon)}$ is bounded by assumption.
Thus $|V|^{\frac{1}{2}}(T(p)+1)^{-1/2}$ is the norm limit of compact operators, hence compact.
\end{proof}

\begin{proof}[Proof of Lemma \ref{lem:rel-compactness2}: ]
  Note that Corollary \ref{cor:compactness} applies with the choice $f=|V|^{1/2}$ 
  and $g=(T+1)^{-1/2}$ since then 
  $ G_1(su) = s^{-2}\wti{G}_\alpha(s^2\sqrt{u})$.
\end{proof}

\section{Nearest point projection}\label{app:nearest-point-projection} 
In this section we give a sketch of the proof of Lemma \ref{lem:transform} for completeness. We follow the presentation of \cite{LSimon} where it was done in the analytic setting.

\begin{proof}[Proof of Lemma \ref{lem:transform}: ]
We assume that $M$ is a $C^2$ submanifold of codimension $n$ embedded in $\R^d$. 
At each point $\omega_0\in \Sigma$, there exists an open set $\omega\in \calO\subset\R^d$ and a chart, i.e., an open set $0\in \calU_1\subset\R^{d-n}$ and a $\calC^2$ map 
$\Phi:\calU_1\to \Sigma$ such that every point $\omega\in \Sigma\cap\calO$ can be uniquely written as $\omega=\Phi(y)$ for $y\in\calU_1$. Without loss of generality we can always assume $\omega_0=\Phi(0)$. 

The vectors $\{\partial_j \phi(y)\}_{j=1}^{d-n}$ form a
basis of the tangent space of $M$ at $\phi(y)$. Using the Gram-Schmidt orthogonalization, we can find $n$ additional vectors $\nu_1(y), \ldots, \nu_n(y)\in\R^d$ such that the vectors 
\begin{equation*}
\partial_1 \phi(y), \ldots, \partial_{d-n} \phi(y), \nu_1(y), \ldots, \nu_n(y)
\end{equation*}
form a basis of $\R^d$ at the point $\phi(y)$. Additionally, since $\Phi$ is $\calC^2$, we have that the above basis vectors of $\R^d$ depend continuously differentiable on $y\in \calU_1$. 
 
Now define a map on $\calU_1\times \R^{n}$ by
\begin{align*}
\psi(y, t) \coloneq \Phi(y) + \sum_{j=1}^{n} t_j \nu_j(y)\, .
\end{align*}
A computation shows $D \psi(0, 0) = I_{d\times d}$, so by the
inverse function theorem, $\psi$ is a $\calC^1$-diffeomorphism on 
a suitable neighborhood $\calU_1\times \calU_2$ of $(0, 0)\in \R^{d-n}\times \R^n$.

For all $\eta$ in a small enough $\delta$ neighborhood of $\Sigma\cap\calO$,  
the problem of minimizing 
$$
|\phi(y) - \eta|^2 
$$
over $y\in \calU_1$ has a unique solution for which we must have that all partial derivatives of $|\phi(y) - \eta|^2$ vanish, i.e., 
\begin{equation*}
	0= \partial_j \sum_{l=1}^d (\Phi_l(y)-\eta_l)^2 
		= 2\sum_{l=1}^d \partial_j\Phi_l(y)(\Phi_l(y)-\eta_l) 
		= 2 \partial_j\Phi(y)\cdot (\Phi(y)-\eta). 
\end{equation*}
Thus for $y$ minimizing the distance $|\Phi(y)-\eta|$, the vector $\Phi(y)-\eta$ is perpendicular to the tangent plane of $\Sigma$ at $\omega=\Phi(y)$.  
It follows that $\phi(y)-\eta$ can be written as a linear combination
$$
\Phi(y)-\eta = \sum_{j=1}^n t_j \nu_j(y), 
$$
or 
\begin{align*}
  \eta = \Phi(y)+ \sum_{j=1}^n t_j \nu_j(y)= \psi(y,t). 
\end{align*}
So locally the function $\psi$ yields a  parametrization of a neighboorhood of $\Sigma$ in $\R^d$ with the property that for $\eta=\psi(y,t)$ we have 
$\dist(\eta,\Sigma)= \dist(\psi(y, t), \Sigma) = t$. 
\end{proof}

\section{The classical phase volume $|\{T_{P,M,\mu_\pm} <u\}|$}
\label{app:calculation}
For the convenience of the reader, we sketch the calculation of $|\{T_{P,M,\mu_\pm}<u\}|$ from \cite{VugalterWeidl} in this section. 
Using a suitable rotation,  we can assume $P=|P|e_1$, with $e_1= (1,0,\ldots,0)^t$, the standard first unit vector in $\R^d$, and by scaling one has 
 \begin{align*}
	|\{T_{P,M,\mu_\pm}(\eta)<u\}| = |P| |\{\zeta\in\R^d: \,T_{\tau,g,\mu_\pm}(\zeta)<u/|P|\}|
 \end{align*}
 with $\tau=M/|P|$, $g=\sqrt{1+\tau^2}$, and 
 \begin{align*}
 	T_{\tau,g,\mu_\pm}(\zeta) = \sqrt{|\mu_+e_1-\zeta|^2 +\mu_+^2 \tau^2}  +\sqrt{|\mu_-e_1+\zeta|^2 +\mu_-^2 \tau^2} \, - g \, .
 \end{align*}
 Split $\zeta = (\vartheta,\xi)\in \R\times \R^{d-1}$ and put 
 $	T_\pm = \sqrt{|\xi|^2 + |\mu_\pm \mp \vartheta|^2 + \mu_\pm^2\tau^2}$. 
 Then $ T_{\tau,g,\mu_\pm}(\zeta)<s$ is equivalent to $T_+ + T_-< s+g$, that is, 
 \begin{equation}
  \begin{split}\label{eq:equivalence1}
 	2T_+ T_- &< (s+g)^2 -T_+^2 - T_-^2 \\
 			&= (s+g)^2 
 				- (2|\xi|^2 + |\mu_+-\vartheta|^2+|\mu_-+\vartheta| +(\mu_+^2+\mu_-^2)\tau^2) \, .
  \end{split}
 \end{equation}
 Define $\wti{\mu}= \tfrac{1}{2}(\mu_--\mu_+)$. Clearly $4\wti{\mu}^2=(\mu_--\mu_+)^2$ and since $\mu_++\mu_-=1$ we also have $1=(\mu_++\mu_-)^2$. 
 Thus 
 \begin{align*}
 	\mu_+^2+\mu_-^2 = 2(\frac{1}{4}+\wti{\mu}^2). 
 \end{align*}
 Set $\vartheta= \varphi -\wti{\mu}$ and $A=s+g$. Then $\mu_\pm \mp\vartheta = \frac{1}{2}\mp\varphi$, so \eqref{eq:equivalence1} is equivalent to 
 \begin{align}\label{eq:equivalence2}
 	T_+ T_- < \frac{A^2}{2} - \Big(|\xi|^2 +\varphi^2 +\frac{1}{4} +\big(\frac{1}{4}+\wti{\mu}^2\big)\tau^2 \Big) \, , 
 \end{align}
 in particular, the right hand side of \eqref{eq:equivalence2} is positive. 
 Note that 
 \begin{align*}
 	\mu_+^2 &= \frac{\mu_+^2}{2} + \frac{\mu_-^2}{2} + \frac{\mu_+^2}{2} - \frac{\mu_-^2}{2} 
 				= \frac{\mu_+^2 + \mu_-^2}{2} + \frac{(\mu_++\mu_-)(\mu_+-\mu_-)}{2} 
 			= \frac{1}{4} +\wti{\mu}^2 -\wti{\mu}	
 \end{align*}
 and similarly 
 $	\mu_-^2 = \frac{1}{4} +\wti{\mu}^2 +\wti{\mu}	$. Thus $T_+^2 T_-^2$ equals 
 \begin{align*}
 		 \Big( |\xi|^2 + &\varphi^2 +\frac{1}{4} + \big( \frac{1}{4} +\wti{\mu}^2 \big)\tau^2 -\big( \varphi+\wti{\mu} \tau^2\big) \Big) 
 			\Big( |\xi|^2 + \varphi^2 +\frac{1}{4} + \big( \frac{1}{4} +\wti{\mu}^2 \big)\tau^2 +\big( \varphi+\wti{\mu} \tau^2 \big) \Big) \\
 		&= \Big( |\xi|^2 + \varphi^2 +\frac{1}{4} + \big( \frac{1}{4} +\wti{\mu}^2 \big)\tau^2  \Big)^2 - \Big( \varphi+\wti{\mu} \tau^2 \Big)^2\, , 
 \end{align*}
 hence \eqref{eq:equivalence2} is equivalent to 
 \begin{align*}
 	A^2 \Big( |\xi|^2 + \varphi^2 +\frac{1}{4} + \big( \frac{1}{4} +\wti{\mu}^2 \big)\tau^2  \Big) -  \Big( \varphi+\wti{\mu} \tau^2 \Big)^2 
 		< \frac{A^4}{4} \, ,
 \end{align*}
 which in turn is equivalent to 
 \begin{align*}
 	|\xi|^2 &+\frac{A^2-1}{A^2} \Big( \varphi -\frac{\wti{\mu}\tau^2}{A^2-1} \Big)^2
 		< \frac{A^2}{4} - \Big( \frac{1}{4} + \big( \frac{1}{4} +\wti{\mu}^2 \big)\tau^2  \Big) + \frac{1}{A^2-1} \big( \wti{\mu}\tau^2 \big)^2 \\
 		&= \frac{(A^2-1)( A^2-1-(1+4\wti{\mu}^2\tau^2) ) +4\wti{\mu}^2\tau^4}{4(A^2-1)} 
 			= \frac{(A^2-1-\tau^2)( A^2-1-4\wti{\mu}^2\tau^2 )}{4(A^2-1)}\, .
 \end{align*}
 So the set where $ T_{\tau,g,\mu_\pm} (\zeta) <s $ is an ellipse with $d-1$ semiaxis of length  
 \begin{align*}
 	\frac{1}{2} \sqrt{ \frac{(A^2-1-\tau^2)( A^2-1-4\wti{\mu}^2\tau^2 )}{A^2-1} }
 \end{align*}
 and one semiaxis of length  
  \begin{align*}
 	\frac{1}{2}  \frac{A\sqrt{(A^2-1-\tau^2)( A^2-1-4\wti{\mu}^2\tau^2 )}}{(A^2-1)} \, .
 \end{align*} 
 Thus its  volume is given by 
 \begin{align*}
 	|\{ T_{\tau,g,\mu_\pm} <s  \}| 
 	&= \frac{|B_d|}{2^d} \frac{A(A^2-1-\tau^2)^{d/2}( A^2-1-4\wti{\mu}^2\tau^2 )^{d/2}}{(A^2-1)^{(d+1)/2}} \\
 	&= \frac{|B_d|}{2^d} \frac{s^{d/2}(s+g)(s+2g)^{d/2}( s^2+2gs +(1-4\wti{\mu}^2)\tau^2 )^{d/2}}{(s^2+2gs +\tau^2)^{(d+1)/2}} \, ,
 \end{align*}
 since $A=s+g$ and $g=\sqrt{1+\tau^2}$. Rescaling this we end up with 
 \begin{align}\label{eq:volume-symbol-rel-pair}
 	|\{T_{P,M,\mu_\pm}(\eta)<u\}| = \frac{|B_d|u^{d/2}(u+\gpm)(u+2\gpm)^{d/2}( u^2+2\gpm u +(1-4\wti{\mu}^2)M^2 )^{d/2}}{2^d (u^2+2\gpm u +M^2)^{(d+1)/2}} 
 \end{align}
 with $\gpm = \sqrt{P^2+M^2}$. 
 \bigskip

\noindent
\textbf{Acknowledgements: }  
 We would like to thank David Damanik for useful remarks on a preliminary version of the paper and Jonas Hirsch for help with the manifolds.
  
 Vu Hoang thanks the Deutsche Forschungsgemeinschaft (DFG) for financial support under 
 grants HO 5156/1-1 and HO 5156/1-2 and the National Science Foundation (NSF) for financial support under grant DMS-1614797. 
 
 Johanna Richter was supported by the Deutsche Forschungsgemeinschaft (DFG) through the Research Training Group 1838: \textit{Spectral Theory and Dynamics of Quantum Systems}.
 
 Dirk Hundertmark and Semjon Vugalter acknowledge financial support by the Deutsche Forschungsgemeinschaft (DFG) through the Coordinated Research Center (CRC) 1173 \textit{Wave phenomena}. Dirk Hundertmark also thanks the Alfried Krupp von Bohlen und Halbach Foundation for financial support. 

\renewcommand{\thesection}{\arabic{chapter}.\arabic{section}}
\renewcommand{\theequation}{\arabic{chapter}.\arabic{section}.\arabic{equation}}
\renewcommand{\thetheorem}{\arabic{chapter}.\arabic{section}.\arabic{theorem}}

\def\cprime{$'$}


\begin{thebibliography}{100}
\small{

\bibitem{BakNewman}  J.~Bak, D.~J.~Newman: \textit{Complex Analysis}. 
	Third edition, Springer, 2010. 
	
\bibitem{BalinskyEvans} A.~Balinsky, W.~D.~Evans:  
	\textit{Spectral Analysis of Relativistic Operators}. 
	London: Imperial College Press,  2010.
	
\bibitem{BerezinShubin} F.~A.~Berezin, M.~A.~Shubin, 
	\textit{The Schr\"odinger Equation}. 
	Mathematics and Its Applications (Soviet Series) \textbf{66}, 
	Springer Netherlands, 1991. 
	

\bibitem{BirmanSolomyak} M.~Birman, M.~Z.~Solomyak: \textit{Spectral Theory of Self-Adjoint Operators in Hilbert Space}.  Translated from Russian. Mathematics and its Applications (Soviet Series). D. Reidel Publishing Co., Dordrecht, 1987, pp.\ xvi+301.

\bibitem{BrueningGeylerPankrashkin} J.~Br\"uning, V.~Geyler, K.~Pankrashkin: 
	\textit{On the discrete spectrum of spin-orbit Hamiltonians with singular interactions}. 
	Russian Journal of Mathematical Physics, 
 	\textbf{14} (2007), Issue 4, 423-–429. 
 	Preprint arXiv:0709.0213v2 

\bibitem{Buell} W.~F.~Buell, B.A.~Shadwick: \textit{Potentials and Bound States}.
American Journal of Physics \textbf{63}, 256 (1995).

\bibitem{Cwikel} M. Cwikel,
\emph{Weak Type Estimates for Singular Values and the Number of Bound States of Schr\"odinger Operators.}
Annals of Mathematics. {\bf 106} (1977), 93-100.


\bibitem{Fefferman-de-la-Llave}
	C.~Fefferman, R.~de la Llave: \textit{Relativistic stability of matter}. 
	I.\ Rev.\ Mat.\ Iberoamericana \textbf{2} (1986), no.\ 1-2, 119-–213.

\bibitem{FrankHainzlNabokoSeiringer} 
	R.~Frank, C.~Hainzl, S.~Naboko, R.~Seiringer: \textit{The Critical Temperature 
	for the BCS Equation at Weak Coupling}. 
	The Journal of Geometric Analysis, \textbf{17} (4) (2007).


\bibitem{HainzlHamzaSeiringerSolovej} C.~Hainzl, E.~Hamza, R.~Seiringer, J.~P.~Solovej: \textit{The BCS Functional for General
Pair Interactions}. Comm.\ Math.\ Phys.\ \textbf{281}, 349--367 (2008).

\bibitem{HainzlSeiringer-degenerate} C.~Hainzl, R.~Seiringer: \textit{Asymptotic behavior of eigenvalues of Schr\"odinger type operators with degenerate kinetic energy}. Math.\ Nachr.\ \textbf{283}, No.\ 3, 489--499 (2010).



\bibitem{HardenkopfSucher} G.~Hardekopf, J.~Sucher, 
	\textit{Critical coupling constant for relativistic equations and vacuum breakdown in quantum electrodynamics}. 
	Phys.\ Rev \textbf{A 31} (1985), 2020-–2029. 
	

\bibitem{Herbst} 
	I.~W.~Herbst, \textit{Spectral theory of the operator 
		$(p^2+m^2)^{1/2}+Ze^2/r$}. 
	Comm.\ Math.\ Phys.\ \textbf{53} (1977), no.\ 3, 285--294.


\bibitem{Hundertmark}
 	D.~Hundertmark: \textit{On the number of bound states for Schr\"odinger operators 
 	with operator--valued potentials}. Arkiv f\"or matematik \textbf{40} (2002), 73--87.
\bibitem{Hundertmark-Simonfest}   
	D.~Hundertmark: \textit{Bound state problems in Quantum Mechanics,}
  		Spectral theory and mathematical physics: a Festschrift in honor of
  		Barry Simon's 60th birthday,  463--496,
  		Proc.\ Sympos.\ Pure Math., \textbf{76}, Part 1,
  		Amer.\ Math.\ Soc., Providence, RI, 2007.
 
	
\bibitem{Landau-Lifshitz} L.D.~Landau and E.M.~Lifshitz: 
	\textit{Quantum Mechanics, Non-Relativistic Theory, Volume 3}. 
	Pergamon Press, London (1965).
	
\bibitem{LaptevSafranovWeidl} A.~Laptev, O.~Safranov, and T.~Weidl: 
	\textit{Bound State Asymptotics for elliptic operators with strongly degenerated symbols}. 
	In \textit{Nonlinear problems in Mathematical Physics and Related Topics I}, International Mathematical Series (New York) Vol 1, Kluwer/Plenum New York (2002), pp. 233-246.
	

\bibitem{LewisSiedentopVugalter}	
	R.T.~Lewis, H.~Siedentop, and S.~Vugalter: \textit{The essential spectrum of relativis- tic multi-particle operators}. Ann.~Inst.~H.~Poincar\'{e}~Phys.~Theor.~\textbf{67}  (1) (1997), 1–-28.

\bibitem{Lieb} E.~H.~Lieb,
\emph{The Number of Bound States of One-Body Schr\"odinger Operators and the Weyl Problem}. Springer, Berlin, Proc. A.M.S. Symp. Pure Math. {\bf 36} (1980), 241--252. 

\bibitem{LiebLoss} E.~H.~Lieb, M.~Loss: 
	\textit{Analysis}, AMS, Providence, Rhode Island (2001).
	
\bibitem{LiebSeiringer} E.~H.~Lieb, R.~Seiringer: 
	\textit{The Stability of Matter in Quantum Mechanics},
	Cambridge Univ.\ Press. 
	
\bibitem{LiebThirring} E.~H.~Lieb, W.~Thirring: 
	\textit{Inequalities for the moments of the eigenvalues of the 
		Schr\"odinger Hamiltonian and their relation to Sobolev inequalities}, 
		Studies in Math.\ Phys., Essays in Honor of Valentine Bargmann, Princeton (1976).
	
\bibitem{LiebYau} E.~H.~Lieb, H.-T.~Yau, 
	\textit{The Stability and Instability of Relativistic Matter}. 
	Comm.\ Math.\ Phys.\  \textbf{118} (1988), 177-–213.
	
\bibitem{MolchanovVainberg} S.~Molchanov, B.~Vainberg: \textit{On general Cwikel-Lieb-Rozenblum and Lieb-Thirring inequalities}. 
Preprint {\tt arXiv:0812.2968v4 [math-ph] 17 Aug 2012}

\bibitem{NetrusovWeidl-1996}
	Y.~Netrusov, T.~Weidl: \textit{On Lieb-Thirring inequalities for higher order operators with critical and subcritical powers}, Commun.\ Math.\ Phys.\, \textbf{182}, (1996) 355--370.

\bibitem{Pankrashkin}
  	K.\ Pankrashkin, \textit{ Variational principle for Hamiltonians with degenerate bottom.} 
  	In I.\ Beltita, G.\ Nenciu, R.\ Purice (Eds.): Mathematical results in quantum mechanics, 231-–240, World Sci.\ Publ., Hackensack, NJ, 2008. 
  	Preprint https://arxiv.org/abs/0710.4790
  	  	
\bibitem{ParzygnatJohnson} A.~Parzygnat, K.~K.~Y.~Lee, Y.~Avniel,
	S.~G.~Johnson: \textit{Sufficient conditions for two-dimensional localization 
	by arbitrarily weak defects in periodic potentials with band gaps.} 
	Phys.~Rev.~\textbf{B 81}, 155324 (2010).

\bibitem{ReedSimon} M.~Reed, B.~Simon: 
	\textit{Methods of Modern Mathematical Physics. IV. Operator Theory.} 
	Academic Press, San Diego (2005).
	
\bibitem{Rozenblum} G. V. Rozenblum,
\emph{Estimates of the spectrum of the Schr\"odinger operator.}
Jour. Sovj. Math. {\bf 10} (1978), 934-944.
	
\bibitem{SimonForms} B.~Simon: 
	\textit{Quantum Mechanics for Hamiltonians defined as Quadratic Forms}. 
	Princeton Series in Physics, Princeton, New Jersey (1971). 
	
\bibitem{Simon} B.~Simon: 
	\textit{The bound states of weakly coupled Schr\"odinger operators in one 
	and two dimensions}. 
	Annals of Physics \textbf{97} (2): 279--288 (1976). 
	
	
\bibitem{Simon-course1} B.~Simon: 
	\textit{Real Analysis. A Comprehensive Course in Analysis, Part 1}. 
	American Mathematical Society AMS, 2015.


\bibitem{Simon-course2} B.~Simon: 
	\textit{Operator Theory. A Comprehensive Course in Analysis, Part 4}. 
	American Mathematical Society AMS, 2015.

\bibitem{LSimon} L.~Simon: \textit{Theorems on Regularity and Singularity of 
	Energy-minimizing Maps}. 
	Birkh\"auser, Basel (1996).


\bibitem{Teschl} G.~Teschl: \textit{mathematical methods in Quantum Mechanics. With applications to Schr\"odinger Operators.} Second edition. Graduate Studies in Mathematics \textbf{157}, Aerican Mathematical Society, Provodence Rhode Island, 2014.  


\bibitem{VugalterWeidl}
 	S.~Vugalter, T.~Weidl: \textit{On the discrete spectrum of a pseudo-relativistic two-body pair operator}. 
 	Ann.\ Henri Poincar\'{e}  \textbf{4} (2003), no. 2, 301–-341.

\bibitem{Weidl1} T.~Weidl: \textit{Another look at Cwikel's inequality}, in Differential Operators and Spectral Theory. M.Sh.~Birman's 70th Anniversary Collection. AMS Translations Series 2 \textbf{189} (1999) 247--254.

\bibitem{Weidl-virtual-bound-states} T.~Weidl: \textit{Remarks on virtual bound states for semi-bounded operators}, Communications in Partial Differential Equations, \textbf{24} (1999), vol.\ 1-2, 25--60. 

\bibitem{Weidls2} T.~Weidl: \textit{Nonstandard Cwikel Type Estimates}.  Contemp.\ Math., \textbf{445} (2007), 
	Amer.\ Math.\ Soc., Providence, RI, 337--357.
	
	

\bibitem{Yang_deLlano} K.~Yang, M~ de Llano,: \textit{Simple variational proof that any two‐dimensional potential well supports at least one bound state}. 
	Am.\ J.\ Phys.\ \textbf{57}, 85 (1989).
}
\end{thebibliography}
\end{document}